%% file: FullQuantum---QuantumJournal---v2.tex
\documentclass[a4paper,onecolumn,11pt,accepted=2024-08-20]{quantumarticle}
\pdfoutput=1

\usepackage[utf8]{inputenc}
\usepackage{amsmath,amsthm}
\usepackage{stmaryrd}
\usepackage{tikz}
\usepackage{tikzit}
\usepackage{mathdots}
\usepackage{yhmath}
\usepackage{cancel}
\usepackage{color}
\usepackage{siunitx}
\usepackage{array}
\usepackage{multirow}
\usepackage{amssymb}
\usepackage{gensymb}
\usepackage{tabularx}
\usepackage{extarrows}
\usepackage{booktabs}
\usepackage{makecell}
\usepackage{xfrac}
\usepackage{comment}
\usepackage{graphbox}
\usepackage[numbers]{natbib}
\usepackage{hyperref}
\usepackage{todonotes}
\usepackage{braket}
\usepackage{mathtools}
\usepackage{skak}
\usepackage{enumitem}
\setlist{topsep=0pt,itemsep=0.2pt,partopsep=2pt, parsep=2pt}
\usetikzlibrary{fadings}
\usetikzlibrary{patterns}
\usetikzlibrary{shadows.blur}
\usetikzlibrary{shapes}

\DeclareRobustCommand{\harp}[1]{\mathpalette\harpoonvec{#1}}
\newcommand{\harpvecsign}{\scriptscriptstyle\rightharpoonup}
\newcommand{\harpoonvec}[2]{%
  \ifx\displaystyle#1\doalign{$\harpvecsign$}{#1#2}\fi
  \ifx\textstyle#1\doalign{$\harpvecsign$}{#1#2}\fi
  \ifx\scriptstyle#1\doalign{\scalebox{.6}[.9]{$\harpvecsign$}}{#1#2}\fi
  \ifx\scriptscriptstyle#1\doalign{\scalebox{.5}[.8]{$\harpvecsign$}}{#1#2}\fi
}
\newcommand{\doalign}[2]{%
 {\vbox{\offinterlineskip\ialign{\hfil##\hfil\cr#1\cr$#2$\cr}}}%
}
\newcommand{\corresponds}{\hat{=}}
\newcommand{\smwhtcircle}{\circ}
\newcommand{\smblkcircle}{\bullet}
\newcommand{\hyphenbullet}{\raisebox{-0.5pt}{\castlinghyphen}}
\newcommand{\harrowextender}{\chesssee}

\input{tikzstyles}

\input{styles.tikzstyles}
\usetikzlibrary{calc, positioning, shapes.geometric}
\usetikzlibrary{
	arrows,
	shapes,
	decorations,
	intersections,
	backgrounds,
	positioning,
	circuits.ee.IEC
	}

\DeclareRobustCommand{\tensormuchi}{\!
\mathbin{\begin{tikzpicture}[scale=0.45, every node/.style={scale=0.6}]
	\begin{pgfonlayer}{nodelayer}
		\node [style=none] (0) at (0, 0.25) {};
		\node [style=none] (3) at (0, -0.25) {};
		\node [style=none] (4) at (0, 0) {$\mu \chi$};
	\end{pgfonlayer}
	\begin{pgfonlayer}{edgelayer}
		\draw [bend left=90, looseness=2.50] (0.center) to (3.center);
		\draw [bend right=90, looseness=2.50] (0.center) to (3.center);
	\end{pgfonlayer}
\end{tikzpicture}}  \!}

\DeclareRobustCommand{\tensormubarchi}{\!
\mathbin{\begin{tikzpicture}[scale=0.45, every node/.style={scale=0.6}]
	\begin{pgfonlayer}{nodelayer}
		\node [style=none] (0) at (0, 0.25) {};
		\node [style=none] (3) at (0, -0.25) {};
		\node [style=none] (4) at (0, 0) {$\overline{\mu} \chi$};
	\end{pgfonlayer}
	\begin{pgfonlayer}{edgelayer}
		\draw [bend left=90, looseness=2.50] (0.center) to (3.center);
		\draw [bend right=90, looseness=2.50] (0.center) to (3.center);
	\end{pgfonlayer}
\end{tikzpicture}}  \!}

\DeclareRobustCommand{\tensorchi}{ \!
\mathbin{ \begin{tikzpicture}[scale=0.45, every node/.style={scale=0.6}]
	\begin{pgfonlayer}{nodelayer}
		\node [style=none] (0) at (0, 0.25) {};
		\node [style=none] (3) at (0, -0.25) {};
		\node [style=none] (4) at (0, 0) {$\chi$};
	\end{pgfonlayer}
	\begin{pgfonlayer}{edgelayer}
		\draw [bend left=90, looseness=1.75] (0.center) to (3.center);
		\draw [bend right=90, looseness=1.75] (0.center) to (3.center);
	\end{pgfonlayer}
\end{tikzpicture}} \! }

\DeclareRobustCommand{\tensorzeta}{\! 
\mathbin{ \begin{tikzpicture}[scale=0.45, every node/.style={scale=0.6}]
	\begin{pgfonlayer}{nodelayer}
		\node [style=none] (0) at (0, 0.25) {};
		\node [style=none] (3) at (0, -0.25) {};
		\node [style=none] (4) at (0, 0) {$\zeta$};
	\end{pgfonlayer}
	\begin{pgfonlayer}{edgelayer}
		\draw [bend left=90, looseness=1.75] (0.center) to (3.center);
		\draw [bend right=90, looseness=1.75] (0.center) to (3.center);
	\end{pgfonlayer}
\end{tikzpicture}} \! }

\DeclareRobustCommand{\tensormu}{ \!
\mathbin{ \begin{tikzpicture}[scale=0.45, every node/.style={scale=0.6}]
	\begin{pgfonlayer}{nodelayer}
		\node [style=none] (0) at (0, 0.25) {};
		\node [style=none] (3) at (0, -0.25) {};
		\node [style=none] (4) at (0, 0) {$\mu$};
	\end{pgfonlayer}
	\begin{pgfonlayer}{edgelayer}
		\draw [bend left=90, looseness=1.75] (0.center) to (3.center);
		\draw [bend right=90, looseness=1.75] (0.center) to (3.center);
	\end{pgfonlayer}
\end{tikzpicture}} \! }

\DeclareRobustCommand{\tensorxi}{\! 
\mathbin{ \begin{tikzpicture}[scale=0.45, every node/.style={scale=0.6}]
	\begin{pgfonlayer}{nodelayer}
		\node [style=none] (0) at (0, 0.25) {};
		\node [style=none] (3) at (0, -0.25) {};
		\node [style=none] (4) at (0, 0) {$\xi$};
	\end{pgfonlayer}
	\begin{pgfonlayer}{edgelayer}
		\draw [bend left=90, looseness=1.75] (0.center) to (3.center);
		\draw [bend right=90, looseness=1.75] (0.center) to (3.center);
	\end{pgfonlayer}
\end{tikzpicture}} \! }

\newtheorem{thm}{Theorem}
\newtheorem{definition}{Definition}
\newtheorem{proposition}{Proposition}
\newtheorem{lema}{Lemma}
\newtheorem{example}[thm]{Example}

\begin{document}

\title{Quantum networks theory}
\author{Pablo Arrighi}
\affiliation{Université Paris-Saclay, Inria, CNRS, LMF, 91190 Gif-sur-Yvette, France}
\author{Amélia Durbec}
\affiliation{CNRS, Centrale Lille, JUNIA, Univ. Lille, Univ. Valenciennes, IEMN, 59046 Lille Cedex, France}
\author{Matt Wilson}
\affiliation{PPLV Group, Department of Computer Science, University College London }
\affiliation{Quantum Group, Department of Computer Science, University of Oxford, UK}
\affiliation{HKU-Oxford Joint Laboratory for Quantum Information and Computation}

\begin{abstract} The formalism of quantum theory over discrete systems is extended in two significant ways. First, quantum evolutions are generalized to act over entire network configurations, so that nodes may find themselves in a quantum superposition of being connected or not, and be allowed to merge, split and reconnect coherently in a superposition. Second, tensors and traceouts are generalized, so that systems can be partitioned according to almost arbitrary logical predicates in a robust manner. The hereby presented mathematical framework is anchored on solid grounds through numerous lemmas. Indeed, one might have feared that the familiar interrelations between the notions of unitarity, complete positivity, trace-preservation, non-signalling causality, locality and localizability that are standard in quantum theory be jeopardized as the neighbourhood and partitioning between systems become both quantum, dynamical, and logical. Such interrelations in fact carry through, albeit two new notions become instrumental: consistency and comprehension.\medskip\medskip
\end{abstract}


\maketitle

\section{Introduction}

{\em Why study fully quantum networks?} Composite systems such as computer processes \cite{PapazianRemila}, neurons \cite{FarrellyQNN}, biochemical agents \cite{MurrayDicksonVol2}, particle systems \cite{MeyerLove}, market agents \cite{KozmaBarrat}, social networks users, are often modelled as dynamical networks. Dynamical network models are suitable whenever the connectivity and population of the objects modelled are subject to change, for instance agents within social networks have the capabilities to spawn, disappear, and form or (more sadly) sever connections. Another class of composite systems which has attracted considerable attention lately are those whose constituents are quantum in nature, as they can be leveraged to perform computational and information processing tasks with efficiencies beyond those of their best known classical counterparts. In this paper we will argue that there are in fact many motivations for developing mathematical framework suitable for reasoning about composite systems in which \textit{all} of the fundamental features of dynamical networks are quantum in nature---including connectivity and population.

One such motivation is the evocative idea of a `quantum internet' \cite{QuantumNetworksKimble,QuantumNetworksCirac,QuantumNetworksBianconi}. The development of a fully quantum internet echoes a fundamental question in Computer Science: \textit{What exactly is a computer}? To the best of our knowledge, the key resources granted to us by nature for the sake of efficient computing are spatial parallelism between systems and quantum parallelism within the systems. %
This observation motivates the questions: {\em Are spatial parallelism and quantum parallelism independent computational resources, or could the former be subject to the latter? 
What then, is a quantum network of quantum computers?}

Indeed, a second motivation lies in the study of models of distributed quantum computing. So far only classical dynamical networks of quantum automata \cite{ArrighiUCAUSAL,GayQPLreview} have been addressed. Yet it is known that superpositions of connectivities within distributed quantum computers \cite{ArrighiQCGD} could used to implement protocols in which the orderings of events \cite{ValironQSwitch} and the trajectories of particles \cite{ChiribellaSecondShannon} are quantum in their specification. The importance of modelling the implementation \cite{QswitchExp1,QswitchExp2,exp-N-switch} of such protocols is well-argued by noting that in several tasks, they are more efficient than their standard quantum counterparts \cite{OreshkovQuantumOrder,FacchiniPermutation,CostaComplexity, ChiribellaEnhanced, ChiribellaSecondShannon, Abbott_2020}. Despite these established advantages of quantising the networking of quantum systems, very little is known about the formalisation and the dynamics of such networkings \cite{BruknerDynamics}.

As it turns out, motivations for considering quantum networks appear not only in the development of quantum technologies, but also in the foundations of physics. Whilst a theory which satisfactorily quantises gravity remains elusive, the intersection of the fundamental principles of general relativity and quantum theory suggests that such a theory will have basic features in common with fully quantum networks. Indeed the geometry of spacetime is dependent on mass distribution in general relativity, and mass distribution may be superposed in quantum theory. As a result a striking feature shared by most attempts to quantise gravity is the possibility of quantum superpositions of spacetime geometries. Whilst the effects of superposing spacetime geometries have long been thought to be inaccessible experimentally, recent realistic experimental proposals have been infanted through research at the crossover of quantum information and quantum gravity \cite{BoseSpinEntanglementWitness2017,BoseMASSIVE2018,MarlettoWhencangravity2018,MarshmanLocalityEntanglementTableTop2019} 
which promise the near-term confirmation or invalidation of this feature \cite{Christodouloupossibilitylaboratoryevidence2019,
Christodouloupossibilityexperimentaldetection2018b}.\\ 
In most theories of quantum gravity, geometries are indeed represented by networks dual to simplicial complexes \cite{RovelliLQG,LollCDT}, and so the networking structure of systems is both a prominent conceptual feature and a feature whose quantum nature deserves considerable attention in its own right. The current network representations of geometry in theories of quantum gravity do however come with a few sources of discomfort. The first, a seemingly technical issue, is that the precise definition of such geometries is often left informal, whereas choices regarding whether to name nodes and how, or whether to order their neighbours or not, can have dramatic consequences in terms of disallowing superluminal signalling, allowing for expansion, allowing for fermionic propagation \cite{ArrighiNamesInQG,ArrighiCreation,ArrighiTetrahedra}. 
The second is an uncomfortable physical consequence of the path integral formulation of dynamics, namely that such a formulation may very well jeopardize the unitary of quantum theory. A third uncomfortable feature of path integral formulation is that fundamental physical principles such as locality and causality are only seen as emergent in these theories, recovered by means of heroic, and not entirely rigorous classical limits.

{\em First strand of contributions.}
Building on the above motivations, this paper is to provide a rigorous mathematical framework for reasoning about fully quantum networks and their dynamics. The framework is sufficiently flexible to allow for arbitrary quantum superpositions of entire networks, including superpositons of connectivity and population. All of these network features evolve unitarily. For instance edges can be dropped or created, nodes can split or merge in a coherent manner. Yet, the framework offers rigorous definitions of local, nearest-neighbour interactions, as well as global, non-signalling causal unitary dynamics---in a strict, non-emergent sense. Thus, the first application of this work may be the provision of a rigorous mathematical framework for those theories of quantum gravity in which geometries are represented by networks, e.g. allowing them to seek for strictly causal, unitary dynamics. The provision of models for quantum complex systems and distributed quantum computing is another application domain.\\
To reconcile the flexibility of fully quantum networks dynamics with rigorous notions of causality and unitarity was a long road. For instance, a first potential ambush was explained in \cite{ArrighiNamesInQG}, i.e. the causality of a quantum network dynamics only makes sense if its nodes are named. This is because of the role that names play in specifying the relative alignment of the network configurations as they are placed in a quantum superposition. Still, despite the fact that nodes need \textit{a} name, the \textit{actual name choice} should have no effect on the evolution of a network beyond the aforementioned role: this independence is formalised by the notion of renaming-invariance. A second potential ambush was explained in \cite{ArrighiCreation}, i.e. we must ensure that names are no obstacle to unitary node creation/destruction. Indeed, suppose that a node $u$ splits into $u$ and $v$. How can this evolution be a unitary $U$? Won't $U^\dagger$ just erase name $v$? If $v$ gets renamed into $v'$ before acting with $U^\dagger$, do we still get the node $u$? These questions are solved by means of a name algebra discovered in the context of reversible causal graph dynamics. In short, a node $u$ can be split into its left part $u.l$ and its right part $u.r$. Such a left-right pair can in turn re-merge to form $u.l \vee u.r=u$, see Fig. \ref{fig:namealgebra}.

\begin{figure}\centering
\includegraphics[width=\textwidth]{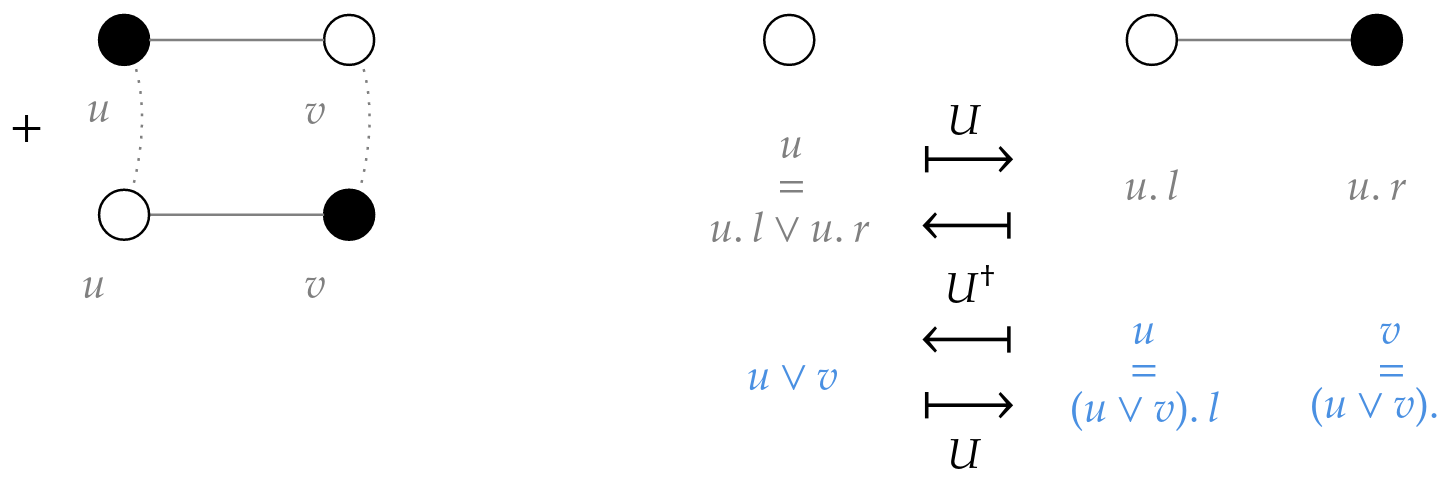}
\caption{\label{fig:namealgebra} {\em Necessity of the name algebra.} Left: naming vertices is necessary in order to track alignment across quantum superpositions. Right/grey: A quantum evolution may split $u$ into $u.l$ and $u.r$. As the inverse evolution merges them back we need $( u.l\lor u.r) =u$. Right/blue: The inverse quantum evolution may also merge vertices $u$ and $v$ into $( u\lor v)$. As the forward evolution splits them back we need $( u\lor v) .l=u$ and $( u\lor v) .r=v$.}
\end{figure}

{\em Second strand of contributions.} Whilst developing a model of fully quantum networks, we realised that by taking entire network configurations as states, the usual notion of tensor product fails us on very basic unsettling questions, such as: \textit{When two nodes are connected, with one on the left of a tensor product, and the other on the right, where does the edge between them live?} \textit{When two nodes $u$ and $v$ are in a quantum superposition of being connected or not, and the neighbours of node $u$ are placed left of the tensor product, which side of the tensor product is node $v$?}

In order to address these questions we were forced to generalize the usual notion of tensor product, by means of an approach which is both decompositional and operational.\\
The approach is decompositional in the sense that, given an entire network $\ket{G}$ we discriminate which nodes go left $\ket{G_{{\chi}}}$ and which nodes go right $\ket{G_{\overline{\chi}}}$ of the tensor product $\tensorchi$, based on an almost arbitrary logical predicate $\chi$, so that we have $\ket{G}=\ket{G_\chi}\tensorchi\ket{G_{\overline{\chi}}}$.
To answer the first question of whereabouts an edge between nodes $u\in V(G_\chi)$ and $v\in V(G_{\overline{\chi}})$, we extend the name algebra with a symbol ``$\hyphenbullet$" used to encode edges within node names, e.g. node $u=\hyphenbullet x$ is connected to node $v= x \vee y$. That way, the information about the cross-connectivity between $G_\chi$ and $G_{\overline{\chi}}$ remains held within their node names, see Fig. \ref{fig:edgesacrosstensor}.\\
\begin{figure}\centering
\[\includegraphics[scale = 0.2]{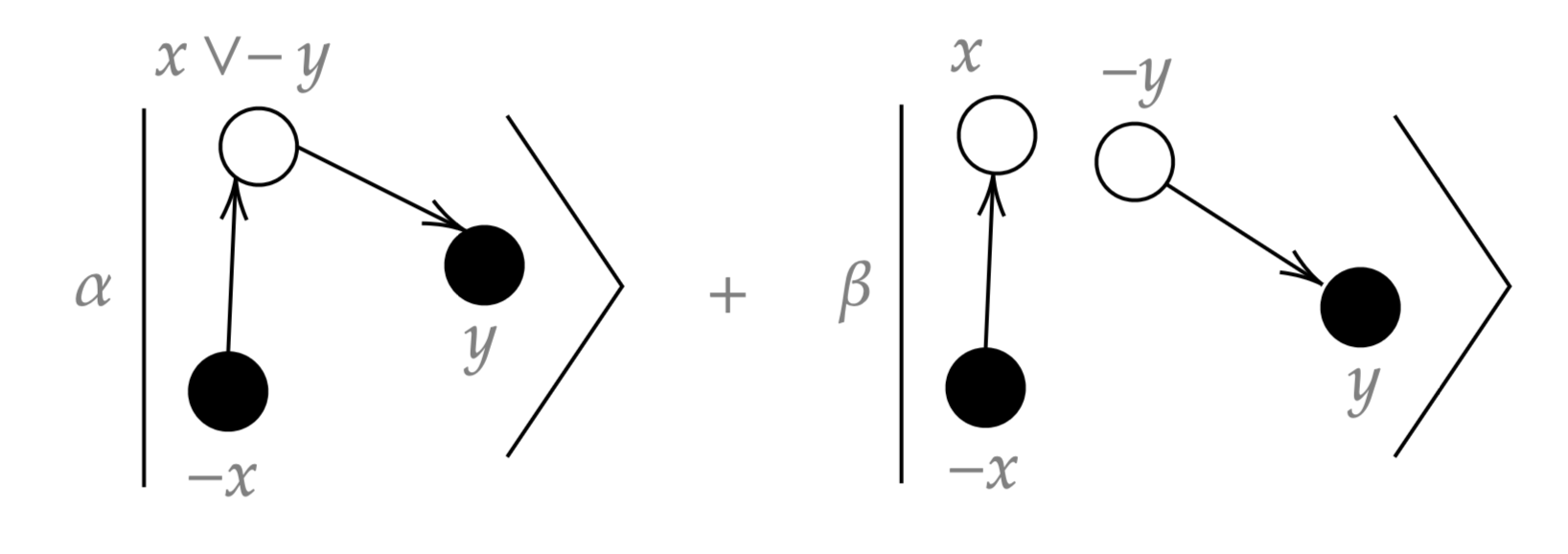}.\]
\[\includegraphics[scale = 0.2]{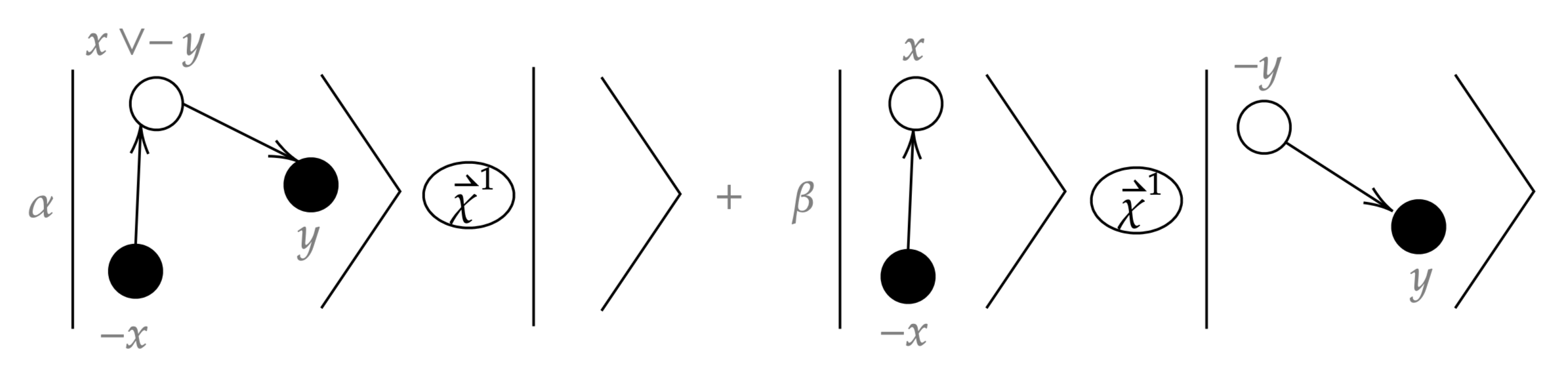}. \]
\caption{\label{fig:edgesacrosstensor} {\em Tensor decomposition operation on a superposition of geometries.} Top:  Connectivity is encoded within the names of the nodes, e.g. $\hyphenbullet x$ being connected to node $x \vee \hyphenbullet y$ in the first branch of the superposition. In this first branch node $y$ lies within radius $1$ of node $\hyphenbullet x$, in the other they are disconnected. Bottom: As this $\harp{\chi }^1:=\harp{\zeta }_{\hyphenbullet x}^{2}$ selects the oriented radius two neighbours of $\hyphenbullet x$, node $y$ finds itself in a superposition of falling left or right of the tensor.}
\end{figure}
The approach is operational in the sense of linearity, e.g. consider the following expression
$$\ket{G}+\ket{H}=\ket{G_\chi}\tensorchi\ket{G_{\overline{\chi}}}+\ket{H_\chi}\tensorchi\ket{H_{\overline{\chi}}},$$
then $\chi$ may well choose to place a note $w$ left of the tensor product in the $G$ branch of the superposition, but right of the tensor product in the $H$ branch of the superposition, see Fig. \ref{fig:edgesacrosstensor}. This answers the second question on the whereabouts of superposed neighbourhoods.

One question leads to another: \textit{Are these generalized tensor products always defined? For instance, what about the tensor product of a given network, with itself?} \textit{On what ground are we allowed to discriminate which nodes go left of the tensor, and which nodes go right? For instance, can we base this on their proximity to other nodes, their states, or combinations of these?}
\textit{In what sense are these generalized tensor products a generalization of the usual ones? Do they also generalize the direct sum?}
So that the decomposition be unambiguous, we ask that whenever $L$ and $R$ are not `consistent', i.e. not of the form $L=G_\chi$ and $R=G_{\overline{\chi}}$, then $\ket{L}\tensorchi \ket{R}=0$. We define a set of sufficiently well-behaved logical predicates $\chi$ for our purposes, which we refer to as `restrictions'.  Restrictions then lead to a natural notion of partial trace $(\ket{G}\bra{H})_{|\chi}=\ket{G_\chi}\bra{H_\chi}\braket{H_{\overline{\chi}}|G_{\overline{\chi}}}$ which is itself a completely positive and trace-preserving operation. Notice that the partial trace is denoted with a vertical bar (e.g. $\rho_{|\chi})$, to distinguish it from the restriction (e.g. $G_\chi$).
{ The Hilbert space ${\cal H}$ and restrictions $\chi$ are left deliberately general in this paper, so that they can be used to reason superpositions of geometries, as well as other scenarios, including the more familiar context of many qudits $u$, $v$, \ldots say. For instance, by considering the restriction $\zeta_u$ which keeps just the node $u$ of any graph 
\[\zeta_u:\quad\includegraphics[align=c,scale = 0.16]{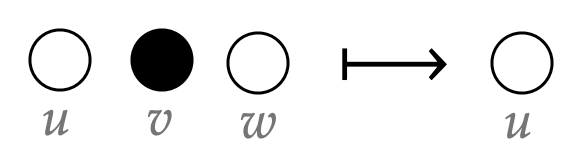}, \] 
we recover the standard tensor product of the first qudit with the next. Similarly (encoding bits for aesthetic purposes as $\{white, black\}$ as opposed to $\{0,1\}$) by considering the restriction $\chi_{w = white}$ which keeps all nodes when the name $w$ is present and in state $\ket{white}$ 
\[\chi_{w = white}:\quad\includegraphics[align=c,scale = 0.16]{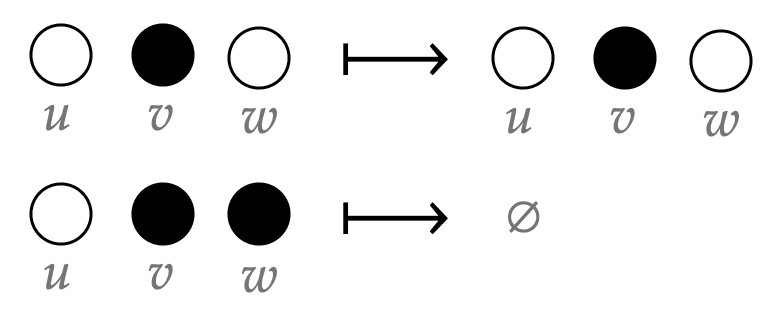},\] 
we recover the direct sum between the subspace of graphs in which $w$ is white and the rest. More elaborate decompositions as tensor products and direct sums can then be modelled by restrictions such as $\zeta_{u:white}$, which keeps only node $u$ of a network, and only if its state is $\ket{white}$ 
\[\zeta_{u:white}:\quad\includegraphics[align=c,scale = 0.16]{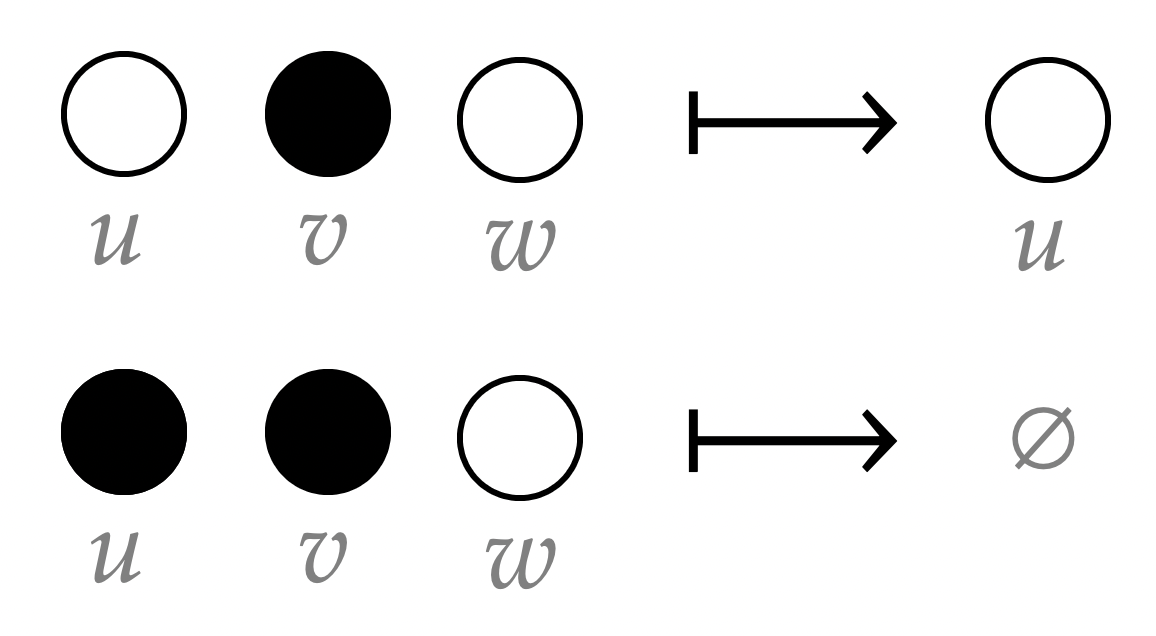},\]
as well as $\mu$ which keeps all nodes marked as white from any graph 
\[\mu:\quad\includegraphics[align=c,scale = 0.16]{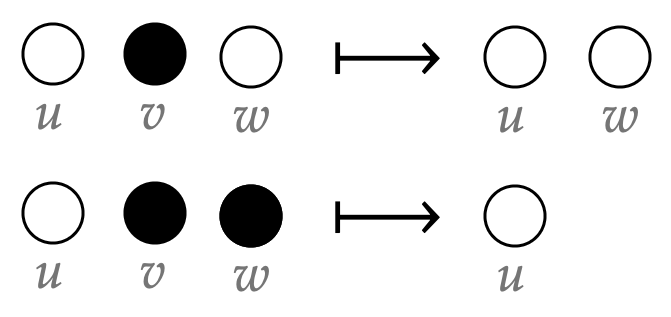}.\] 
These of course, are presented in addition to restrictions which more explicitly reference graph geometry, and so pull us away from simple concrete representations in terms of tensor products and direct sums, for instance the restriction $\zeta_x^{2}$ which takes the non-oriented radius $2$ neighbourhood around $x$ \[\includegraphics[scale = 0.2]{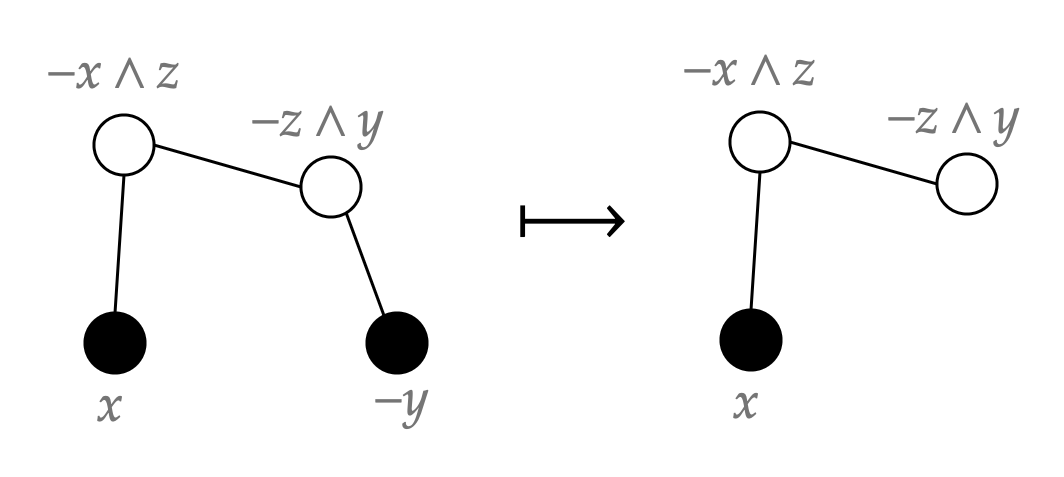}.\] The point here, is that we can reason about any of these ways of carving out subsystems, simply because in each case those ways are representable as restrictions. Thus, generalised tensor products may constitute a powerful analytical tool in quantum information, with many potential applications e.g. defining generalised, robust forms of entanglement.
 \color{black}}

{\em Relationship between the two strands.}
It is the notion of a restriction $\chi$ on a network which is then used to define the notions of locality and causality, which in the presence of unitarity will turn out to interact in a physically intuitive way. We saw above (Fig. \ref{fig:edgesacrosstensor}) that superpositions of geometries can lead to superpositions of populations left and right of the tensor product $\tensorchi$. Yet, and even in these situations, the connective $\tensorchi$ is sufficiently well-behaved to serve as the stepping stone for a robust notions of locality and causality.
An operator is considered to be local on a restricted part $\chi$ of a network if it alters only that which is within $\chi$, ignoring the remainder. Here we define $\chi$-locality of operators as the satisfaction of $\bra{H}A\ket{G}=\bra{H_\chi}A\ket{G_\chi}\braket{H_{\overline{\chi}}|G_{\overline{\chi}}}$ and then prove the equivalence of this definition with both gate-locality in the operational picture, and dual locality in the Heisenberg picture. An operator considered to be causal between two restrictions $\chi,\zeta$, may on the other hand alter the entire network, but its effects on region $\zeta$ must be fully determined by causes in region $\chi$. Here we define $\chi\zeta$-causality as the satisfaction of $(U\rho U^\dagger)_{|\zeta}=(U\rho_{|\chi} U^\dagger)_{|\zeta}$ and prove equivalence of this definition with dual causality in the algebraic, Heisenberg picture. Whilst all the usual interrelations between these notions carry through, the path to them is full ambushes: consistency checking requires great care, and the usual notion of subrestriction requires an extra condition (`comprehension') before it behaves as expected---which fortunately vanishes in the name-preserving superselected space.

{\em Plan.} Sec. \ref{sec:graphs} describes the name algebra used for the naming of nodes, defines network configurations and their induced state space, as well as defining the notions of renaming-invariance and name-preservation. Sec. \ref{sec:locality} defines locality in its various forms, proves their equivalences, and proves that every unitary operator can be extended into a local unitary operator. Sec. \ref{def:causality} defines causality in its two forms, proves their equivalences, and shows that causal unitary operators can be implemented by a product of local ones. Sec. \ref{sec:conclusion} summarizes the contributions of this paper and outlines potential avenues for future work on the formalism presented. Sec. \ref{sec:conclusion} furthermore enumerates several perspective applications of the formalism, beyond those that motivated work originally.\\
Introducing a new formalism for quantum theory is a slippery exercise. We have had reestablish basic facts first through numerous lemmas found in Appendix \ref{sec:lemmas}, including compositionality laws akin to the axioms of categorical approaches to quantum theory. The consequences of imposing renaming-invariance on the dynamics of networks are explored in \ref{sec:renaminginvariance}. The formal sense in which these restrictions really recover notions of direct sum and tensor product is given in Appendix \ref{sec:reconstruction}, since some care has to be taken to move away from the Fock-space viewpoint of this paper.

\section{An algebra for naming nodes of quantum networks}\label{sec:graphs}

The problem of defining superpositions of graphs immediately leads to the following conundrum. Consider a pair of systems, white and black, superposed with again a pair of systems, black and white. One must decide whether the mathematics assigned to this sentence should be either of $\ket{\smwhtcircle \harrowextender \smblkcircle } +\ket{\smblkcircle \harrowextender \smwhtcircle }$ or $2\ket{\smwhtcircle \harrowextender \smblkcircle }$ (where no distinguishment is made between black-white and white-black). The only way to disambiguate this situation is by naming those vertices. The alternative choice, to neglect this alignment information by claiming that it does not matter since the graphs are isomorphic, leads to the physically unreasonable consequence of permitting super-luminal signalling \cite{ArrighiNamesInQG}.

Still, vertex names can be cumbersome. In the classical regime, and in a variety of different early formalisms, it was shown that their presence leads to vertex-preservation, i.e. the forbidding of vertex creation/destruction \cite{ArrighiRCGD} . This was a somewhat uncomfortable situation, because the informally defined model of Hassalcher and Meyer \cite{MeyerLGA} did seem to feature reversibility, vertex creation/destruction, and non-signalling causality. Again in the classical regime, the issue was finally solved by introducing a \textit{name algebra} \cite{ArrighiCreation}. We now bring the notion of a name algebra over to the quantum regime, simplified. First, let us remind the reader of why we cannot do without such an algebra. 

Indeed, say as in Fig. \ref{fig:namealgebra} that some quantum evolution splits a vertex $u$ into two. We need to name the two infants in a way that avoids name conflicts with the vertices of the rest of the graph. But if the evolution is locally-causal, we are unable to just `pick a fresh name out of the blue', because we do not know which names are available. Thus, we have to construct new names locally. A natural choice is to use the names $u.l$ and $u.r$ (for left and right respectively). Similarly, say that some other evolution merges two vertices $u,v$ into one. A natural choice is to call the resultant vertex $u\lor v$, where the symbol $\lor$ is intended to represent a merger of names.

This is, in fact, what the inverse evolution will do to vertices $u.l$ and $u.r$ that were just split: merge them back into a single vertex $u.l\lor u.r$. But, then, in order to get back where we came from, we need that the equality $u.l\lor u.r=u$ holds. Moreover, if the evolution is unitary, as is prescribed by quantum mechanics, then this inverse evolution does exists, therefore we are compelled to accept that vertex names obey this algebraic rule.

Reciprocally, say that some evolution merges two vertices $u,v$ into one and calls them $u\lor v$. Now say that some other evolution splits them back, calling them $(u\lor v).l$ and $(u\lor v).r$. This is, in fact, what the inverse evolution will do to the vertex $u\lor v$, split it back into $(u\lor v).l$ and $(u\lor v).r$. But then, in order to get back where we came from, we need the equalities $(u\lor v).l=u$ and $(u\lor v).r=v$.

{\em A quick note on notations. Throughout the paper, the symbol $:=$ means `is defined by'. Unquantified letters are implicitly introduced by a ``for all" across the range that corresponds to the typographic convention used.} 

We now formally introduce the algebra that we will use the name nodes:
\begin{definition}[Name algebra]\label{def:namealgebra}

Let $ \mathbb{K}$ be a countable set. Let $ \hyphenbullet \mathbb{K} :=\{\hyphenbullet c\ |\ c\in \mathbb{K}\}$. 
The name algebra $ \mathcal{N}[\mathbb{K}]$ has terms given by the grammar 
\begin{equation*}
u,v\ ::=\ c\ |\ u.t\ |\ u\lor v\quad\text{with} \quad c\in \mathbb{K} ,\ t\in \{l,r\}^{*}
\end{equation*}
and is endowed with the following equality theory over terms (with $ \varepsilon $ the empty word):
\begin{equation*}
( u\lor v) .l=u\qquad ( u\lor v) .r=v\qquad u.\varepsilon =u\qquad u.l\lor u.r=u
\end{equation*}
We define $ \mathcal{V:=N}[\mathbb{K} \cup \hyphenbullet \mathbb{K}]$. We write $ V\ \corresponds V'$ if and only $ \mathcal{N}[ V] =\mathcal{N}[ V']$.
\end{definition}

{\em From now on we take $\mathbb{K} =\mathbb{N} \setminus \{0\}$. We use letters $x,y$ to denote elements of $\mathbb{K} \cup \hyphenbullet \mathbb{K}=\mathbb{Z}\setminus \{0\}$.}

The $\pm $ sign will be used to capture the two tips of the edges of the graphs. (NB: To deal with $d-$dimensional simplicial complexes instead, we might for instance encode them as  $\pm -$bipartite graphs, or take $x\in \llbracket 0,d\rrbracket \times \mathbb{K}$.)

\subsection{Defining graphs}

Next, we take a `system' to mean both a `state' and a `name', whereas a `graph' is a set of systems having disjoint names, see Fig. \ref{fig:graphs}. Our formal definition of a graph is given by the following:

\begin{definition}[Graphs]\label{def:graphs}
Let $ \Sigma $ be the set of internal states. We define the set $S := \Sigma \times \mathcal{V}$, whose elements will be referred to as \textit{systems}, with 
\begin{itemize}
\item $ \sigma \in \Sigma $ the internal state of the system
\item $ v\in \mathcal{V}$ the name which supports the system
\end{itemize}
A graph $ G$ is a finite set of systems such that
\begin{align}
\sigma .v,\ \sigma '.v'\in S \textrm{ and } v.t=v'.t' \textrm{ implies } \sigma =\sigma' \textrm{ and } v=v' \textrm{ and } t=t'\label{eq:wellnamedness}
\end{align}
We define its support $ V( G) :=\{v\ |\ \exists \sigma \in \Sigma \ s.t. \ \sigma .v\in G\}$. We denote by $\mathcal{G}:=\{ G\ |\ G\subseteq \Sigma\times {\cal V} \textrm{ and } G \textrm{ is a graph}\}$ the set containing every possible graph. \\ 
We denote by $ \mathcal{H}$ the Hilbert space whose canonical basis is labelled by the elements of $ \mathcal{G}$. 
We denote by $\mathcal{B}_1(\mathcal{H})$ the trace class operators on $\mathcal{H}$.
Consider $\mathcal{C}\subseteq\mathcal{G}$ some constrained configurations. We denote by $ \mathcal{H}^\mathcal{C}$ the Hilbert space whose canonical basis is labelled by the elements of $ \mathcal{C}$. 
\end{definition}
Notice that a graph $G$ could be seen as a partial function from ${\cal V}$ to $\Sigma$. Defining it as a relation, i.e. a set of pairs, will be convenient later in order to define set theoretical operations upon them. 

We are aware that the above definition is not quite the traditional one. Here, edges are derived information from the systems. 
\begin{figure}\centering
\includegraphics[width=0.8\textwidth]{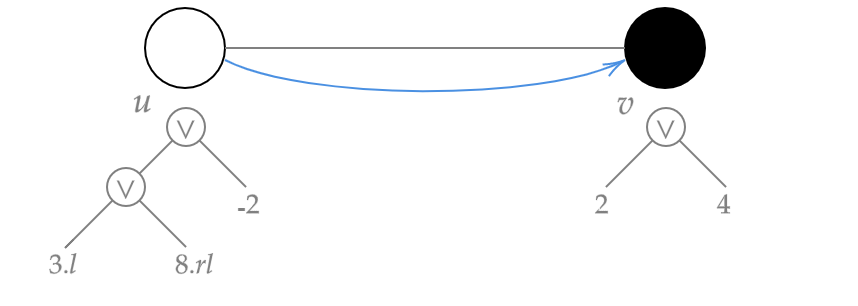}
\caption{\label{fig:graphs} {\em Graphs.} Left: A system with state 'white' and name $u=(( 3.l\ \lor 8.rl) \lor \hyphenbullet 2)$. Right: A system with state 'black' and name $v=( 2\lor 4)$. Middle/grey: Here we decided to interprete $u.r=\hyphenbullet 2$ and $v.l=2$ as the presence of an unoriented edge $\{u,v\}$. Middle/blue: We could have chosen to interprete it as an oriented edge $( u,v)$ instead. Middle: In both cases, geometry is derived from relative information that is already present within systems, and which is invariant under renamings.}
\end{figure}

\begin{definition}[Induced edges]\label{def:inducededges}
The following defines the induced undirected edges: 
\begin{equation*}
E( G) :=\left\{\{v,v'\} \ |\ v.t=\hyphenbullet x.s\quad \text{and} \quad v'.t'=x.s\quad \text{and} \quad \sigma .v,\sigma '.v'\in G\quad\text{with}\quad\sigma .v\neq \sigma '.v'\right\}
\end{equation*}
The following defines the induced directed edges: 
\begin{equation*}
 \harp{E}( G) :=\left\{( v,v') \ |\ v.t=\hyphenbullet x.s\quad \text{and} \quad v'.t'=x.s\quad \text{and} \quad \sigma .v,\sigma '.v'\in V(G)\quad \text{with}\quad \sigma .v\neq \sigma '.v'\right\}.
\end{equation*}
\end{definition}

Notice that in both these conventions, geometrical information is encoded by means of names.

Most often we want those names to indeed describe the geometry, and nothing else. In other words the geometry and the dynamics that governs its evolution need be renaming-invariant. 

\begin{definition}[Renaming and renaming-invariance]\label{def:renaminginvariance}

A renaming is an isomorphism $ R:\mathcal{N}[\mathbb{K}]\rightarrow \mathcal{N}[\mathbb{K}]$, i.e. a bijection such that
\begin{equation*}
R( u.t) =R( u) .t\qquad R( u\lor v) =R( u) \lor R( v)
\end{equation*}
It is fully specified by its action on domain $ \mathbb{K}$.\\
It is extended to $ \mathcal{V}$ by letting $ R( \hyphenbullet x) :=\hyphenbullet R( x)$, with $ \hyphenbullet ( u.t) :=\hyphenbullet u.t$ and $ \hyphenbullet ( u\lor v) :=\hyphenbullet u\lor \hyphenbullet v$.\\
It is extended to $ \mathcal{S}$ by letting $ R( \sigma .v) =\sigma .R( v)$. It is extended to $ \mathcal{G}$ by acting pointwise. It is extended to $ \mathcal{H}$ by linearity.

Let $ A$ be an operator over graphs, possibly parameterized by $ v\in \mathcal{V}$. It said to be renaming-invariant if and only if $ RA_{v} =A_{R( v)} R$.
\end{definition}

Renaming-invariance and its consequences are worked out in appendix \ref{sec:renaminginvariance}. For instance it leads to $\pm $-name-preservation. Yet, several results of this paper require full name-preservation:

\begin{definition}[Name-preservation]\label{def:np}
Let $A$ be an operator over graphs. It is said to be name-preserving (n.-p. for short) if and only if $V(G)\not{\corresponds} V( H)$ implies $\bra{H} A\ket{G} =0$.
\end{definition}

Again, notice that name-preservation does not prevent node-creation, for instance node $2$ is allowed to split into $2.l$ and $2.r$. Nor does it prevent edge-creation, for instance both $2.l$ and $2.r$ will then be connected to node $\hyphenbullet 2$, say. 

Throughout the paper we track what becomes of the renaming-invariance and name-preservation properties, making it clear whenever they are used as necessary premises of the established result.

Interestingly traceouts preserve\ldots name-preservation.

\section{Generalized tensors and traces over quantum networks}\label{sec:tensors}

We will now generalise the partial trace first, and the tensor product, by means of a decompositional and operational approach.

In the traditional, compositional approach, qubit $u$ has Hilbert space ${\cal H}^u$, qubit $v$ has Hilbert space ${\cal H}^v$, and the pair of qubits $uv$ has Hilbert space ${\cal H}^{uv}$. The tensor product takes $\ket{\sigma}\in {\cal H}^u$ and $\ket{\sigma'}\in {\cal H}^v$ to yield $\ket{\sigma,\sigma'}\in {\cal H}^{uv}$, as such it is not an internal binary operator of some Hilbert space, and in fact implicitly carries a strong `typing information' ${\cal H}^u\times {\cal H}^v\rightarrow {\cal H}^{uv}$.\\ 
In contrast, the starting point of our approach is a Hilbert space that is quite large: it can accommodate both a graph and its subgraphs; a single qubit $u$ or $v$, and the pair of them $uv$ or even none of these. I.e. ${\cal H}={\cal H}^u\oplus {\cal H}^v \oplus {\cal H}^{uv} \oplus \ldots$. E.g. the superposition $\alpha\ket{0.u}+\beta\ket{0.v}$ happily denotes a qubit with a superposition of names, whereas $\alpha\ket{\varnothing}+\beta\ket{0.u}+\gamma\ket{1.u}\otimes\ket{1.v}$ denotes a superposition of having zero, one or two qubits. This is reminiscent of a Fock space of particles, but rather, it is similar to a Fock space of qubits as can be found in the quantum programming language literature whenever the language includes a \texttt{new qubit} creation operation. As a consequence, our generalised tensor product takes $\ket{\sigma.u}\in {\cal H}$ and $\ket{\sigma'.v}\in {\cal H}$ to yield $\ket{\sigma.u,\sigma'.v}\in {\cal H}$, as such it is an internal binary operator of some Hilbert space and has no typing information. Then, the tensor product becomes a way of decomposing a state in ${\cal H}$, e.g. 
$\ket{0.u}\tensorchi_{\!u}\ket{0.v,0.w}=\ket{0.u,0.v,0.w}$, but there could be many others, e.g. $\ket{0.v}\tensorchi_{\!v}\ket{0.u,0.w}=\ket{0.u,0.v,0.w}$. The restriction $\chi_u$, $\chi_v$, \ldots is what specifies what needs to go left of the tensor, thereby replacing its typing information. The formalism is designed so that any graph $\ket{G}$ can be decomposed according to any restriction $\chi$ according to $\ket{G}=\ket{G_\chi}\tensorchi\ket{G_{\overline{\chi}}}$, for instance if the restriction takes all, the rest can always be the empty graph, e.g. 
$\ket{0.u}=\ket{0.u}\tensorchi_{\!u}\ket{\varnothing}$. It is not the case, however, that any two pairs of graphs can be composed\footnote{The choice of sending these expression to the null vector, rather than letting the tensor product be partially defined, will be necessary in order to be able to characterize $\chi$-local operators as being of the form $A\tensorchi I$ later.}, e.g. $\ket{0.u}\tensorchi\ket{0.u}=0$. Still, the generalised tensor product does generalize the traditional tensor product, as made formal in App. \ref{sec:reconstruction}. Interestingly it even generalizes the direct sum (for suitable choices of restrictions), combinations of direct sums and tensor products, and beyond.  

We begin by introducing these restrictions $\chi$ and generalised partial traces, which are similarly internalised act over ${\cal B}_1({\cal H})$, e.g. dropping system $v$ in $\ket{0.u,0.v}\bra{0.u,0.v}$ yields $\ket{0.u}\bra{0.u}$. App. \ref{sec:reconstruction} also discusses the nature of those partial traces induced by restrictions suited to direct sums. Interestingly, and thanks to the empty graph playing the role of a placeholder, these can be interpreted as a linear trace-preserving alternative to the Born rule.

\begin{definition}[Restrictions, partial trace, comprehension]\label{def:traceouts}
Consider a function
\begin{align*}
\chi &: \mathcal{G}\rightarrow \mathcal{G}\\
\chi &: G\mapsto G_{\chi } \subseteq G
\end{align*}
It is called a restriction if and only if $G_{\chi\chi}=G_\chi$, where we introduced the notation $ G_{\chi \zeta } :=( G_{\chi })_{\zeta }$ and thus $ \chi \zeta :=\zeta \circ \chi $. We use the notation $ G_{\overline{\chi }} :=G\setminus G_{\chi }$ even though the function $ \overline{\chi }: G\mapsto G_{\overline{\chi }}$ is not necessarily a restriction..\\
A restriction is extensible if and only if $\ G_{\chi} \subseteq H\subseteq G \Rightarrow H_{\chi} = G_{\chi}$. Given $ \zeta $ an extensible restriction, $ \zeta ^{r}$ is the extensible restriction such that $ G_{\zeta ^{r}}$ is induced by $ V( G_{\zeta })$ and its $ r$-neighbours according to $ E( G)$.
Similarly for  $  \harp{\zeta }^{r}$ and $  \harp{E}(G)$, where the $ r$-neighours are those pointing toward $ V( G_{\zeta })$.\\
Given $\chi, \zeta $ two extensible restrictions, $ \chi \cup \zeta $ is the extensible restriction such that $ G_{\chi \cup \zeta } :=G_{\chi } \cup G_{\zeta }$.\\
A restriction is pointwise if and only if $ G_{\chi } =\bigcup _{\sigma .v\ \in \ G}\{\sigma .v\}_{\chi }$.\\
They induce a partial trace:
$ (\ket{G}\bra{H})_{|\chi } := \ket{G_{\chi }} \bra{H_{\chi }} \braket{ H_{\overline{\chi }} | G_{\overline{\chi }}}$.\\
Let $ \rho $ denote a trace-class operator.\\
$ \rho _{|\chi }$ is defined from the above by linear extension.\\
$ \rho _{|\emptyset }$ denotes the usual, full trace $ \text{Tr}( \rho )$.\\
We use the notation $ [ \chi ,\zeta ] =0$ to mean that $ \chi \zeta =\zeta \chi $.\\
We say that $\zeta$ is comprehended within $\chi$ and write $ \zeta \sqsubseteq \chi $ if and only if $ G_{\chi \zeta } =G_{\zeta }$ and 
\begin{equation}
\braket{H_{\overline{\zeta }}|G_{\overline{\zeta }}} =\braket{H_{\chi \overline{\zeta }}|G_{\chi \overline{\zeta }}}\braket{H_{\overline{\chi }} | G_{\overline{\chi }} }\label{eq:comprehension}
\end{equation}
as illustrated in Fig. \ref{fig:comprehension}.
\end{definition}
{\em Soundness.} We need to show that if $\chi, \zeta$ are two extensible restrictions, then $\zeta^r$ and $\chi\cup\zeta$ are extensible restrictions. See Lemmas \ref{lem:restrictions}, \ref{lem:specialrestrictions}, \ref{lem:combiningrestrictions}. 
\hfill
$\qed $\\

Before going on lets be explicit about some basic examples
\begin{example}[Namewise restriction]\label{ex:1}

The most familiar way of picking out a subsystem in quantum information theory, is to pick a factor of some standard tensor product.  The closest analogy to this here is to use a namewise restriction such as $\zeta_u$, where $G_{\zeta_u} = \bigcup_{\sigma.v \in G:v=u}\sigma.v$. In this setting we choose to consider the constrained Hilbert space $\mathcal{H}^{\mathcal{C}}$ of graphs spanned by the set $\mathcal{C}$ of all graphs such that if $v \in V(G)$ and $v \wedge u$ then $v = u$.  Consider the unitary isomorphism $E: \mathcal{H}^{\mathcal{C}} \cong (\mathcal{H}^{u + \varnothing}) \otimes \mathcal{H}^{\mathcal{C}-u}$ where $\mathcal{H}^{u + \varnothing}$ is the Hilbert space generated by states of the form $\ket{\sigma . u}$ or $\ket{\varnothing}$ with $u$ fixed and $\sigma$ allowed to vary within $\Sigma$.
The set $\mathcal{C}-u$ is taken to be the set of all graphs in $\mathcal{C}$ which do not have a node $u$. 
Then, the generalised trace for this restriction induces a map $(\bullet)_{| \zeta_u}: \mathcal{B}_1(\mathcal{H}) \rightarrow \mathcal{B}_1(\mathcal{H}^{u + \varnothing})$ 
which acts explicitly as \[ \rho_{| \zeta_u} \cong Tr_{\overline{u + \varnothing}} ( \rho ) ,  \] where we use $\cong$ to mean equality up-to the unitary isomorphism $E$. In other words, for the most natural kind of restriction, the standard partial trace is recovered, indeed, 
\begin{align*}
(\ket{\sigma .u \cup G} \bra{\sigma ' . u \cup G'})_{| \zeta_u} & = \ket{\sigma .u}\bra{\sigma ' . u} \braket{G|G'} \\
& =   \ket{ \sigma .u}\bra{\sigma ' . u} \delta_{GG'} \\
& =  Tr_{\overline{u + \varnothing}}[(\ket{\sigma . u} \otimes \ket{G})( \bra{\sigma ' . u} \otimes \bra{G'} )  ] \\
& = Tr_{\overline{u + \varnothing}}[E(\ket{\sigma . u \cup G} \bra{\sigma ' . u \cup G ' })E^{\dagger}]
\end{align*}
The same syntactic proof method works in the special case in which one or more of the graphs is $\ket{\varnothing}$ simply by making the notational substitution $\varnothing \rightarrow \varnothing . u$
\end{example}

Since we work with name algebras to identify when pieces of some name $u$ such as $u.l$ or $u.r$ are present as parts of other names of nodes within the graph, it is often more natural to consider $\zeta_{\hat{u}}$. Here $G_{\zeta_{\hat{u}}} = \bigcup_{\sigma.v \in G:v \wedge u}\sigma.v$. Where we have used the symbol $v \wedge u$ to encode the statement ''$\exists t,t'$ such that $v.t = u . t^{'}$". In simple terms this restriction picks out any node which contains some part of $u$, allowing us to define a notion of subsystem which keeps track of the pieces of $u$ and they split, merge, and distribute across the graph over time.

\begin{example}[Statewise restriction]\label{ex:2}
Instead of picking out names, we could choose instead to pick out particular internal states, this is analogous to picking out part of a direct sum decomposition of Hilbert spaces. Formally we can define for any subset $\{ \varnothing \} \subseteq \mathcal{C} \subseteq \mathcal{G}$ the corresponding $\chi_{\mathcal{C}}$ by $G_{\chi_{\mathcal{C}}} = G$ if $G \in \mathcal{C}$ and $G_{\chi_{\mathcal{C}}} = \varnothing$ otherwise (note that this restriction is not in general extensible). Consider this time the isomorphism $E: \mathcal{H} \cong \mathcal{H}^{\mathcal{C}} \oplus \mathcal{H}^{\mathcal{G} -  \mathcal{C}}$ we can construct the map $(\bullet)_{|\chi_{\mathcal{C}}}: \mathcal{B}_1(\mathcal{H}) \rightarrow \mathcal{B}_1(\mathcal{H}^{\mathcal{C}})$ by \[ \rho_{| \chi_{\mathcal{C}}} \cong \pi_{\mathcal{C}} \rho \pi_{\mathcal{C}} + \ket{\varnothing} \bra{\varnothing} Tr[\pi_{\mathcal{G} - \mathcal{C}} \rho \pi_{\mathcal{G} - \mathcal{C}}].  \] Where $\pi_{\mathcal{C}
}$ is simply the orthogonal projection from $\mathcal{H}$ to $\mathcal{H}^{\mathcal{C}}$. Indeed, see that
\begin{align*}
& ( (\alpha \ket{G} + \beta \ket{H} ) (\alpha' \bra{G'} + \beta ' \bra{H'})  )_{| \chi_{\mathcal{C}}} \\ 
& =     \alpha \alpha' \ket{G} \bra{G'} + \beta \beta' \ket{\varnothing} \bra{\varnothing} \delta_{HH'}     \\
& = (    \pi_{\mathcal{C}} (\bullet) \pi_{\mathcal{C}} + \ket{\varnothing} \bra{\varnothing} Tr[\pi_{\mathcal{G} - \mathcal{C}} (\bullet) \pi_{\mathcal{G} - \mathcal{C}}]      )(  (\alpha \ket{G} \oplus \beta \ket{H} ) (\alpha' \bra{G'} \oplus \beta ' \bra{H'})    ) \\
& = (    \pi_{\mathcal{C}} (\bullet) \pi_{\mathcal{C}} + \ket{\varnothing} \bra{\varnothing} Tr[\pi_{\mathcal{G} - \mathcal{C}} (\bullet) \pi_{\mathcal{G} - \mathcal{C}}]      )( E (\alpha \ket{G} + \beta \ket{H} ) (\alpha' \bra{G'} + \beta ' \bra{H'}) E^{\dagger}   ) 
\end{align*}
This could be interpreted as the viewpoint upon state $\rho$, of a limited observer, who sees nothing beyond the states $\mathcal{C}$. I.e., in so far as $\rho$ lies beyond $\mathcal{C}$, the observer sees nothing, where nothing is represented by the empty graph---as opposed to the null vector $0$. Interestingly this interpretation of the generalised partial trace offers an linear trace-preserving alternative to Born rule for quantum measurements (e.g. the projective update rule for some observable $\pi_{\mathcal{C}}$).
\end{example}
\begin{example}[A pointwise restriction]\label{ex:3}
Some pointwise restrictions model mixtures of direct sums and tensor products. For instance $\zeta_{u:S}$ can be defined which keeps only node $u$ and only when $u$ is in a state in a subset $S \subseteq \Sigma$. Here $G_{\zeta_{u:S}} = \bigcup_{\sigma.v \in G:v=u \textrm{ and }\sigma \in S}\sigma.v$. We again choose to consider the constrained Hilbert space $\mathcal{H}^{\mathcal{C}}$ of graphs spanned by the set $\mathcal{C}$ of all graphs such that if $v \in V(G)$ and $v \wedge u$ then $v = u$. This time thinking in terms of the unitary isomrophism $E: \mathcal{H}^{\mathcal{C}} \cong ((\mathcal{H}^{u + \varnothing})^S \oplus (\mathcal{H}^{u + \varnothing})^{\Sigma - S}) \otimes \mathcal{H}^{\mathcal{C} - u}$ we can then write \[ \rho_{| \zeta_{u:S}} \cong \pi_{S} Tr_{\overline{u + \varnothing}}[\rho] \pi_{S} + \ket{\varnothing} \bra{\varnothing} Tr[\pi_{\Sigma - S} Tr_{\overline{u + \varnothing}}[\rho] \pi_{\Sigma-S}].   \] The proof of this identity consists in combining the techniques of the previous two examples. To unpack this in intuitive terms, the $\zeta_{u:S}$-partial trace first applies a standard partial trace to reduce to $u$ and then the viewpoint of a limited observer who sees nothing beyond $S$. 
\end{example} 
\begin{example}[Marked restriction]\label{ex:4}
Instead of picking out names, we could choose instead to just pick out particular internal states. Formally we can define for any subset $S \subseteq \Sigma$ the corresponding $\mu_{S}$ by $G_{\mu_{S}} = \bigcup_{\sigma.v \in G: \sigma \in S}\sigma.v$. In this case we can rewrite $\mathcal{H} \cong \bigoplus_{\mathcal{Q} \subseteq \mathcal{V}} \mathcal{H}^{(\neg \mathcal{Q}: S) + \varnothing } \otimes \mathcal{H}^{\mathcal{Q}:\Sigma \setminus S}$ where for each $\mathcal{Q} \subseteq \mathcal{V}$ then $\mathcal{H}^{\mathcal{Q}:\Sigma \setminus S}$ is the space of graphs with names exactly the elements of $\mathcal{Q}$ and with states in $ \Sigma \setminus S$, and $\mathcal{H}^{(\neg \mathcal{Q}: S) + \varnothing }$ is the space of graphs such that for every node in the graph, its name does not intersect $\mathcal{Q}$ and its state lies in $S$. Notice the empty $\varnothing$ belongs to that space. 
\[  \rho_{|\mu_{S}} \cong \Sigma_{\mathcal{Q}} Tr_{\overline{(\neg \mathcal{Q}: S) + \varnothing } }[ \pi_{\mathcal{Q}} \rho \pi_{\mathcal{Q}} ],   \] where we define $\pi_{Q}$ to be the projection of $\mathcal{H}$ into $\mathcal{H}^{(\neg \mathcal{Q}: S) + \varnothing } \otimes \mathcal{H}^{\mathcal{Q}:\Sigma \setminus S}$. Indeed, see that for any pair of name sets $\mathcal{Q},\mathcal{Q}'$ then (using $G_{(\mathcal{Q})}$ as notation for the graph $G$ appearing in direct sum branch $\mathcal{Q}$):
\begin{align*}
\ket{G_{(\mathcal{Q})} \cup H_{(\mathcal{Q})} } \bra{G^{'}_{(\mathcal{Q}')} \cup H^{'}_{(\mathcal{Q}')} }_{|\mu_{S}} & = \ket{G_{(\mathcal{Q})}} \bra{G^{'}_{(\mathcal{Q}')}} \delta_{HH'} \delta_{\mathcal{Q}\mathcal{Q}'} \\
& = \Sigma_{\mathcal{P}\subseteq \mathcal{V}} Tr_{\overline{(\neg \mathcal{P}: S) + \varnothing } }[ \pi_{\mathcal{P} } ((\ket{G_{(\mathcal{Q})} } \otimes \ket{ H_{(\mathcal{Q})} })( \bra{G^{'}_{(\mathcal{Q}')} } \otimes \bra{H^{'}_{(\mathcal{Q}')}})) \pi_{\mathcal{P}} ]   \\
& =  \Sigma_{\mathcal{P} \subseteq \mathcal{V}} Tr_{\overline{(\neg \mathcal{P}: S) + \varnothing } }[ \pi_{\mathcal{P}} E (\ket{G_{(\mathcal{Q})} \cup H_{(\mathcal{Q})} }  \bra{G^{'}_{(\mathcal{Q}')} \cup H^{'}_{(\mathcal{Q}')}}) E^{\dagger} \pi_{\mathcal{P}} ]. 
\end{align*}
Therefore, to implement $(\bullet)_{|\mu_{S}}$ in terms of standard partial traces and projectors, one needs to first project on the subspace where the bad states carrying systems are well-identified, before tracing out those systems. 
\end{example}
\begin{example}[Disks]
A key motivation of this paper is to be able to talk about neighbourhoods of nodes, even when the neighbourhoods of those nodes are in superposition, we have referred to such neighbourhoods as disks. To recap, given $ \zeta $ an extensible restriction, $ \zeta ^{r}$ is the extensible restriction such that $ G_{\zeta ^{r}}$ is induced by $ V( G_{\zeta })$ and its radius $r$ non-oriented neighbours according to $ E(G)$. Similarly, $\harp{\zeta}^{r}$ selects radius $r$ oriented neighbourhoods. It is unclear how to express the corresponding generalised trace (in the presence of superpositions of geometries as we have here), as a composition of standard partial traces and projections---nonetheless we can reason directly in terms of $\rho_{| \zeta^r}$. 
\end{example}
Note that since all (except the statewise restriction) of the above restrictions are extensible, they can be combined under unions (and taking neighbourhoods) to produce more elaborate ones.
\color{black}

\begin{figure}\centering
\includegraphics[width=0.8\textwidth]{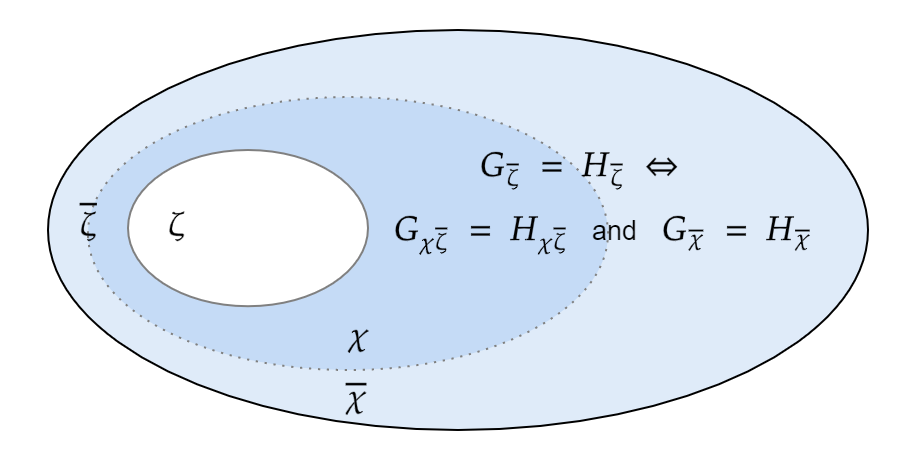}
\caption{\label{fig:comprehension} {\em Comprehension} of restrictions $\zeta \ \sqsubseteq \ \chi $ demands condition Eq. \eqref{eq:comprehension}, which states that for any $G$, $H$, equality outside the small restriction $\zeta $ (i.e. whether $G_{\overline{\zeta }} =H_{\overline{\zeta }}$) may be decomposed as both equality outside $\zeta $ but inside $\chi $ (i.e. whether $G_{\chi \overline{\zeta }} =H_{\chi \overline{\zeta }}$), and equality outside $\chi $ (i.e. whether $G_{\overline{\chi }} =H_{\overline{\chi }}$). Condition Eq. \eqref{eq:comprehension} may fail if a difference lying within $\zeta $ influences the way $\chi $ partitions the outside of $\zeta $. The condition holds in most relevant cases as shown by Prop. \ref{prop:npcomprehension}. It is needed to establish Lem. \ref{lem:tracetrace}.}
\end{figure}

Notice that restrictions are in general oblique projections, aka non-hermitian projections. They must not be confused with the partial traces $( .)_{|\chi }$ that they induce, as illustrated in Fig. \ref{fig:traceouts}.
\begin{figure}\centering
\includegraphics[width=0.9\textwidth]{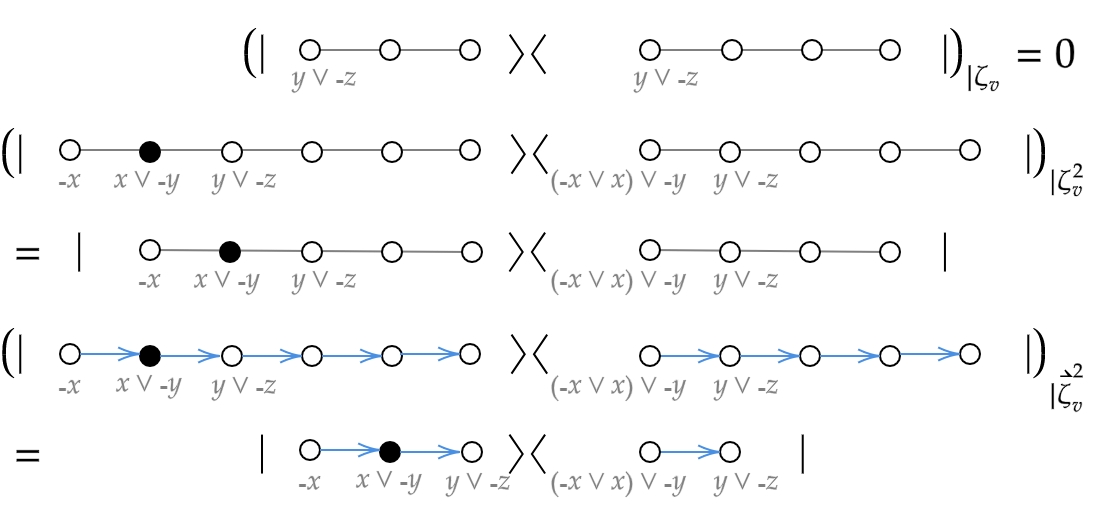}
\caption{\label{fig:traceouts}{\em Generalized partial trace.} Across figures $v:=y\lor \hyphenbullet z$, restriction $\zeta _{v}$ retains vertex $v$. Top: the ket and bra do not coincide on the complement of the neighbourhood, this goes to zero. Middle: the ket and bra coincide beyond first neighours, this goes to the restriction of the ket and bra on first neighbours. Below: with oriented edges, the neighbours are those which can signal to $v$. Overall: the question of what to do with edges of the frontier zone does not arise, as edges are derived information.}
\end{figure}

Now, the tensor product corresponding to a restriction $\chi$ works by weaving a restricted graph $\ket{G_{\chi }}$ and its complement $\ket{G_{\overline{\chi }}}$ back together, as illustrated in Fig. \ref{fig:tensors} and formalized as follows.

\begin{definition}[Tensor, consistency]\label{def:tensors}

Every restriction $\chi$ induces a tensor:

$ \ket{L} \tensorchi \ket{R} :=\begin{cases}
\ket{G} & \text{when } L=G_{\chi } ,R=G_{\overline{\chi }} \text{ for some }G\in{\mathcal{G}}  \\
0 & \text{otherwise}
\end{cases}$\\
When working over constrained configurations, $G$ needs belong to $\mathcal{H}^\mathcal{C}$ in the above, if not we return the null vector.\\
$ \ket{\psi } \tensorchi \ket{\psi '}$ is defined from the above, by bilinear extension.\\
$ \ket{G}\bra{H} \tensorchi \ket{G'}\bra{H'} :=\left(\ket{G} \tensorchi \ket{G'}\right)\left(\bra{H} \tensorchi \bra{H'}\right)$.\\
For any two operators $ A,B$, we define $ A\tensorchi B$ from the above by bilinear extension.\\
$ \ket{\psi } ,\ket{\psi '}$ are $ \chi $-consistent if and only if $ \braket{G| \psi }\braket{G'|\psi '} \neq 0$ implies $ \ket{G} \tensorchi \ket{G'} \neq 0$.\\ 
$ \rho ,\sigma $ are $ \chi $-consistent if and only if $ \rho _{GH} \sigma _{G'H'} \neq 0$ implies $ \ket{G} \tensorchi \ket{G'} \neq 0\neq \ket{H} \tensorchi \ket{H'}$, where $ \rho _{GH} :=\bra{G} \rho \ket{H}$.\\
$ A$ is $ \chi $-consistent-preserving if and only if $ \bra{H} A\ket{G_{\chi }} \neq 0$ entails $ \ket{H} \tensorchi \ket{G_{\overline{\chi }}} \neq 0$, and $ \bra{H} A^{\dagger }\ket{G_{\chi }} \neq 0$ entails $ \ket{H} \tensorchi \ket{G_{\overline{\chi }}} \neq 0$. 
\end{definition}
Notice that the definition is unambiguous, as $G:=H\cup H'$.

\begin{figure}\centering
\includegraphics[width=\textwidth]{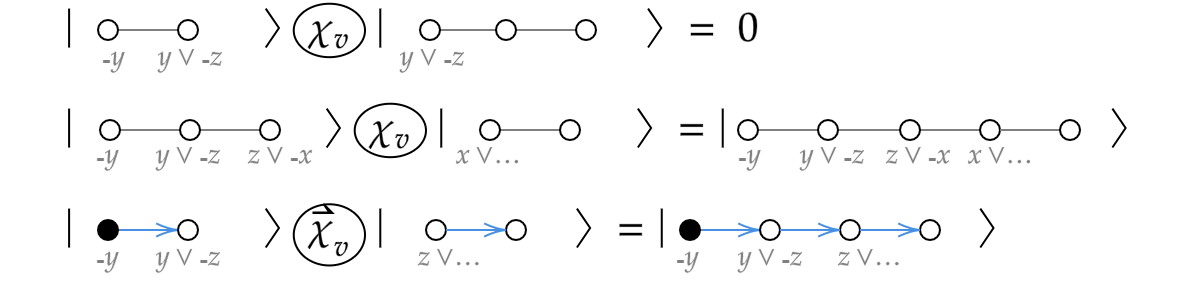}
\caption{\label{fig:tensors} {\em Generalized tensor product.} Here, $\chi_{v} :=\zeta _{v}^{1}$ and $\harp{\chi }_{v} := \harp{\zeta }_{v}^{1}$ Top: the two graphs do not correspond to a disk around $v$ and its complement, this goes to zero. Middle: the two graphs correspond to a disk and its complement. Moreover, an edge between them can be derived from names. Below: Same with oriented edges.}
\end{figure}

To get an intuition for the strictness of this notion of tensor product, consider three disjoint non-empty graphs $G,\ M,\ H$ such that $G\cup M\cup H$ is defined. With the above definition, $\ket{G\cup M} \tensorchi \ket{M\cup H} =0$, whatever the $\chi$. This may seem unnecessarily strict; a more permissive alternative would have been to let $\ket{G\cup M} \tensorchi \ket{M\cup H} =\ket{G\cup M\cup H}$. This, however, would entail $A\tensorchi I=I$ with $I$ the identity operator, which we will find we do not want (cf. Prop. \ref{prop:gatelocality}). 

Note that for our simpler examples of restrictions then $\tensorchi$ can often be directly computed.
\begin{itemize}
    \item In the case of Example \ref{ex:1}, i.e. $\zeta_u$ one can concretely say that \[A \tensorzeta_{\!u} I \cong A_{u + \varnothing} \otimes I_{\mathcal{G}-u}, \] where we have again decomposed $\mathcal{H}^\mathcal{C}$ as $\mathcal{H}^\mathcal{C} \cong (\mathcal{H}^{u + \varnothing}) \otimes \mathcal{H}^{\mathcal{C}-u}$ where $A_{u + \varnothing} \cong \pi_{u + \varnothing} A \pi_{u + \varnothing}$ with $\pi_{u + \varnothing}$ the projection onto the space $\mathcal{H}^{u + \varnothing}$. Indeed,  
    \begin{align*}
     \bra{\sigma . u \cup G} A \tensorzeta_{\!u} I \ket{\sigma ' . u \cup G'} & =  \bra{\sigma . u } A \ket{\sigma ' . u} \braket{G|G'}   \\
    & =  \bra{\sigma . u } A_{u + \varnothing} \ket{\sigma ' . u} \delta_{GG'} \\
    & =( \bra{\sigma . u } \otimes \bra{G} )(A_{u + \varnothing}  \otimes I ) (\ket{\sigma ' . u} \otimes \ket{G'}) \\
    & = ( \bra{\sigma . u \cup G} )E(A_{u + \varnothing}  \otimes I )E^{\dagger} (\ket{\sigma ' . u \cup G'})  .
    \end{align*} An identical result holds for the restriction which picks out all names which intersect with $u$ (rather than just $u$).
        \item In the case of Example \ref{ex:2}, i.e. $\chi_{\mathcal{C}}$ and using the isomorphism $E : \mathcal{H} \cong \mathcal{H}^{\mathcal{C}} \oplus \mathcal{H}^{\mathcal{G}-\mathcal{C}}$ we can write \[ A \tensorchi I \cong A_{\mathcal{C}} \oplus A_{\varnothing} I_{\mathcal{G}-\mathcal{C}}.  \] 
where $A_\varnothing$ stands $\bra{\varnothing}A\ket{\varnothing}$. Notice that $A \tensorchi_{\!\mathcal{C}} I$ does not exactly decompose as $A \oplus I$. This has a natural interpretation: as our limited agent sees nothing beyond $\mathcal{C}$, it must treat it just like it treats the empty graph. Indeed:
        \begin{align*}
     &   (\alpha \bra{G} + \beta \bra{H}) A \tensorchi_{\!\mathcal{C}} I (\alpha ' \ket{G'} + \beta' \ket{H'}) \\
     & = \alpha \alpha' \bra{G} A \ket{G'} + 0 + 0 + \beta \beta' \bra{\varnothing}A \ket{\varnothing} \delta_{HH'} \\
        & = (\alpha \bra{G} + \beta \bra{H}) (A \oplus \bra{\varnothing}A \ket{\varnothing} I) (\alpha ' \ket{G'} + \beta' \ket{H'}) \\
        & = (\alpha \bra{G} + \beta \bra{H})E( A \oplus A_{\varnothing} I) E^{\dagger} (\alpha ' \ket{G'} + \beta' \ket{H'})
        \end{align*}
    \item In the case of Example \ref{ex:3}, i.e. $\chi_{u:S}$ then we can use similar observations to the above primitive cases to see that \[ A \tensorchi_{\!u:S} I  = (A_{S+\varnothing} \oplus A_{\varnothing} I_{\Sigma - S}) \otimes I_{\mathcal{C} - u} ,  \] where we have used the decomposition $\mathcal{H}^{\mathcal{C}} \cong (\mathcal{H}^{S + \varnothing} \oplus \mathcal{H}^{\Sigma - S}) \otimes \mathcal{H}^{\mathcal{C} -u}$, and, we define $A_{S + \varnothing} = \pi_{S + \varnothing}A  \pi_{S + \varnothing}$. 
    \item In the case of Example \ref{ex:4}, $\mu_{S}$ we have that \[  A \tensormu_{\!S} I = \bigoplus_{\mathcal{Q}\subseteq\mathcal{V}} \pi_{ \neg \mathcal{Q} :S}^{\dagger} A \pi_{ \neg \mathcal{Q}: S}  \otimes  I_{\mathcal{Q}: \neg S},  \] where $\pi_{\neg \mathcal{Q}:S}$ is the composition $\pi_{S} \pi_{\mathcal{Q}}$ with $\pi_{\mathcal{Q}}$ defined as before and $\pi_{S}$ the  projection into states which have names non-intersecting with $\mathcal{Q}$ and states in $S$ (along as always with the empty graph). Indeed, let us show this directly:
    \begin{align*}
    & \bra{G_{(\mathcal{P})} \cup H_{(\mathcal{P})}} A \tensormu_{\!S} I \ket{G^{'}_{(\mathcal{P}')} \cup H^{'}_{(\mathcal{P}')}} \\ 
    & = \bra{G_{(\mathcal{P})}} A \ket{G^{'}_{(\mathcal{P}')}} \delta_{HH'} \delta_{\mathcal{P}\mathcal{P}'} \\
    & =  \bra{G_{(\mathcal{P})}} \otimes \bra{H_{(\mathcal{P})}}( \bigoplus_{\mathcal{Q}\subseteq\mathcal{V}} \pi_{ \neg \mathcal{Q} :S}^{\dagger} A \pi_{ \neg \mathcal{Q}: S}  \otimes  I_{\mathcal{Q}: \neg S}) \ket{G^{'}_{(\mathcal{P}')}} \otimes \ket{H^{'}_{(\mathcal{P}')}} \\
    & = \bra{G_{(\mathcal{P})} \cup H_{(\mathcal{P})}} E ( \bigoplus_{\mathcal{Q}\subseteq\mathcal{V}} \pi_{ \neg \mathcal{Q} :S}^{\dagger} A \pi_{ \neg \mathcal{Q}: S}  \otimes  I_{\mathcal{Q}: \neg S}) E^{\dagger} \ket{G^{'}_{(\mathcal{P}')} \cup H^{'}_{(\mathcal{P}')}}.
    \end{align*}
\end{itemize}
\color{black}

These generalized partial traces and tensors are powerful tools, but working with them sometimes feels like a step into the unknown. Our old intuitions about traceouts and tensors guide us, but sometimes they mislead us. We have as a result had to check the conditions of applications of several basic facts about the way these tensor products and tracing operators interact with one another, leading to the Toolbox of table \ref{tab:toolbox}.\\{\em Again a quick note on notations. Throughout the paper, greek symbols $\rho ,\sigma $ represent elements of the space $\mathcal{B}_1 (\mathcal{H})$ of trace class operators (e.g. states), capital letters $A,B$ represent bounded operators (e.g. transformations and observables.}

\begin{table}
\caption{\label{tab:toolbox} {\em Mathematical toolbox}.}  
\renewcommand{\arraystretch}{2}        
\begin{tabular}{p{0.22\textwidth}p{0.78\textwidth}}
\hline 
  Lem. \ref{lem:tensorbracket} & $ \braket{H|G} =\braket{H_{\chi }|G_{\chi }}\braket{H_{\overline{\chi }}|G_{\overline{\chi }}}$  \\
\hline 
  Lem. \ref{lem:restrictions}. & $ \chi \chi =\chi $ \qquad $ \chi \overline{\chi } =\emptyset $  \\
\hline 
 ~ & $ \left( \rho \ket{G}\bra{H}\right)_{|\emptyset } =\bra{H} \rho \ket{G}$ \qquad $ ( \rho A)_{|\emptyset } =( A\rho )_{|\emptyset }$ \qquad $ ( \alpha \rho )_{|\chi } =\alpha ( \rho_{|\chi })$  \\
\hline 
~ & If $ \ket{G} \tensorchi \ket{G'} \neq 0$\\  & then $ \ket{G} \tensorchi \ket{G'} \ =\ \ket{G\cup G'}$ \quad and \quad $ \left(\ket{G} \tensorchi \ket{G'}\right)_{\chi } =\ket{G}$  \\
\hline 
~ & $ \ \rho _{|\chi } \ =\ \sum _{G,H\in \mathcal{G} ,\ G_{\overline{\chi }} =H_{\overline{\chi }}} \rho _{G H} \ \ket{G_{\chi }}\bra{H_{\chi }} \ $  \\
\hline 
Lem. \ref{lem:tracetrace} & If $ \zeta \sqsubseteq \chi $, $ \begin{aligned}
( \rho _{|\chi })_{|\zeta } & =\rho _{|\zeta }
\end{aligned}$ and \quad $A$ $ \zeta $-local is $ \chi $-local. \\
\hline 
 Lem. \ref{lem:tensor} & $ \ A\tensorchi B\ =\ \sum _{G,H\in \mathcal{G}} A_{G_{\chi } H_{\chi }} B_{G_{\overline{\chi }} H_{\overline{\chi }}} \ \ket{G}\bra{H} \ $ \qquad $ A\tensorchi I =A\tensorchi I_{\overline{\chi }}$  \\
\hline 

 Lem. \ref{lem:tensortensor} & If $ [ \chi ,\zeta ] =[\overline{\chi } ,\zeta ] =[ \chi ,\overline{\zeta }] =[\overline{\chi } ,\overline{\zeta }] =0$,\\ & then $ ( A \tensorzeta   B) \tensorchi ( C \tensorzeta   D) =( A\tensorchi C)  \tensorzeta   ( B\tensorchi D) \ $  \\
\hline 
 Lem. \ref{lem:tensortrace1} & If $ \rho ,\sigma $ $ \chi $-consistent, \ $ ( \rho \tensorchi \sigma )_{|\chi } =\rho \ \sigma _{|\emptyset } \ $\\ & If $ \zeta \sqsubseteq \chi $, $ \rho ,\sigma $ $ \chi $-consistent, \ $ ( \rho \tensorchi \sigma )_{|\zeta } =\rho _{|\zeta } \ \sigma _{|\emptyset } \ $  \\
\hline 
 Lem. \ref{lem:tensortrace2} & If $ [ \chi ,\zeta ] =[\overline{\chi } ,\zeta ] =[ \chi ,\overline{\zeta }] =[\overline{\chi } ,\overline{\zeta }] =0$, \ and $ \rho ,\sigma $ $ \chi $-consistent,\\ & $ ( \rho \tensorchi \sigma )_{|\zeta } =\rho _{|\zeta } \tensorchi \sigma _{|\zeta } \ $  \\
\hline 
 Lem. \ref{lem:interchangelaws} &  $ ( A\tensorchi I)\ket{G} =A\ket{G_{\chi }} \tensorchi \ket{G_{\overline{\chi }}}$.\\ & If $ A$, $ A'$, $ B$, $ B'$ are $ \chi $-consistent-preserving,$ ( A'\tensorchi B')( A\tensorchi B) =\ A'A\tensorchi B'B$.  \\
 \hline
\end{tabular}
\end{table}

A good rule of thumb is that usual intuitions about $A\tensorchi B$ will carry through provided that $\chi $-consistency conditions are met. In fact much of the attention in the proofs is spent keeping track of which terms get zeroed by the $\tensorchi$. 

Another good rule of thumb is that our usual intuitions about subsystems $\zeta $ of a wider system $\chi $ will carry through, provided that the comprehension condition given by Eq. \eqref{eq:comprehension} is met, which is true of most natural cases thanks to name-preservation, see Prop. \ref{prop:npcomprehension}. 

Sometimes the two rules of thumb interact, e.g. it is the name-preservation assumption that helps meet $\chi $-consistency, as in Prop. \ref{prop:unitaryextension}.

\subsection*{Properties of traceouts over quantum networks}

An early attempt to define a (non-modular) partial trace for quantum causal graph dynamics actually failed to exhibit positivity-preservation, i.e. there exists $\rho $ non-negative with $\rho _{|\chi }$ not non-negative \cite{ArrighiQCGD}. Here we show that partial traces are actually positive-preserving. In fact we check they are completely-positive-preserving, meaning that they remain positive-preserving when tensored with the identity, as required for general quantum operations. We do the same for trace-preservation and name-preservation. 

We denote by $~_{|\chi}   \tensorzeta   I$ the map that is the partial trace $~_{|\chi}$ tensored with the identity by means of some arbitrary generalised tensor product $\tensorzeta$, i.e. the linear map $\rho \mapsto (~_{|\chi}   \tensorzeta   I)(\rho)$ which is such that whenever $\rho=\sigma \tensorzeta\sigma'$ then $(~_{|\chi}   \tensorzeta   I)(\rho)=\sigma_{|\chi }  \tensorzeta   \sigma'$.

\begin{proposition}[Traceouts positivity-preservation, trace-preservation, name-preservation]\label{prop:traceouts}

The map \ $ \rho \mapsto (~_{|\chi}   \tensorzeta   I)( \rho )$ over ${\cal B}_1({\cal H})$ is completely positive-preserving and name-preservation preserving.

If moreover $ \ket{G_{\zeta \chi }}  \tensorzeta   \ket{G_{\overline{\zeta }}} \neq 0$, then the same map is trace-preserving.
\end{proposition}
\begin{proof} 
[Positivity preservation]

A trace class operator $\rho $ is a compact operator, hence it is non-negative if and only if it has the form $\sum _{i}\ket{\psi ^{i}}\bra{\psi ^{i}}$.
\begin{align*}
\textrm{Consider some }\ket{\psi } & =\sum _{G,\ G'\ \in \ \mathcal{G}} \alpha _{G G'} \ \ket{G} \tensorchi \ket{G'}\\
\bra{\psi } & =\sum _{H,\ H'\ \in \ \mathcal{G}} \alpha _{HH'}^{*} \ \bra{H} \tensorchi \bra{H'}\\
\ket{\psi }\bra{\psi } & =\sum _{G,\ G',\ H,\ H'\ \in \ \mathcal{G}} \alpha _{G G'} \alpha _{HH'}^{*} \ \ket{G}\bra{H} \tensorchi \ket{G'}\bra{H'}\\
\left(\ket{\psi }\bra{\psi }\right)_{|\chi } & =\sum _{ \begin{array}{l}
G,\ G',\ H,\ H'\ \in \ \mathcal{G}\\
\ket{G} \tensorchi \ket{G'} \ \neq \ 0\\
\ket{H} \tensorchi \ket{H'} \ \neq \ 0
\end{array}} \alpha _{G G'} \alpha _{HH'}^{*} \ \ket{G}\bra{H} \ \braket{H'|G'}\\
 & =\sum _{ \begin{array}{l}
G,\ H,\ K\ \in \ \mathcal{G}\\
\ket{G} \tensorchi \ket{K} \ \neq \ 0\\
\ket{H} \tensorchi \ket{K} \ \neq \ 0
\end{array}} \alpha _{GK} \alpha _{HK}^{*} \ \ket{G}\bra{H}\\
 & =\sum _{K\ \in \ \mathcal{G}}\left(\sum _{ \begin{array}{l}
G\ \in \ \mathcal{G}\\
\ket{G} \tensorchi \ket{K} \ \neq \ 0
\end{array}} \alpha _{GK} \ \ket{G}\right)\left(\sum _{ \begin{array}{l}
H\ \in \ \mathcal{G}\\
\ket{H} \tensorchi \ket{K} \ \neq \ 0
\end{array}} \alpha _{HK}^{*} \ \bra{H}\right)\\
 & =\sum _{K\ \in \ \mathcal{G}}\ket{\phi ^{K}}\bra{\phi ^{K}}\\
\rho _{|\chi } & =\left(\sum _{i}\ket{\psi ^{i}}\bra{\psi ^{i}}\right)_{|\chi }\\
 & =\sum _{i,\ K\ \in \ \mathcal{G}}\ket{\phi ^{i,K}}\bra{\phi ^{i,K}}
\end{align*}

[Complete positivity preservation]

\begin{align*}
\ket{G}\bra{H} & =\ket{G_{\zeta }}\bra{H_{\zeta }}  \tensorzeta   \ket{G_{\overline{\zeta }}}\bra{H_{\overline{\zeta }}}\\
(~_{|\chi}   \tensorzeta   I) \ket{G}\bra{H} & =\ket{G_{\zeta \chi }}\bra{H_{\zeta \chi }}\braket{H_{\zeta \overline{\chi }}|G_{\zeta \overline{\chi }}}  \tensorzeta   \ket{G_{\overline{\zeta }}}\bra{H_{\overline{\zeta }}}
\end{align*}

Let $\alpha '_{G_{\chi } KG'} :=\begin{cases}
\alpha _{GG'} & \text{if} \ \ket{G_{\chi }} \tensorchi \ket{K} \ =\ \ket{G}\\
0 & \text{otherwise}
\end{cases}$ .
\begin{align*}
\textrm{Consider some }\ket{\psi } & =\sum _{G,\ G'\ \in \ \mathcal{G}} \alpha _{G G'} \ \ket{G}  \tensorzeta   \ket{G'}\\
\bra{\psi } & =\sum _{H,\ H'\ \in \ \mathcal{G}} \alpha _{HH'}^{*} \ \bra{H}  \tensorzeta   \bra{H'}\\
\ket{\psi }\bra{\psi } & =\sum _{G,\ G',\ H,\ H'\ \in \ \mathcal{G}} \alpha _{G G'} \alpha _{HH'}^{*} \ \ket{G}\bra{H}  \tensorzeta   \ket{G'}\bra{H'}\\
(~_{|\chi}   \tensorzeta   I)\left(\ket{\psi }\bra{\psi }\right) & =\sum _{ \begin{array}{l}
G,\ G',\ H,\ H'\ \in \ \mathcal{G}\\
\ket{G}  \tensorzeta   \ket{G'} \ \neq \ 0\\
\ket{H}  \tensorzeta   \ket{H'} \ \neq \ 0
\end{array}} \alpha _{G G'} \alpha _{HH'}^{*} \ \ket{G_{\chi }}\bra{H_{\chi }}\braket{H_{\overline{\chi }}|G_{\overline{\chi }}}  \tensorzeta   \ket{G'}\bra{H'}\\
 & =\sum _{ \begin{array}{l}
G ,\ H ,\ G',\ H'\ \in \ \mathcal{G}\\
G_{\overline{\chi }} =H_{\overline{\chi }}\\
\ket{G}  \tensorzeta   \ket{G'} \ \neq \ 0\\
\ket{H}  \tensorzeta   \ket{H'} \ \neq \ 0
\end{array}} \alpha _{G G'} \alpha _{HH'}^{*} \ \ket{G_{\chi }}\bra{H_{\chi }}  \tensorzeta   \ket{G'}\bra{H'}\\
 & =\sum _{ \begin{array}{l}
G_{\chi } ,\ H_{\chi } ,\ K,\ G',\ H'\ \in \ \mathcal{G}\\
\left(\ket{G_{\chi }} \tensorchi \ket{K}\right)  \tensorzeta   \ket{G'} \ \neq \ 0\\
\left(\ket{H_{\chi }} \tensorchi \ket{K}\right)  \tensorzeta   \ket{H'} \ \neq \ 0
\end{array}} \alpha '_{G_{\chi } KG'} \alpha ^{\prime *}_{H_{\chi } KH'} \ \ket{G_{\chi }}\bra{H_{\chi }}  \tensorzeta   \ket{G'}\bra{H'}\\
\ket{\phi ^{K}} & :=\sum _{ \begin{array}{l}
G_{\chi } ,\ G'\ \in \ \mathcal{G}\\
\left(\ket{G_{\chi }} \tensorchi \ket{K}\right)  \tensorzeta   \ket{G'} \ \neq \ 0
\end{array}} \alpha '_{G_{\chi } KG'}\ket{G_{\chi }}  \tensorzeta   \ket{G'}\\
(~_{|\chi}   \tensorzeta   I)\left(\ket{\psi }\bra{\psi }\right) & =\sum _{K\ \in \ \mathcal{G}}\ket{\phi ^{K}}\bra{\phi ^{K}}\\
(~_{|\chi}   \tensorzeta   I)\left(\sum _{i}\ket{\psi ^{i}}\bra{\psi ^{i}}\right) & =\sum _{K\ \in \ \mathcal{G} ,\ i}\ket{\phi ^{i,K}}\bra{\phi ^{i,K}}
\end{align*}

[Name-preservation preservation]

An operator $\rho $ is name-preserving if and only if it is a sum of terms of the form $\ket{G}\bra{H}$ with $V( G) \corresponds V( H)$.
\begin{align*}
\ket{G} & =\ \ket{G_{\chi }} \tensorchi \ket{G_{\overline{\chi }}}\\
\bra{H} & =\bra{H_{\chi }} \tensorchi \bra{H_{\overline{\chi }}}\\
\left(\ket{G}\bra{H}\right)_{|\chi } & =\ket{G_{\chi }}\bra{H_{\chi }} \ \braket{H_{\overline{\chi }}|G_{\overline{\chi }}}
\end{align*}
When this is non-zero, $V( G_{\overline{\chi }}) =V( H_{\overline{\chi }})$.\\ 
Then by Lem. \ref{lem:complementnames}, $V( G_{\chi }) =V( G) \setminus V( G_{\overline{\chi }}) \corresponds V( H) \setminus V( H_{\overline{\chi }}) =V( H_{\chi })$.\\
So, $\rho _{|\chi }$ is a sum of terms of the form $\ket{G_{\chi }}\bra{H_{\chi }}$ with $V( G_{\chi }) \corresponds V( H_{\chi })$.

[Complete name-preservation preservation]

An operator $\rho $ is name-preserving if and only if it is a sum of terms of the form $\ket{G}\bra{H}$ with $V( G) \corresponds V( H)$.
\begin{align*}
\ket{G}\bra{H} & =\ket{G_{\zeta }}\bra{H_{\zeta }}  \tensorzeta   \ket{G_{\overline{\zeta }}}\bra{H_{\overline{\zeta }}}\\
(~_{|\chi}   \tensorzeta   I)\ \ket{G}\bra{H} & =\ket{G_{\zeta \chi }}\bra{H_{\zeta \chi }}\braket{H_{\zeta \overline{\chi }}|G_{\zeta \overline{\chi }}}  \tensorzeta   \ket{G_{\overline{\zeta }}}\bra{H_{\overline{\zeta }}}\\
 & =\braket{H_{\zeta \overline{\chi }}|G_{\zeta \overline{\chi }}}\left(\ket{G_{\zeta \chi }}  \tensorzeta   \ket{G_{\overline{\zeta }}}\right)\left(\bra{H_{\zeta \chi }}  \tensorzeta   \bra{H_{\overline{\zeta }}}\right)
\end{align*}
When this is non-zero, $V( G_{\zeta \overline{\chi }}) =V( H_{\zeta \overline{\chi }})$. Let $\begin{aligned}
\ket{G'} & :=\ket{G_{\zeta \chi }}  \tensorzeta   \ket{G_{\overline{\zeta }}}
\end{aligned}$ and $\begin{aligned}
\ket{H'} & :=\ket{H_{\zeta \chi }}  \tensorzeta   \ket{H_{\overline{\zeta }}}
\end{aligned}$.\\
Then by Lem. \ref{lem:complementnames}, $V( G') =V( G) \setminus V( G_{\zeta \overline{\chi }}) \corresponds V( H) \setminus V( H_{\zeta \overline{\chi }}) =V( H')$.\\
So, $(~_{|\chi}   \tensorzeta   I)( \rho )$ is a sum of terms of the form $\ket{G'}\bra{H'}$ with $V( G') \corresponds V( H')$.

[Trace Preservation]

Notice that $\overline{\zeta \chi \cup \overline{\zeta }} =\zeta \overline{\chi }$.
\begin{align*}
(( .)_{\chi }  \tensorzeta   I)\left(\ket{G}\bra{H}\right) & =\ \left(\ket{G_{\zeta }}\bra{H_{\zeta }}\right)_{|\chi }   \tensorzeta         \ \ket{G_{\overline{\zeta }}}\bra{H_{\overline{\zeta }}}\\
 & =\ \left(\ket{G_{\zeta \chi }}\bra{H_{\zeta \chi }}  \tensorzeta   \ket{G_{\overline{\zeta }}}\bra{H_{\overline{\zeta }}}\right)\braket{H_{\zeta \overline{\chi }}|G_{\zeta \overline{\chi }}}\\
 & =\ket{G_{\zeta \chi \cup \overline{\zeta }}}\bra{H_{\zeta \chi \cup \overline{\zeta }}}\braket{H_{\zeta \overline{\chi }}|G_{\zeta \overline{\chi }}}\\
\left((( .)_{\chi }  \tensorzeta   I)\left(\ket{G}\bra{H}\right)\right)_{|\emptyset } & =\braket{H_{\zeta \chi \cup \overline{\zeta }}| G_{\zeta \chi \cup \overline{\zeta }}}\braket{H_{\zeta \overline{\chi }}|G_{\zeta \overline{\chi }}}\\
\text{By Lem. \ref{lem:tensorbracket}} & =\braket{H|G}
\end{align*}
\end{proof}

\begin{figure}\centering
\includegraphics[width=\textwidth]{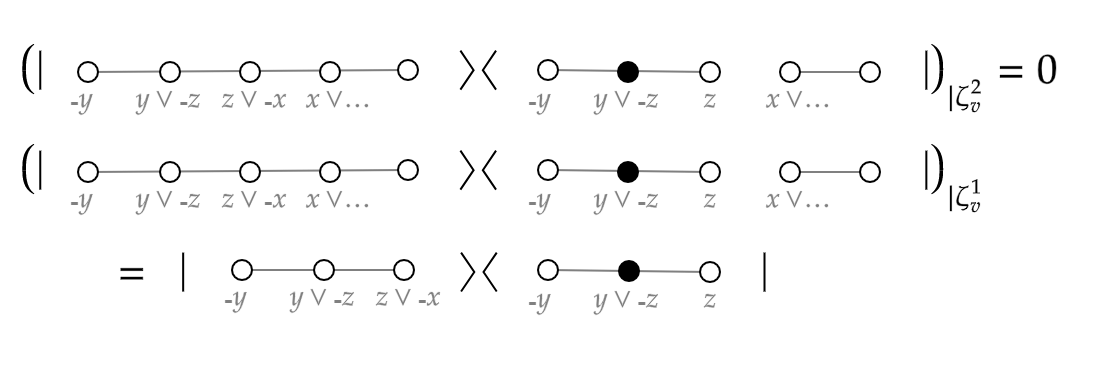}
\caption{\label{fig:incomprehension} When a ``wider'' traceout decoheres an ``inner'' traceout.}
\end{figure}

Even though, $\zeta _{v}^{2} \zeta _{v}^{1} =\zeta _{v}$, i.e. $\zeta _{v}^{1}$ is included in $\zeta _{v}^{2}$ in a natural sense, the comprehension condition Eq. \eqref{eq:comprehension} does not hold in general and thus $\zeta _{v}^{1} \ \not\sqsubseteq \ \zeta _{v}^{2}$. As a consequence, Lem. \ref{lem:tracetrace} does not apply and it is not the case that $\rho _{|\zeta _{v}^{1}} =( \rho _{|\zeta _{v}^{2}})_{|\zeta _{v}^{1}}$, as shown in Fig. \ref{fig:incomprehension}, where we see that $\rho _{|\zeta _{v}^{2}}$ decoheres certain superpositions of names whilst $\rho _{|\zeta _{v}^{1}}$ does not. This counter-intuitive behaviour disappears over name-preserving states.

\begin{proposition}[Name-preservation and comprehension]\label{prop:npcomprehension}

Consider $ \zeta $ an extensible restriction such that $ \chi :=\zeta ^{r}$ verifies $ \chi \zeta =\zeta $.\\
When $ V( G) \corresponds V( H)$, we have $ \braket{H_{\overline{\zeta }}|G_{\overline{\zeta }}} =\braket{H_{\chi \overline{\zeta }}|G_{\chi \overline{\zeta }}}\braket{H_{\overline{\chi }}|G_{\overline{\chi }}}.$\\
Hence, over name-preserving superselected states and operators, $ \zeta \sqsubseteq \chi $.
\end{proposition}
\begin{proof}
The LHS and RHS of Eq. \eqref{eq:comprehension} can either be $0$ or $1$.

[RHS=1 $\Rightarrow $LHS=1]

RHS=1 implies $H_{\chi \overline{\zeta }} =G_{\chi \overline{\zeta }}$ and $H_{\overline{\chi }} =G_{\overline{\chi }}$. Thus, $H_{\chi \overline{\zeta } \cup \overline{\chi }} =H_{\chi \overline{\zeta }} \cup H_{\overline{\chi }} =G_{\chi \overline{\zeta }} \cup G_{\overline{\chi }} =G_{\chi \overline{\zeta } \cup \overline{\chi }}$. 

But $\overline{\zeta } =\overline{\chi \zeta } =\chi \overline{\zeta } \cup \overline{\chi }$. So, $H_{\overline{\zeta }} =G_{\overline{\zeta }}$, i.e. LHS=1.

[LHS=1 $\Rightarrow $RHS=1]

LHS=1 implies $H_{\overline{\zeta }} =G_{\overline{\zeta }} =K$. \\
Combined with name-preservation, $V( H_{\zeta }) =V( H) \setminus V( H_{\overline{\zeta }}) \corresponds V( G) \setminus V( G_{\overline{\zeta }}) =V( G_{\zeta })$. \\
Thus, $H_{\zeta }$ and $G_{\zeta }$ have the same $r$-neighbours in $K$, namely $H_{\chi \overline{\zeta }} =G_{\chi \overline{\zeta }}$, and the same complement to the $r$-neighbours, namely $H_{\overline{\chi }} =G_{\overline{\chi }}$. Hence RHS=1.
\end{proof}

\section{Local operators over quantum networks}\label{sec:locality}

\begin{figure}\centering
\includegraphics[width=\textwidth]{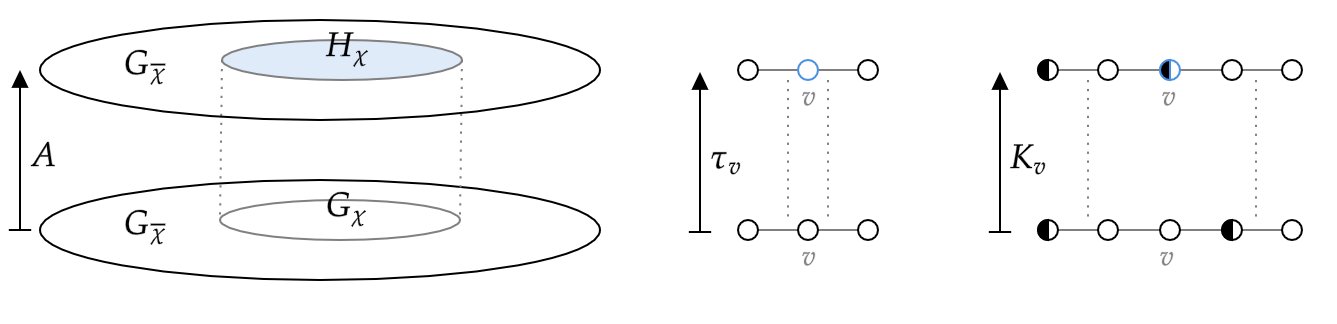}
\caption{\label{fig:locality}{\em Local operators.} Left: $A\ \chi $-local will only modify $G_{\chi }$. Middle: $\zeta _{v}$-local operator $\tau _{v}$ just toggles a 'black/blue' bit inside the system at $v$. Right: $\chi _{v}$-local operator $K_{v}$ is a reversible update rule which computes the future state of the system at $v$ (according to $M$, see Fig. \ref{fig:causality}) and toggles it out, whilst attempting to leave the rest mostly unchanged, cf. Th. \ref{th:blockdecomposition}.}
\end{figure}

Since a restriction $\chi$ isolates a part of each possible graph, one can introduce the notion of a $\chi $-local operator, one that only acts on the restriction $\chi $, leaving its complement $\overline{\chi }$ unchanged. I.e. a $\chi $-local operator acts only on the left of $ \tensorchi $.

\begin{definition}[Locality]\label{def:locality}

$A$ is $\chi $-local if and only if 
\begin{equation}
\bra{H} A\ket{G} =\bra{H_{\chi }} A\ket{G_{\chi }}\braket{H_{\overline{\chi }}|G_{\overline{\chi }}}\label{eq:locality}
\end{equation}
$A$ is strictly $\chi $-local if, moreover, $A^{\dagger } A$ and $AA{^{\dagger }}^{\ }$are $\chi $-local.\\ 
In particular, every unitary $\chi $-local is strictly $\chi $-local.
\end{definition}
{\em Soundness.} [Unitary case] Suppose $U$ is unitary $\chi $-local. Then $U^{\dagger } U=U^{\dagger } U=I$, which is $\chi $-local by Lem. \ref{lem:tensorbracket}.\hfill$\qed$\\

\color{black}

The standard way to state the locality of $A$ is to write it as a facotrisation of the form $A=B\otimes I$. Here follows a generalization of this statement, a key point that allows for this generalisation is that tensoring with the identity zeroes out the non-local terms of $B$.

\begin{proposition}[Operational locality]\label{prop:gatelocality}
$A$ is $\chi $-local if and only if $A=A\tensorchi I$.\\
For all $ B$, even if it is non-$\chi $-local, $B\tensorchi I$ is $\chi $-local.\\
Moreover, if $ B$ is n.-p., so is $B\tensorchi I$. 
\end{proposition}

\begin{proof}

[Preliminary]

\begin{align*}
A\tensorchi I & =\left(\sum _{G ,H \in \mathcal{G}}\bra{H} A\ket{G}\ket{H}\bra{G}\right) \tensorchi \left(\sum _{K \in \mathcal{G}}\ket{K}\bra{K}\right)\\
 & =\left(\sum _{G,H \in \mathcal{G}}\bra{H} A\ket{G}\ket{H}\bra{G}\right) \tensorchi \left(\sum _{G' ,H' \in \mathcal{G}}\braket{H'|G'}\ket{H'}\bra{G'}\right)\\
 & =\sum _{ \begin{array}{l}
G,H \in \mathcal{G}\\
G',H' \in \mathcal{G}
\end{array}}\bra{H} A\ket{G}\braket{H'|G'}\left(\ket{H}\bra{G} \tensorchi \ket{H'}\bra{G'}\right)\\
\text{By Lem. \ref{lem:tensor}} & =\sum _{G,H\in \mathcal{G}}\bra{H_{\chi }} A\ket{G_{\chi }}\braket{H_{\overline{\chi }}|G_{\overline{\chi }}}\ket{H}\bra{G}
\end{align*}

[First part]

\begin{align*}
A\ \chi \text{-local} & \Leftrightarrow \bra{H} A\ket{G} =\bra{H_{\chi }} A\ket{G_{\chi }}\braket{H_{\overline{\chi }}|G_{\overline{\chi }}}\\
 & \Leftrightarrow A=\sum _{G,H\in \mathcal{G}}\bra{H_{\chi }} A\ket{G_{\chi }}\braket{H_{\overline{\chi }}|G_{\overline{\chi }}}\ket{H}\bra{G} \ =\ A\tensorchi I
\end{align*}

[Second part]

$\bra{H}( B\tensorchi I)\ket{G} =\bra{H_{\chi }} B\ket{G_{\chi }}\braket{H_{\overline{\chi }}|G_{\overline{\chi }}}$ by the preliminaries. So, $B\tensorchi I$ is $\chi $-local.

We show that $( B\tensorchi I) =(( B\tensorchi I) \tensorchi     \ I)$. By preliminaries:
\begin{align*}
\bra{H}(( B\tensorchi I) \tensorchi     \ I)\ket{G} & =\bra{H_{\chi }}( B\tensorchi I)\ket{G_{\chi }}\braket{H_{\overline{\chi }}|G_{\overline{\chi }}}\\
\text{By prelim. } & =\bra{H_{\chi \chi }} B\ket{G_{\chi \chi }} \braket{ H_{\chi \overline{\chi }} | G_{\chi \overline{\chi }}}\braket{H_{\overline{\chi }}|G_{\overline{\chi }}}\\
\text{By idempotency and Lem \ref{lem:restrictions}} \  & =\bra{H_{\chi }} B\ket{G_{\chi }}\braket{\varnothing |\varnothing }\braket{H_{\overline{\chi }}|G_{\overline{\chi }}}\\
 & =\bra{H_{\chi }} B\ket{G_{\chi }}\braket{H_{\overline{\chi }}|G_{\overline{\chi }}}\\
\text{By prelim. } & =\bra{H}( B\tensorchi I)\ket{G}
\end{align*}
[Name-preserving case]

A matrix $B$ is n.-p. if and only if it is a sum of terms of the form $\ket{G}\bra{H}$ with $V( G) \corresponds V( H)$.\\
Then, $B\tensorchi I$ is a sum of terms of the form $\ket{G'}\bra{H'} =\left(\ket{G} \tensorchi \ket{K}\right)\left(\bra{H} \tensorchi \bra{K}\right)$, with \ $V( G') =V( G) \cup V( K) \corresponds V( H) \cup V( K) =V( H')$.

\end{proof}

Let us now go back to simple concrete examples:
At this stage we have characterised local operators in terms of generalised tensors, and for some example we have expressed generalised tensors in terms of standard tensors and direct sums. For these we can express local operators in terms of standard tensors and direct sums:
\begin{itemize}
\item For Example \ref{ex:1}, i.e. $\zeta_u$, we have that from the above theorem $\zeta_u$-locality of an operator $A$ entails \[  A \cong A_{u + \varnothing} \otimes I    .   \]
\item For Example \ref{ex:2}, i.e. $\chi_{\mathcal{C}}$ we have that from the above theorem $\chi_{\mathcal{C}}$-locality of an operator $A$ entails \[  A \cong A_{\mathcal{C}} \oplus A_{\varnothing} I  .   \]
\item For Example \ref{ex:3}, i.e. $\chi_{u:S}$, we have that from the above theorem $\chi_{u:S}$-locality of an operator $A$ entails \[  A \cong (A_{S + \varnothing} \oplus A_{\varnothing} I ) \otimes I    .   \]
\item For Example \ref{ex:4}, i.e. $\mu_{S}$, we have that from the above theorem $\mu$-locality of an operator $A$ entails \[  A \cong \bigoplus_{\mathcal{Q}} (\pi_{\neg \mathcal{Q}:S} A \pi_{\neg \mathcal{Q}: S} \otimes I)    .   \]
\end{itemize}

\color{black}

\begin{proposition}[Strict locality and consistency]\label{prop:strictlocality}
$A$ is strictly $ \chi $-local if and only if $ A$ is $ \chi $-local and $ \chi $-consistent-preserving.
\end{proposition}

\begin{proof}

[Preliminary]

Notice that $A^{\dagger }$ is also $\chi $-consistent-preserving by the symmetry of the definition, and $\chi $-local since
\begin{align*}
\bra{G} A^{\dagger }\ket{H} & =\bra{H} A\ket{G}^{*}\\
 & =\ \bra{H_{\chi }} A\ket{G_{\chi }}{^{*}}^{\ }\braket{H_{\overline{\chi }}|G_{\overline{\chi }}}^{*}\\
 & =\ \bra{G_{\chi }} A^{\ \dagger }\ket{H_{\chi }}^{\ }\braket{G_{\overline{\chi }}|H_{\overline{\chi }}}
\end{align*}
[$\Rightarrow $]

Notice that $A$ $\chi $-consistent-preserving is equivalent to $||\left( A\ket{G_{\chi }}\right) \tensorchi \ket{G_{\overline{\chi }}} ||=||A\ket{G_{\chi }} ||$ and $||\left( A^{\dagger }\ket{G_{\chi }}\right) \tensorchi \ket{G_{\overline{\chi }}} ||=||A^{\dagger }\ket{G_{\chi }} ||$. Indeed, the only reason why the norm conditions would not hold, would be if some $\ket{H}$ was such that $\bra{H} A\ket{G_{\chi }}$ or $\bra{H} A^{\dagger }\ket{G_{\chi }}$, and yet \ $\ket{H} \tensorchi \ket{G_{\overline{\chi }}} =0$, i.e. if $A$ weren't $\chi $-consistent.

Since $A$ is $\chi $-local, $||A\ket{G} ||=||\left( A\ket{G_{\chi }}\right) \tensorchi \ket{G_{\overline{\chi }}} ||$. But since $A^{\dagger } A$ is $\chi $-local we also have that
\begin{align*}
||A\ket{G} || & =\bra{G} A^{\dagger } A\ket{G}\\
 & =\bra{G_{\chi }} A^{\dagger } A\ket{G_{\chi }}\braket{G_{\overline{\chi }}|G_{\overline{\chi }}}\\
 & =||A\ket{G_{\chi }} ||
\end{align*}
So the first norm condition is fullfilled. Similarly with $A^{\dagger }$ for the second norm condition.

[$\Leftarrow $] \ 

By Prop. \ref{prop:gatelocality}, $A=A\tensorchi I$ and $A^{\dagger } =A^{\dagger } \tensorchi I$. \\
By Lem. \ref{lem:interchangelaws}, $A^{\dagger } A=A^{\dagger } A\tensorchi I$ and $AA^{\dagger } =AA^{\dagger } \tensorchi I$.\\
Thus $A$ is strictly $\chi $-local. 
\end{proof}

\begin{figure}\centering
\includegraphics[width=0.4\textwidth]{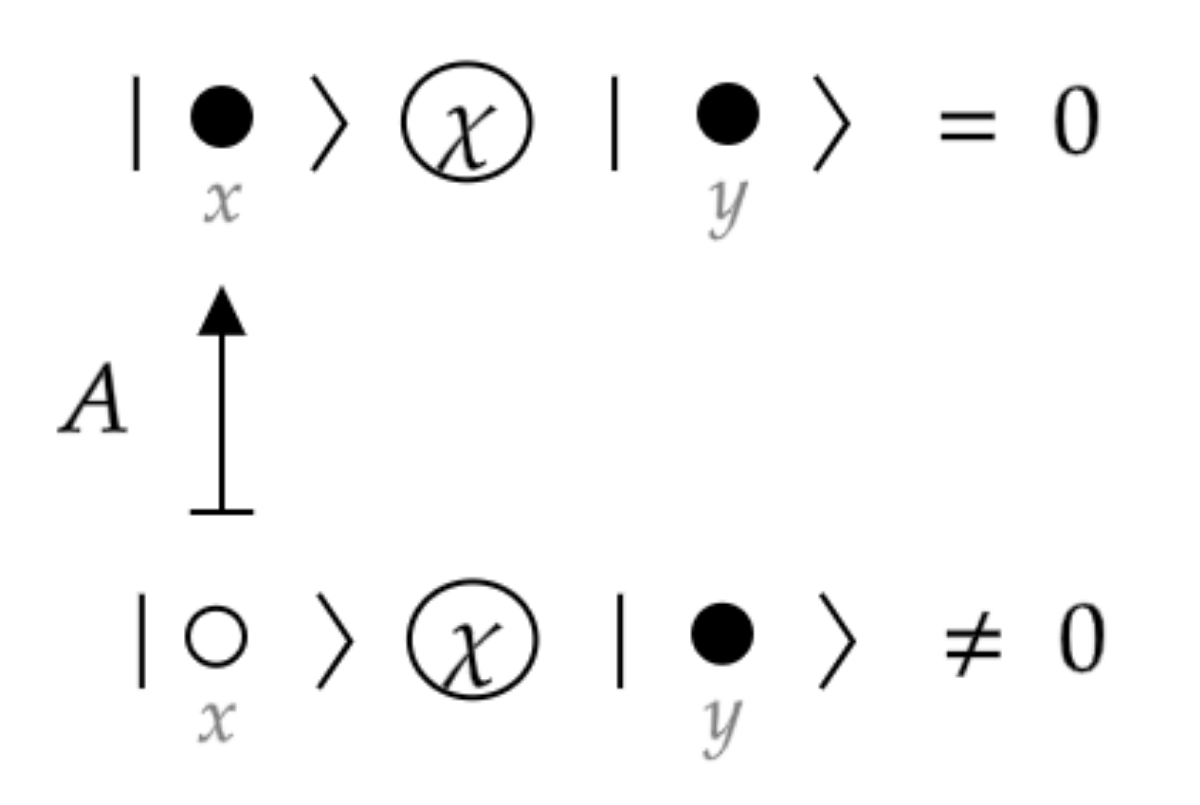}
\caption{\label{fig:strictlocality}{\em Local but not strictly local.} 
Consider the restriction $\chi$ which keeps all the nodes if none are white, and keeps just the white nodes otherwise. 
It can be checked to be extensible. 
Consider $A$ which maps fully white graphs into the corresponding fully black graphs, and sends every other graph to zero. Such an $A$ is $\chi$-local: $\bra{H}A\ket{G}\neq 0$ if and only if $H=H_\chi$ is the black coloured version of the white coloured $G=G_\chi$, which can be seen equivalent to $\bra{H_\chi}A\ket{G_\chi}\braket{H_{\overline{\chi}}|G_{\overline{\chi}}}\neq 0$. 
It is not strictly local because $\bra{11}A^\dagger A\ket{01}=0 \neq 1=\bra{1}A^\dagger A\ket{0}\braket{1|1}$. 
It is not consistent preserving because $\ket{0}\tensorchi \ket{1}\neq 0$ yet $\ket{1}\tensorchi \ket{1}= 0$.} 
\end{figure}

An operator $A$ may sometimes be $\chi $-local but not strictly $\chi $-local, see Fig. \ref{fig:strictlocality}. Since such an $A$ is not consistent-preserving, Lem. \ref{lem:interchangelaws} fails, it follows that the composition of two $\chi $-local operators is not always $\chi $-local.\\
None of these issues arise, however, if $A$ is unitary, or if it is just $\chi $-consistent-preserving. These entail strict $\chi $-locality, which is composable by Prop. \ref{prop:gatelocality} and Lem. \ref{lem:interchangelaws}.
For instance, in the Example \ref{ex:2}, i.e. $\chi_{\mathcal{C}}$, whilst locality entails $A \cong A_{\mathcal{C}} \oplus A_{\varnothing}I_{\mathcal{G} \setminus\mathcal{C}}$ the additional requirement of strict locality entails that instead $A \cong A_{\mathcal{H} \setminus \varnothing} \oplus A_{\varnothing}I_{\mathcal{G} \setminus\mathcal{C} + \varnothing}$. This stronger decomposition in this case ensures that $A$ cannot excite \textit{any} state out of the empty graph.  \color{black}

\subsection{Locality in the Heisenberg picture}

The result of an $\chi $-local observable on $\rho $ solely depends on its partial trace $\rho _{|\chi }$.

\begin{proposition}[Dual locality]\label{prop:duallocality}

$A$ is $\chi $-local if and only if $ ( A \rho )_{|\varnothing } =( A \rho _{|\chi })_{|\varnothing }$.
\end{proposition}
\begin{proof}

[$\Rightarrow $]
\begin{align*}
\left( A\ket{G}\bra{H}\right)_{|\emptyset } & =\bra{H} A\ket{G}\\
 & =\bra{H_{\chi }} A\ket{G_{\chi }}\braket{H_{\overline{\chi }}|G_{\overline{\chi }}}\\
 & =\left( A\ket{G_{\chi }}\bra{H_{\chi }}\braket{H_{\overline{\chi }}|G_{\overline{\chi }}}\right)_{|\emptyset }\\
 & =\left( A\left(\ket{G}\bra{H}\right)_{|\chi }\right)_{|\emptyset }
\end{align*}
[$\Leftarrow $] 
\begin{align*}
\bra{H} A\ket{G} & =\left( A\ket{G}\bra{H}\right)_{|\emptyset }\\
 & =\left( A\left(\ket{G}\bra{H}\right)_{|\chi }\right)_{|\emptyset }\\
 & =\left( A\ket{G_{\chi }}\bra{H_{\chi }}\braket{H_{\overline{\chi }}|G_{\overline{\chi }}}\right)_{|\emptyset }\\
 & =\bra{H_{\chi }} A\ket{G_{\chi }}\braket{H_{\overline{\chi }}|G_{\overline{\chi }}}
\end{align*}
\end{proof}

The above proposition states that $\rho _{|\chi }$ contains the part of $\rho $ that is observable by $\chi $-local operators. The next proposition states that $\rho _{|\chi }$ does contain anything more.

\begin{proposition}[Local tomography]\label{prop:tomography}
If for all $A$ $\chi $-local $ ( A\rho )_{|\emptyset } =( A\sigma )_{|\emptyset }$ then $ \rho _{|}{}_{\chi } =\sigma _{|\chi }$. \\
Moreover, if $ \rho $, $ \sigma $ are name-preserving and for all $A$ $\chi $-local and name-preserving $ ( A\rho )_{|\emptyset } =( A\sigma )_{|\emptyset }$, then $ \rho _{|}{}_{\chi } =\sigma _{|\chi }$. 
\end{proposition}

\begin{proof}

In general, 
$$\rho _{|}{}_{\chi } =\sum _{G_{\chi } ,\ H_{\chi } \ \in \ \mathcal{G}_{\chi }} \alpha _{G_{\chi } H_{\chi }}\ket{G_{\chi }}\bra{H_{\chi }}$$ 
$$\textrm{and }\sigma _{|}{}_{\chi } =\sum _{G_{\chi } ,\ H_{\chi } \ \in \ \mathcal{G}_{\chi }} \beta _{G_{\chi } H_{\chi }}\ket{G_{\chi }}\bra{H_{\chi }} \ $$.

Let $E_{\chi }^{H_{\chi } G_{\chi }} :=\ket{H_{\chi }}\bra{G_{\chi }}$ and $E^{H_{\chi } G_{\chi }} :=E_{\chi }^{H_{\chi } G_{\chi }} \tensorchi     \ I$, which is local by Prop. \ref{prop:gatelocality}.

We have $\left( E^{H_{\chi } G_{\chi }} \rho \right)_{|\emptyset } =\left( E^{H_{\chi } G_{\chi }} \rho _{|}{}_{\chi }\right)_{|\emptyset } =\left( E_{\chi }^{H_{\chi } G_{\chi }} \rho _{|}{}_{\chi }\right)_{|\emptyset } =\alpha _{G_{\chi } H_{\chi }}$, as the following shows:
\begin{align*}
\left( E^{H_{\chi } G_{\chi }} \rho _{|}{}_{\chi }\right)_{|\emptyset } & =\sum _{G'_{\chi } ,\ H'_{\chi } \ \in \ \mathcal{G}_{\chi } \ K\ \in \ \mathcal{G}} \alpha _{G'_{\chi } H'_{\chi }}\left(\left(\ket{H_{\chi }} \tensorchi \ket{K}\right)\left(\bra{G_{\chi }} \tensorchi \bra{K}\right)\ket{G'_{\chi }}\bra{H'_{\chi }}\right)_{|\emptyset }\\
 & =\sum _{G'_{\chi } ,\ H'_{\chi } \ \in \ \mathcal{G}_{\chi } \ K\ \in \ \mathcal{G}} \alpha _{G'_{\chi } H'_{\chi }}\left(\bra{G_{\chi }} \tensorchi \bra{K}\right)\ket{G'_{\chi }}\bra{H'_{\chi }}\left(\ket{H_{\chi }} \tensorchi \ket{K}\right)\\
\left(\bra{G_{\chi }} \tensorchi \bra{G_{\overline{\chi }}}\right)\ket{G'_{\chi }} & =\left(\bra{G_{\chi }} \tensorchi \bra{G_{\overline{\chi }}}\right)\left(\ket{G'_{\chi \chi }} \tensorchi \ket{G'_{\chi \overline{\chi }}}\right)\\
\text{By idempotency.} & =\left(\bra{G_{\chi }} \tensorchi \bra{G_{\overline{\chi }}}\right)\left(\ket{G'_{\chi }} \tensorchi \ket{\varnothing }\right)\\
 & =\braket{G_{\chi }|G'_{\chi }}\braket{G_{\overline{\chi }}|\emptyset }\\
\left( E^{H_{\chi } G_{\chi }} \rho _{|}{}_{\chi }\right)_{|\emptyset } & =\sum _{ \begin{array}{l}
G'_{\chi } ,\ H'_{\chi } \ \in \ \mathcal{G}_{\chi } \ K\ \in \ \mathcal{G}\\
\ket{G_{\chi }} \tensorchi K\ \neq \ 0\\
\ket{H_{\chi }} \tensorchi K\ \neq \ 0
\end{array}} \alpha _{G'_{\chi } H'_{\chi }}\braket{G_{\chi }|G'_{\chi }}\braket{K|\varnothing }\braket{H'_{\chi }|H_{\chi }}\braket{\varnothing |K}\\
 & =\alpha _{G_{\chi } H_{\chi }}
\end{align*}
so (n.-p.) $\chi $-local "measurements" can tell any difference between $\rho _{|}{}_{\chi }$ and $\sigma _{|}{}_{\chi }$.

[Name-preserving case]

In this case, 
$$\rho _{|}{}_{\chi } =\sum _{ \begin{array}{l}
G_{\chi } ,\ H_{\chi } \ \in \ \mathcal{G}_{\chi }\\
V( G_{\chi }) \corresponds V( H_{\chi })
\end{array}} \alpha _{G_{\chi } H_{\chi }}\ket{G_{\chi }}\bra{H_{\chi }} \ $$ 
$$\textrm{and } \sigma _{|}{}_{\chi } =\sum _{ \begin{array}{l}
G_{\chi } ,\ H_{\chi } \ \in \ \mathcal{G}_{\chi }\\
V( G_{\chi }) \corresponds V( H_{\chi })
\end{array}} \beta _{G_{\chi } H_{\chi }}\ket{G_{\chi }}\bra{H_{\chi }} \ $$.

The fact that $V( G_{\chi }) \corresponds V( H_{\chi })$ comes from the assumption that $\rho $, $\sigma $ are n.-p. and the fact that the partial trace preserves that by Prop. \ref{prop:traceouts}. \\
Then $E_{\chi }^{H_{\chi } G_{\chi }} :=\ket{H_{\chi }}\bra{G_{\chi }}$ and $E^{H_{\chi } G_{\chi }} :=E_{\chi }^{H_{\chi } G_{\chi }} \tensorchi     \ I$ \ are n.-p. by Prop. \ref{prop:traceouts}.
\end{proof}

Notice that if we stick to n.-p. observables, we limit our power of observation. For instance whenever $V( G)\not{\corresponds } V( H)$, n.-p. observables cannot tell the difference between:
\begin{equation*}
\rho =\frac{1}{2}\left(\ket{G} +\ket{H}\right)\left(\bra{G} +\bra{H}\right) \quad \text{and} \quad \tilde{\rho } =\frac{1}{2}\ket{G}\bra{G} +\frac{1}{2}\ket{H}\bra{H} \ 
\end{equation*}
i.e. they cannot read-out superpositions of supports coherently. That is unless states are n.-p., too.

\subsection{Extending unitaries acting on a subnetwork}

Often we are given a operator over $\mathcal{H}_{\chi }$, and we want to extend it to $\mathcal{H}$. In standard quantum theory it is easy to show that any such unitary operator can be extended, with the result being unitary. To generalise this to unitaries over quantum networks we will need name-preservation. 

\begin{proposition}[Unitary extension]\label{prop:unitaryextension}
Consider $ \chi $ pointwise.\\
If $ U$ is a n.-p. operator over $ \mathcal{H}_{\chi }$, then $ U$ $ \chi $-consistent-preserving.\\
If $ U$ is a n.-p. unitary over $ \mathcal{H}_{\chi }$, then $ U':=U\tensorchi I$ is a n.-p. unitary with $ U^{\prime \dagger } =U^{\dagger } \tensorchi I$.\\
If moreover $ \chi $ and $ U$ are renaming-invariant, so is $ U'$.
\end{proposition}
\begin{proof}

[Consistency-preservation]

We need to prove first that $\bra{G'_{\chi }} U\ket{G_{\chi }} \neq 0$ entails $\ket{G'_{\chi }} \tensorchi \ket{G_{\overline{\chi }}} \neq 0$ and second that $\bra{G'_{\chi }} U^{\dagger }\ket{G_{\chi }} \neq 0$ entails $\ket{G'_{\chi }} \tensorchi \ket{G_{\overline{\chi }}} \neq 0$.

Let us prove the first.

Since $U$ is over $\mathcal{H}_{\chi }$, 
$U\ket{G_{\chi }}  =\sum _{H_{\chi } \ \in \ \mathcal{G}_{\chi }} U_{H_{\chi } G_{\chi }}\ket{H_{\chi }}$.\\
For any $G$, consider $H_{\chi }$ such that $\bra{H_{\chi }} U\ket{G_{\chi }} \neq 0$ and construct the graph $G'=H_{\chi } \cup G_{\overline{\chi }}$.\\ 
Notice that this union is always defined since $U$ is assumed n.-p., and hence $\mathcal{N}[ V( G_{\chi })] \cap \mathcal{N}[ V( G_{\overline{\chi }})] =\emptyset $ entails $\mathcal{N}[ V( H_{\chi })] \cap \mathcal{N}[ V( G_{\overline{\chi }})] =\emptyset $ --- assuming that we work over the set of all graphs ${\cal G}$. In general this could fail for some arbitrary constrained configurations ${\cal C}$ and a restriction $\chi$ over it.\\
Moreover since $\chi $ is pointwise it verifies that $G'_{\chi } =H_{\chi }$ and $G'_{\overline{\chi }} =G_{\overline{\chi }}$.\\
Thus, $\ket{H_{\chi }} \tensorchi \ket{G_{\overline{\chi }}} =\ket{G'} \neq 0$.\\
It follows that $\left( U\ket{G_{\chi }}\right)$, $\ket{G_{\overline{\chi }}}$ are $\chi $-consistent.

Similarly for $U^{\dagger }$. Hence $U$ is $\chi $-consistent-preserving.

[Unitarity]

By Lem. \ref{lem:interchangelaws},
\begin{align*}
( U\tensorchi I)\left( U^{\dagger } \tensorchi I\right) & =\left( UU^{\dagger } \tensorchi I\right) =( I\tensorchi I) =I\\
\left( U^{\dagger } \tensorchi I\right)( U\tensorchi I) & =\left( U^{\dagger } U\tensorchi I\right) =( I\tensorchi I) =I
\end{align*}
\begin{align*}
\left( U\ket{G_{\chi }}\right) \tensorchi \ket{G_{\overline{\mu }}} & =\sum _{H_{\chi } \ \in \ \mathcal{G}_{\chi }} U_{H_{\chi } G_{\chi }}\left(\ket{H_{\chi }} \tensorchi \ket{G_{\overline{\chi }}}\right)\\
||U'\ket{G} ||^{2} & =\sum _{H_{\chi } ,\ H'_{\chi } \in \ \mathcal{G}_{\chi }} U_{H'_{\chi } G_{\chi }}^{*} U_{H_{\chi } G_{\chi }}\left(\bra{H'_{\chi }} \tensorchi \bra{G_{\overline{\chi }}}\right)\left(\ket{H_{\chi }} \tensorchi \ket{G_{\overline{\chi }}}\right)\\
 & =\sum _{ \begin{array}{l}
H_{\chi } ,\ H'_{\chi } \in \ \mathcal{G}_{\chi }\\
\ket{H_{\chi }} \tensorchi \ket{G_{\overline{\chi }}} \ \neq \ 0\\
\ket{H'_{\chi }} \tensorchi \ket{G_{\overline{\chi }}} \ \neq \ 0
\end{array}} U_{H'_{\chi } G_{\chi }}^{*} U_{H_{\chi } G_{\chi }}\braket{H'_{\chi }|H_{\chi }}\braket{G_{\overline{\chi }}|G_{\overline{\chi }}}\\
\text{By} \ U,\ I\ \chi \text{-consistent-preserving.} & =\sum _{H_{\chi } ,\ H'_{\chi } \in \ \mathcal{G}_{\chi }} U_{H'_{\chi } G_{\chi }}^{*} U_{H_{\chi } G_{\chi }}\braket{H'_{\chi }|H_{\chi }}\braket{G_{\overline{\chi }}|G_{\overline{\chi }}}\\
\text{By unitarity of} \ U. & =\sum _{H_{\chi } \ \in \ \mathcal{G}_{\chi }} U_{H_{\chi } G_{\chi }}^{*} U_{H_{\chi } G_{\chi }} =1
\end{align*}
Thus $U'$ is an isometry, i.e. $U^{\prime \dagger } U'=I$. 

Notice that $\left( U^{\dagger } \tensorchi I\right)$ is its right inverse since by means of Lem. \ref{lem:interchangelaws}. 

Since $U$ is over $\mathcal{H}_{\chi }$, we have that $U^{\ \dagger }$preserves the range of $\chi $, and so $U^{\dagger } ,\ I$ are $\chi $-consistent-preserving. 

Thus, 
\begin{align*}
\left( U^{\dagger } \tensorchi I\right)\ket{G} & =\sum _{H_{\chi } \ \in \ \mathcal{G}_{\chi }} U_{G_{\chi } H_{\chi }}^{*}\left(\ket{H_{\chi }} \tensorchi \ket{G_{\overline{\chi }}}\right)\\
U'\left( U^{\dagger } \tensorchi I\right)\ket{G} & =\sum _{ \begin{array}{l}
H_{\chi } \ \in \ \mathcal{G}_{\chi }\\
\ket{H_{\chi }} \tensorchi \ket{G_{\overline{\chi }}} \ \neq \ 0
\end{array}} U_{H'_{\chi } H_{\chi }} U_{G_{\chi } H_{\chi }}^{*}\left(\ket{H'_{\chi }} \tensorchi \ket{G_{\overline{\chi }}}\right)\\
\text{By} \ U^{\dagger } ,\ I\ \chi \text{-consistent-preserving.} & =\sum _{H_{\chi } \ \in \ \mathcal{G}_{\chi }} U_{H_{\chi } H'_{\chi }} U_{G_{\chi } H_{\chi }}^{*}\left(\ket{H'_{\chi }} \tensorchi \ket{G_{\overline{\chi }}}\right)\\
 & =\sum _{H_{\chi } \ \in \ \mathcal{G}_{\chi }} I_{H'_{\chi } G_{\chi }}\left(\ket{H'_{\chi }} \tensorchi \ket{G_{\overline{\chi }}}\right)\\
 & =\ket{G}
\end{align*}
[Name-preservation]

Follows from Prop. \ref{prop:gatelocality}.

[Renaming-invariance]

Let $G'=RG$ and $H=RH'$.
\begin{align*}
\bra{H}( U \tensorchi I) R\ket{G} & =\bra{H}( U \tensorchi I)\ket{G'}\\
 & =\bra{H}( U \tensorchi I)\left(\ket{G'_{\chi }} \tensorchi \ket{G'_{\overline{\chi }}}\right)\\
 & =\left(\bra{H_{\chi }} \tensorchi \bra{H_{\overline{\chi }}}\right)\left( U\ket{G'_{\chi }} \tensorchi \ket{G'_{\overline{\chi }}}\right)\\
U\\ 
\chi \text{-consistent-preserving.} & =\bra{H_{\chi }} U\ket{G'_{\chi }} \ \braket{H_{\overline{\chi }}|G'_{\overline{\chi }}}\\
\text{Since } \chi \ \text{renaming-invariant.} & =\bra{H_{\chi }} UR\ket{G_{\chi }} \ \bra{H_{\overline{\chi }}} R\ket{G_{\overline{\chi }}}\\
\text{Since } U\ \text{renaming-invariant.} & =\bra{H_{\chi }} RU\ket{G_{\chi }} \ \bra{H_{\overline{\chi }}} R\ket{G_{\overline{\chi }}}\\
\text{Since } \chi \ \text{renaming-invariant.} & =\bra{H'_{\chi }} U\ket{G_{\chi }} \ \braket{H'_{\overline{\chi }}|G_{\overline{\chi }}}\\
U\ \chi \text{-consistent-preserving.} & =\left(\bra{H'_{\chi }} \tensorchi \bra{H'_{\overline{\chi }}}\right)\left( U\ket{G_{\chi }} \tensorchi \ket{G_{\overline{\chi }}}\right)\\
 & =\bra{H'}( U \tensorchi I)\left(\ket{G_{\chi }} \tensorchi \ket{G_{\overline{\chi }}}\right)\\
 & =\bra{H} R( U \tensorchi I)\ket{G}
\end{align*}
\end{proof}

\section{Causal operators over quantum networks}\label{sec:causality}

\begin{figure}\centering
\includegraphics[width=\textwidth]{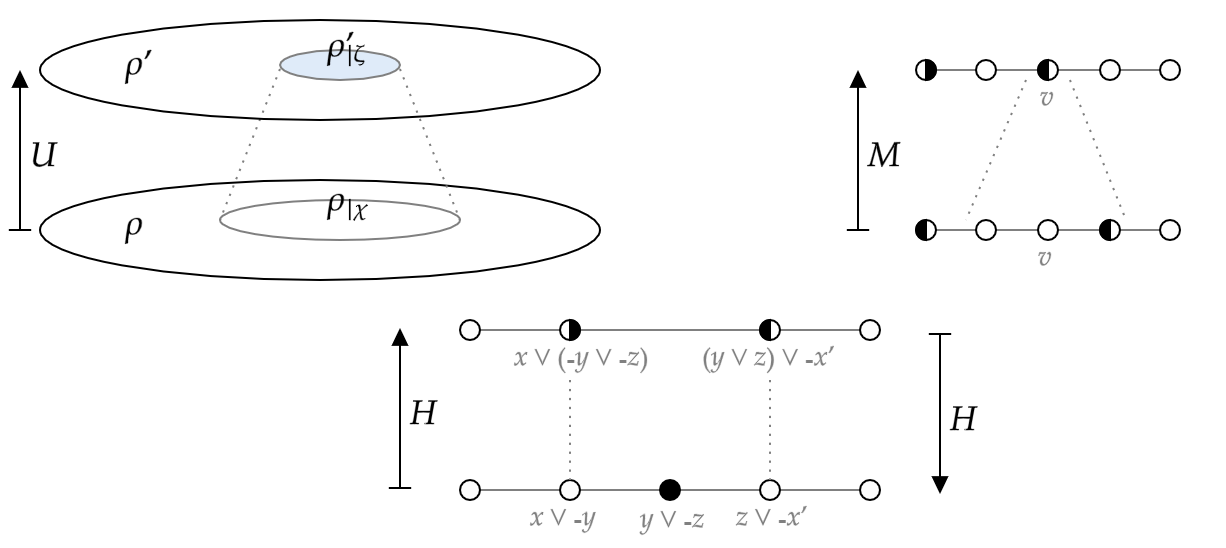}
\caption{\label{fig:causality}{\em Causal operators.} Left: $U\ \chi \zeta $-causal may modify the whole of $\rho $. But it is such that $\rho '_{|\zeta }$ solely depends on $\rho _{|\chi }$. Right: $\chi _{v} \zeta _{v}$-causal operator $M$ propagates particles. They bounce on borders. Middle: Causal $H$ is involutive merges/splits all occurrences of these particular patterns, synchronously.}
\end{figure}

\begin{figure}\centering
\includegraphics[width=\textwidth]{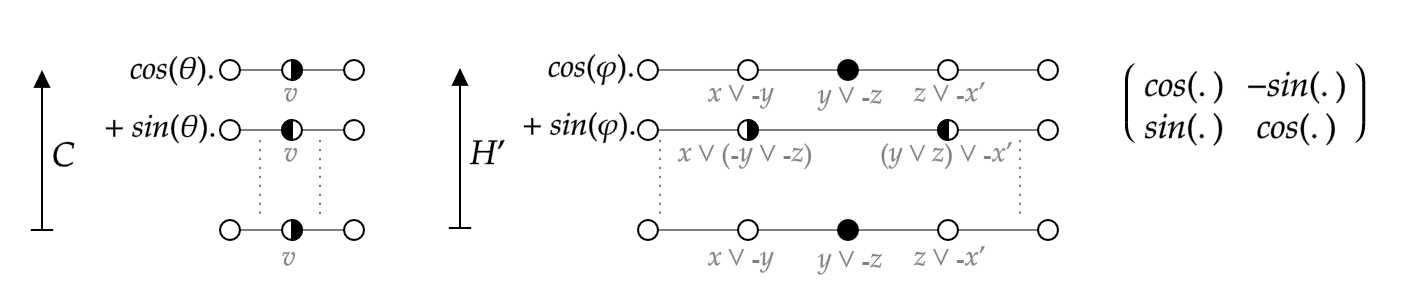}
\caption{\label{fig:quantumcausality} {\em Causal operators yielding superpositions of states.} Left: Instead of simply iterating $M$, we could iterate $MC$, where $C$ acts on every right-moving (resp. left-moving) particle by placing it in a superposition of being right-moving with amplitude $cos( \theta )$ and left-moving with amplitude $sin( \theta )$ (resp. left-moving with amplitude $cos( \theta )$ and right-moving with amplitude $-sin( \theta )$. Right: Instead of simply merging or splitting according to $H$, we could do so in a superposition. Here the splitting applies with amplitude $sin( \varphi )$. Similarly the merge needs be applied with amplitude $-sin( \varphi )$. Lastly, we may compose these, e.g. \ $U=H'MC$ makes for an interesting quantum causal graph dynamics.}
\end{figure}
A key principle in physical frameworks is causality, understood as constraints upon the propagation of information. In the quantum information-theoretic settings, it has become standard to express causality constraints by means of partial traces \cite{Beckman}. Now that we have a generalised trace, we can have generalised notions of causality.  A canonical example of a causal unitary operator is given in Fig. \ref{fig:causality}. A generalised causality example is given in \ref{fig:quantumcausality}.\\
Consider two restrictions $\chi,\zeta$ over networks, a $\chi \zeta $-causal operator is one which restricts information propagation by imposing that region $\zeta $ at the next time step depends only upon region $\chi $ at the previous time step. Subject to this constraint, $\chi \zeta $-causal operator will be permitted to edit the entirety of the graphs they act on.

\begin{definition}[Causality]\label{def:causality}

$U$ is $\chi \zeta $-causal if and only if
\begin{equation}
\left( U\rho U^{\dagger }\right)_{|\zeta } =\left( U\rho _{|\chi } U^{\dagger }\right)_{|\zeta }\label{eq:causalitly}
\end{equation}
$U$ is name-preserving $\chi \zeta $-causal if and only if it is n.-p. and for all $ \rho $ n.-p., Eq. \eqref{eq:causalitly} holds.
\end{definition}

Notice that it is not the case that n.-p. causality implies causality. In particular, the identity is n.-p. $\zeta _{v}^{2} \zeta _{v}^{1}$-causal, but not $\zeta _{v}^{2} \zeta _{v}^{1}$-causal, because $\zeta _{v}^{1} \ \not\sqsubseteq \ \zeta _{v}^{2}$, as discussed in Fig. \ref{fig:incomprehension}. That is unless we restrict ourselves to n.-p. superselected states as in Prop. \ref{prop:npcomprehension}. This suggests that n.-p. causality is potentially more relevant than causality.

A robust notion of causality ought to be composable.

\begin{proposition}[Composability]\label{prop:causalitycomposability}
Say that for all $ n$, there exists $ m$ such that $ U,V$ are (n.-p.) $ \zeta ^{m} \zeta ^{n}$-causal.\\
We have that for all $ n$ there exists $ m$ such that $ UV$ is (n.-p.) $ \zeta ^{m} \zeta ^{n}$-causal.
\end{proposition}

\begin{proof}

[$\Rightarrow $] 

For all $\rho $ (n.-p.),
\begin{align*}
 & \left( U\left( V\rho V^{\dagger }\right) U^{\dagger }\right)_{|\zeta ^{n}}\\
U\ \text{caus. } & =\left( U\left( V\rho V^{\dagger }\right)_{|\zeta ^{k}} U^{\dagger }\right)_{|\zeta ^{n}}\\
V\ \text{caus. } & =\left( U\left( V\rho _{|\zeta ^{m}} V^{\dagger }\right)_{|\zeta ^{k}} U^{\dagger }\right)_{|\zeta ^{n}}\\
U\ \text{caus. } & =\left( U\left( V\rho _{|\zeta ^{m}} V^{\dagger }\right) U^{\dagger }\right)_{|\zeta ^{n}}
\end{align*}

\end{proof}

\subsection{Causality in the Heisenberg picture}
In the Heisenberg picture, it turns out that $\chi \zeta $-causality actually states that whatever can be $\zeta $-locally observed at the next time step, could be $\chi $-locally observed at the previous time step.

\begin{proposition}[Dual causality]\label{prop:dualcausality}
$U$ is (n.-p.) $\chi \zeta $-causal if and only if (it is n.-p. and) for all $ A$ (n.-p.) $ \zeta $-local, $ U^{\dagger } AU$ is (n.-p.) $ \chi $-local.\\
If $U$ is (n.-p.) $\chi \zeta $-causal, then for all $A$ (n.-p.) strictly $\zeta $-local, $U^{\dagger } AU$ is (n.-p.) strictly $\chi $-local.
\end{proposition}

\begin{proof}

[$\Rightarrow $] \ 

For all $\rho $ (n.-p.), $\left( U\rho U^{\dagger }\right)_{|\zeta } =\left( U\rho _{|\chi } U^{\dagger }\right)_{|\zeta }$.\\ 
By Prop. \ref{prop:duallocality}, $A$ $\zeta $-local entail $\left( AU\rho U^{\dagger }\right)_{|\emptyset } =\left( AU\rho _{|\chi } U^{\dagger }\right)_{|\emptyset }$.\\
Thus, for all $\rho $, $\left( U^{\dagger } AU\rho \right)_{|\emptyset } =\left( U^{\dagger } AU\rho _{|\chi }\right)_{|\emptyset }$. 

(For the name-preserving case, then $U,\ A$ are n.-p., so $U^{\dagger } AU$ is n.-p., and so if $\rho $ is not n.p. both sides of the equation are zero, so the above does stand for all $\rho $.) 

So, by Prop. \ref{prop:duallocality}, $B=U^{\dagger } AU$ is (n.-p.) $\chi $-local.

[ Strict$\Rightarrow $]

If $A$ is strictly $\zeta $-local then $A^{\dagger } A$ and $AA^{\dagger }$ are $\zeta $-local. \\
From the above it follows that $U^{\dagger } A^{\dagger } AU$ and $UAA^{\dagger } U$ are $\chi $-local.\\ 
But $U^{\dagger } A^{\dagger } AU=U^{\dagger } A^{\dagger } UU^{\dagger } AU=B^{\dagger } B$ and $UAA^{\dagger } U=UAUU^{\dagger } A^{\dagger } U=BB^{\dagger }$. \\
So, $B=U^{\dagger } AU$ is strictly $\chi $-local.

[$\Leftarrow $] 

For all $A$ (n.-p.) $\zeta $-local, $U^{\dagger } AU$ is (n.-p.) $\chi $-local. By Prop. \ref{prop:duallocality}, for all $\rho $, $\left( U^{\dagger } AU\rho \right)_{|\emptyset } =\left( U^{\dagger } AU\rho _{|\chi }\right)_{|\emptyset }$, from which it follows that for all $A$ (n.-p.), for all $\rho $, $\left( AU\rho U^{\dagger }\right)_{|\emptyset } =\left( AU\rho _{|\chi } U^{\dagger }\right)_{|\emptyset }$.\\
Finally by Prop. \ref{prop:tomography}, for all $\rho $ (n.-p.), we have$\left( U\rho U^{\dagger }\right)_{|\zeta } =\left( U\rho _{|\chi } U^{\dagger }\right)_{|\zeta }$. 

(In the name-preserving case, by taking $\rho $ to be n.-p. we have that \ $\rho _{|\chi } ,\ U\rho U^{\dagger } ,\ U\rho _{|\chi } U^{\dagger }$ are n.-p, so that n.-p. local tomography can still be used to reach the final equality.)

Thus $U$ is (n.-p.) $\chi \zeta $-causal. 
\end{proof}

Often we are given a causal operator over $\mathcal{H}_{\chi }$, and we want to extend it to a causal operator over $\mathcal{H}$, the above theorem on the dual notion of causality in the Heisenberg picture can be used to show that such an extension is always possible. 

\begin{proposition}[Causal extension]\label{prop:causalextension}
First consider $ U$ a $ \chi '\zeta '$-causal operator and $ \chi '\sqsubseteq \chi $, $ \zeta \sqsubseteq \zeta '$. \\
Then $ U$ is an $ \chi \zeta $-causal operator.

Second say $\chi$ is extensible and consider $ \mu ,\zeta $ two restrictions such that $ [ \mu ,\zeta ] =[\overline{\mu } ,\zeta ] =[ \mu ,\overline{\zeta }] =[\overline{\mu } ,\overline{\zeta }] =0$, with $ \mu $ pointwise and $\zeta$ extensible.
Consider $ U$ an n.-p. $ \chi \zeta $-causal unitary operator over $ \mathcal{H}_{\mu }$, and $ U':=U \tensormu  I$ its unitary extension. \\
Then $ U'$ is $ \xi \zeta $-causal operator, w.r.t the extensible restriction $ \xi :=\mu \chi \cup \overline{\mu } \zeta $.
\end{proposition}

\begin{proof}

[First part]

Suppose $U$ is $\chi '\zeta '$-causal and $\chi '\sqsubseteq \chi $, $\zeta \sqsubseteq \zeta '$. 

Prop. \ref{prop:dualcausality} and using Lem. \ref{lem:tracetrace}, $A$ $\zeta $-local implies $A$ $\zeta '$-local implies \ $U^{\dagger } AU$ $\chi '$-local implies $U^{\dagger } AU\ \chi $-local. Thus $A$ $\zeta $-local implies $U^{\dagger } AU$ $\chi $-local which by Prop. \ref{prop:dualcausality} is equivalent to $\chi \zeta $-causality.

[Second part]
From Lem. \ref{lem:combiningrestrictions} we have $[ \mu ,\xi ] =[\overline{\mu } ,\xi ] =[ \mu ,\overline{\xi }] =[\overline{\mu } ,\overline{\xi }] =0$, and $\xi $ a restriction.

Any $A$ $\zeta $-local is of the form $A=L \tensorzeta   I$ with $L=\sum \alpha _{G_{\zeta } H_{\zeta }}\ket{G_{\zeta }}\bra{H_{\zeta }}$.
\begin{align*}
\ket{G_{\zeta }}\bra{H_{\zeta }}  \tensorzeta   I & =\ \left(\ket{G_{\zeta \mu }}\bra{H_{\zeta \mu }}  \tensormu  \ket{G_{\zeta \overline{\mu }}}\bra{H_{\zeta \overline{\mu }}}\right)  \tensorzeta   ( I_{\overline{\zeta } \mu }  \tensormu  I_{\overline{\zeta }\overline{\mu }})\\
\text{Commut. \& Lem. \ref{lem:tensortensor}.} & =\ \left(\ket{G_{\mu \zeta }}\bra{H_{\mu \zeta }}  \tensorzeta   I_{\mu \overline{\zeta }}\right) \tensormu      \left(\ket{G_{\overline{\mu } \zeta }}\bra{H_{\overline{\mu } \zeta }}  \tensorzeta   I_{\overline{\mu }\overline{\zeta }}\right)\\
U'\left(\ket{G_{\zeta }}\bra{H_{\zeta }}  \tensorzeta   I\right) U^{\prime \dagger } & =\left( U\left(\ket{G_{\mu \zeta }}\bra{H_{\mu \zeta }}  \tensorzeta   I_{\mu \overline{\zeta }}\right) U^{\dagger }\right)  \tensormu  \left(\ket{G_{\overline{\mu } \zeta }}\bra{H_{\overline{\mu } \zeta }}  \tensorzeta   I_{\overline{\mu }\overline{\zeta }}\right)\\
\text{(By dual caus.) } & =\left( M^{G_{\mu \zeta } H_{\mu \zeta }} \tensorchi I_{\mu \chi }\right) \tensormu      \left(\ket{G_{\overline{\mu } \zeta }}\bra{H_{\overline{\mu } \zeta }}  \tensorzeta   I_{\overline{\mu }\overline{\zeta }}\right)\\
U\ \text{over} \ \mathcal{H}_{\mu } & =\left( M^{G_{\mu \zeta } H_{\mu \zeta }} \ \tensormuchi \ I_{\mu \chi }\right) \tensormu      \left(\ket{G_{\overline{\mu } \zeta }}\bra{H_{\overline{\mu } \zeta }} \ \tensormubarchi \ I_{\overline{\mu }\overline{\zeta }}\right)\\
 & =\left( M^{G_{\mu \zeta } H_{\mu \zeta }}  \tensorxi    I_{\mu \overline{\xi }}\right) \tensormu      \left(\ket{G_{\overline{\mu } \zeta }}\bra{H_{\overline{\mu } \zeta }}  \tensorxi    I_{\overline{\mu }\overline{\xi }}\right)\\
\text{Commut. \& Lem. \ref{lem:tensortensor}.} & =\left( M^{G_{\mu \zeta } H_{\mu \zeta }}  \tensormu  \ket{G_{\overline{\mu } \zeta }}\bra{H_{\overline{\mu } \zeta }}\right)  \tensorxi    ( I_{\mu }  \tensormu  I_{\overline{\mu }})\\
\text{Lem. \ref{lem:tensor}.} & =\left( M^{G_{\mu \zeta } H_{\mu \zeta }}  \tensormu  \ket{G_{\overline{\mu } \zeta }}\bra{H_{\overline{\mu } \zeta }}\right)  \tensorxi    I_{\overline{\xi }}\\
U'A U^{\prime \dagger } & =\sum \alpha _{G_{\zeta } H_{\zeta }}\left(\left( M^{G_{\mu \zeta } H_{\mu \zeta }}  \tensormu  \ket{G_{\overline{\mu } \zeta }}\bra{H_{\overline{\mu } \zeta }}\right)  \tensorxi    I_{\overline{\xi }}\right)\\
\text{By bilinearity } & =\left(\sum \alpha _{G_{\zeta } H_{\zeta }}\left( M^{G_{\mu \zeta } H_{\mu \zeta }}  \tensormu  \ket{G_{\overline{\mu } \zeta }}\bra{H_{\overline{\mu } \zeta }}\right)\right)   \tensorxi          I_{\overline{\xi }}
\end{align*}
So, $U'A U^{\prime \dagger }$ is $\xi $-local by the Prop. \ref{prop:gatelocality}. By Prop. \ref{prop:unitaryextension}, $U'$ is name-preserving. By Prop. \ref{prop:dualcausality}, \ $U'$ is $\xi \zeta $-causal. 
\end{proof}

\subsection{Operational causality}

Causality is a basic Physics principle, anchored on the postulate that information-propagation is bounded by the speed of light. Yet causality is a top-down axiomatic constraint. \\
When modelling an actual Physical phenomenon, we need a bottom-up, constructive way of expressing the dynamics. We usually proceed by describing it in terms of local interactions, happening simultaneously and synchronously. \\

The following shows that causal operators are always of that form. Caution: this relies on Prop. \ref{prop:causalextension}, which does not hold for an arbitrarily constrained configuration space ${\cal H}^{\cal C}$.

\begin{thm}[Renaming-invariant block decomposition]\label{th:blockdecomposition}

Let $ \zeta _{v}$ be the pointwise restriction such that $ \zeta _{v}(\{\sigma '.u\}) :=\begin{cases}
\{\sigma '.u\} & \text{if} \ u=v\\
\emptyset  & \text{otherwise}
\end{cases}$.\\
Consider a family $\chi_v$, of extensible restrictions indexed by the names of $\mathcal{V}$. Consider $ U$ a n.-p. unitary operator over $ \mathcal{H}$, which for all $ v\in \mathcal{V}$ is $ \chi _{v} \zeta '_{v}$-causal, with $ \zeta _{v} \sqsubseteq \zeta '_{v}$.

Let $ \Sigma '=\{0,1\} \times \Sigma $. Let $ \mathcal{G} '$ be the set of finite subsets of $ \mathcal{S} ':=\Sigma '\times \mathcal{V}$ and $ \mathcal{H} '$ the Hilbert space whose canonical o.n.b is $ \mathcal{G} '$---as obtained by considering the free (complex) vector space using the configurations as the generating set; equipping it with the inner product that is such that the generating forms an orthonormal basis; and taking the completion of resulting inner product space to obtain the Hilbert space. \\
Let $ \mu $ be the pointwise restriction such that $ \mu (\{b.\sigma .u\}) :=\begin{cases}
\{0.\sigma .u\} & \text{if} \ b=0\\
\emptyset  & \text{otherwise}
\end{cases} \ $. \\
Over $ \mathcal{H} '$, there exists $ \tau _{v}$ a renaming-invariant strictly $ \zeta _{v}$-local unitary and $ K_{v}$ a strictly $ \xi _{v}$-local unitary such that
\begin{equation*}
\begin{aligned}
\forall \ket{\psi } \in \mathcal{H} '_{\mu } \cong \mathcal{H} ,\ \left(\prod _{v\ \in \ \mathcal{V}} \tau _{v}\right)\left(\prod _{v\ \in \ \mathcal{V}} K_{v}\right)\ket{\psi } & =U\ket{\psi }
\end{aligned}
\end{equation*}
where $ \xi _{v} :=\mu \chi _{v} \cup \overline{\mu } \zeta _{v}$ is an extensible restriction and $\mathcal{H}_{\mu}'$ is the sub-Hilbert space of graphs with all nodes of the form $0.\sigma.u$ for some $\sigma.u$. In addition,$ \ [ K_{x} ,K_{y}] =[ \tau _{x} ,\tau _{y}] =0$. 

If moreover $U$ is renaming-invariant, then so are $ U'$ and $ K_{v}$.
\end{thm}
\begin{proof}
Notice that $\mu $ is renaming-invariant. \\
Notice that $\zeta _{v}$ is renaming-invariant. \\
Clearly $[ \mu ,\zeta ] =[\overline{\mu } ,\zeta ] =[ \mu ,\overline{\zeta }] =[\overline{\mu } ,\overline{\zeta }] =0$ as both are pointwise.

By Prop. \ref{prop:unitaryextension}, $U':=U \tensormu  I$ is unitary over $\mathcal{H} '$.\\
Since $\mu $ is renaming-invariant, if $U$ is renaming-invariant, so is $U'$.\\ 
By Prop. \ref{prop:causalextension} and since $U$ is $\chi _{v} \zeta '_{v}$-causal, it is $\chi _{v} \zeta _{v}$-causal, and $U'$ is $\xi _{v} \zeta _{v}$-causal, with $\xi _{v} :=\mu \chi _{v} \cup \overline{\mu } \zeta _{v}$ an extensible restriction.

Let the toggle $\tau $ be the bijection over systems such that $\tau ( b.\sigma .u) =\neg b.\sigma .u$.\\
Extend $\tau $ to $\mathcal{G} '$ by acting pointwise upon each system, and to $\mathcal{H} '$ by linearity.\\
Notice that it is unitary, name-preserving and renaming-invariant.\\
Notice that $\tau \left(\ket{G_{\mu }}  \tensormu  \ket{G_{\overline{\mu }}}\right) =\left( \tau \ket{G_{\overline{\mu }}}  \tensormu  \tau \ket{G_{\mu }}\right)$.\\
It is also unitary over $\mathcal{H}_{\zeta _{v}}$, thus $\tau _{v} :=\tau \   \tensorzeta         _{v} \ I$ is unitary over $\mathcal{H} '$ by Prop. \ref{prop:unitaryextension}.\\
By Prop. \ref{prop:gatelocality}, $\tau _{v}$ is $\zeta _{v}$-local. By unitarity, it is strictly $\zeta _{v}$-local.\\
Since $\tau $ and $\zeta _{v}$ are renaming-invariant, so is $\tau _{v}$.\\
Moreover, 
\begin{equation*}
\left(\prod _{v\ \in \ \mathcal{V}} \tau_{v}\right)=\tau
\end{equation*}
Notice also that $[ \tau _{u} ,\tau _{v}] =0$.

Let $K_{v} :=U^{\prime \dagger } \tau _{v} U'$. \\
It is name-preserving as a composition of name-preserving operators.\\
If $U'$ is renaming-invariant, since $\tau _{v}$ is renaming-invariant, so is $K_{v}$.\\
Since adjunction by a unitary is a morphism, $[ K_{u} ,K_{v}] =0$.\\
By Prop. \ref{prop:dualcausality}, it is $\xi _{v}$-local. By unitarity, it is strictly $\xi _{v}$-local.

Finally,
\begin{align*}
\left(\prod _{v\ \in \ \mathcal{V}}\tau _{v}\right)\left(\prod _{v\ \in \ \mathcal{V}}K_{v}\right)\ket{G} & =\tau \dotsc \left(U'^{\dagger }\tau _{v_{2}}U'\right)\left(U'^{\dagger }\tau _{v_{1}}U'\right)\ket{G}\\
\text{By unitarity of }U'. & =\tau \ U'^{\dagger }\ \left(\prod _{v\ \in \ \mathcal{V}}\tau _{v}\right)\ U'\ \ket{G}\\
 & =\tau \ U^{\prime \dagger } \ \tau \ U'\ \ket{G}\\
\text{By Prop. \ref{prop:unitaryextension}} & =\tau \ \left( U^{\dagger }  \tensormu  I\right) \ \tau \ ( U \tensormu  I) \ \left(\ket{G_{\mu }}  \tensormu  \ket{G_{\overline{\mu }}}\right) \ \\
 & =\tau \ \left( U^{\dagger }  \tensormu  I\right) \ \tau \ \left( U\ket{G_{\mu }}  \tensormu  \ket{G_{\overline{\mu }}}\right)\\
\text{Since} \ U\ \text{preserves the range of } \mu . & =\tau \ \left( U^{\dagger }  \tensormu  I\right)\left( \tau \ket{G_{\overline{\mu }}}  \tensormu  \tau U\ket{G_{\mu }}\right)\\
 & =\tau \ \left( U^{\dagger } \tau \ket{G_{\overline{\mu }}}  \tensormu  \tau U\ket{G_{\mu }}\right)\\
\text{Since} \ U\ \text{preserves the range of } \mu . & =\tau ^{2} U\ket{G_{\mu }}  \tensormu  \tau \ U^{\dagger } \tau \ket{G_{\overline{\mu }}}\\
\text{Since} \ \tau \ \text{involutive.} & =U\ket{G_{\mu }}  \tensormu  \tau \ U^{\dagger } \tau \ket{G_{\overline{\mu }}}\\
\left(\prod _{v\ \in \ \mathcal{V}}\tau_{v}\right)\left(\prod _{v\ \in \ \mathcal{V}}K_{v}\right)\ket{G_{\mu }} & =\left(\prod _{v\ \in \ \mathcal{V}}\tau _{v}\right)\left(\prod _{v\ \in \ \mathcal{V}}K_{v}\right)\left(\ket{G_{\mu }} \tensormu  \ket{\varnothing }\right)\\
 & =U\ket{G_{\mu }}  \tensormu  \tau \ U^{\dagger } \tau \ket{\varnothing }\\
\text{By n.-p.}, & =U\ket{G_{\mu }}  \tensormu  \ket{\varnothing }\\
 & =U\ket{G_{\mu }}
\end{align*}
\end{proof}

Notice that a similar theorem was proven in the particular case of $\zeta^r_v\zeta_v$-causal operators over static networks first \cite{ArrighiUCAUSAL}, and then for node-preserving but connectivity-varying networks \cite{ArrighiQCGD}, a.k.a `quantum causal graph dynamics'. The point here is that the result carries through to arbitrary restrictions $\chi_v$ and over dynamical networks, both of which were non-trivial extensions. Moreover, from a methodological point of view, we used this theorem a test bench, to make sure that we had put together a set of mathematical tools that would be sufficient to combine and establish non-trivial results in this kernel of a quantum networks theory.

\section{Conclusion}\label{sec:conclusion}

{\em Summary of contributions.}

In this paper each node has an internal state and is identified by a unique name. The names are constructed by means of operators used for linking (e.g. node $\hyphenbullet y$ is understood as connected to node $y$), merging (e.g. nodes $u$ and $v$ may merge into node $u \vee v$), splitting (e.g. node $w$ may split into $w.l$ and $w.r$). The fact that the inverse of a merger operation is required to split $w=u \vee v$ back into $w.l=u$ and $w.r=v$ imposes equalities such as $(u \vee v).l=u$, leading to a simple name algebra. Notice that splits and merges are name-preserving (up to algebraic closure) and that the names of nodes are used to carry edge information.  

We place ourselves in the Hilbert space whose canonical basis are network configurations. We study operators over that space, including those leading to quantum superpositions of network configurations. 

We then introduce the notion of restriction, a function $\chi$ mapping $G$ a network into $G_\chi \subseteq G$ a subnetwork, fulfilling the idempotency condition $G_{\chi\chi}=G_{\chi}$. We also introduce extensible restrictions, fulfilling the the stronger condition that $G_\chi \subseteq H \subseteq G \Rightarrow H_\chi=G_\chi$ to ensure stability under taking unions $\chi \cup \zeta$ and neighbourhoods $\chi^r$. 

Each restriction leads to a partial trace $(\ket{G}\bra{H})_{|\chi}=\ket{G_\chi}\bra{H_\chi}\braket{H_{\overline{\chi}}|G_{\overline{\chi}}}$ which is completely positive trace-preserving, as well as name-preservation preserving. This generalized partial trace is robust, e.g. comprehension $\zeta \sqsubseteq \chi$ implies $(\rho_{|\chi})_{|\zeta}=\rho_\zeta$. The notion of comprehension is well-behaved over name-preserving states, e.g. for every extensible restriction $\zeta$, we have that $\zeta \sqsubseteq \zeta^r$.

Each restriction also defines a parallel composition a.k.a tensor product 
$$\ket{L}\tensorchi \ket{R}=\begin{cases}\ket{G} &\textrm{if}\; L=G_\chi\;\text{and}\; R=G_{\overline{\chi}}\\0 &\textrm{otherwise}\end{cases}$$
which is unambiguous, at the cost of zeroing out inconsistent terms. The notion of $\chi$-consistency becomes central: for instance the $\chi$-consistency-preserving operators defined by $ \bra{H} A\ket{G_{\chi }} \neq 0 \Rightarrow \ket{H} \tensorchi \ket{G_{\overline{\chi }}} \neq 0$ (plus the same for $A^\dagger$) are intuitively those which ``do not break the $\chi$-wall'' and hence gently slide along the tensor: $(A'\tensorchi B')(A\tensorchi B) = A'A\tensorchi B'B$.

Intuitively, local operators alter only a part $\chi$ of the network, and ignore the rest. In this paper we say that an operator is $\chi$-local whenever $\bra{H}A\ket{G}=\bra{H_\chi}A\ket{G_\chi}\braket{H_{\overline{\chi}}|G_{\overline{\chi}}}$ and prove the equivalence with the requirement that $A=A\tensorchi I$ and $\textrm{Tr}(A\rho)=\textrm{Tr}(A\rho_{|\chi})$. We say that an operator $A$ is strictly $\chi$-local whenever $A^\dagger A$ and $AA^\dagger$ are also $\chi$-local. Interestingly this corresponds to $A$ being both $\chi$-local and $\chi$-consistency-preserving, from which it follows that every $\chi$-local unitary is automatically $\chi$-consistency-preserving.

Intuitively, causal operators act over the entire network, yet respecting that effects on region $\zeta$ be fully determined by causes in region $\chi$. In this paper we say that operator $U$ is causal when $(U\rho U^\dagger)_{|\zeta}=(U\rho_{|\chi} U^\dagger)_{|\zeta}$ and prove equivalence with asking for $A$ $\zeta$-local to imply $U^\dagger A U$ $\chi$-local.

Causality refers to the physical principle according to which information propagates at a bounded speed, localizability refers to the principle that all must emerge constructively from underlying local mechanisms, that govern the interactions of closeby systems. The two notions of causality and localizability are related by our final theorem which shows that for fully quantum networks, causality implies localizability. 

{\em Further work.}

A number of mathematical results seem within reach and many more questions have yet to be considered, as we had to end somewhere.
\begin{itemize}
\item The examples provided in Figs \ref{fig:locality}, \ref{fig:causality} and \ref{fig:quantumcausality} are intuitive enough, but in all rigour they should be formalised and their corresponding properties proven. To this end we may need to show that causality is preserved under simple encodings/decodings such as splitting/merging all nodes.
\item Schmidt decomposition, purification, Stinespring dilation, are all fundamental results of quantum theory that crucially rely on properties of the tensor product. Important next steps should include phrasing and assessing the validity of these result in terms of generalized tensor products of quantum networks.
\item The results of this paper carry through to arbitrarily constrained configuration spaces ${\cal H}^{\cal C}$, except for Lemma \ref{lem:combiningrestrictions}, Props. \ref{prop:unitaryextension}, \ref{prop:causalextension}, \ref{prop:namewiseunitaryextension}, Th. \ref{th:blockdecomposition} and \ref{def:blockcecompositionnoancilla}. Even these ought to hold in many relevant cases of constrained configuration spaces, which ought to be investigated.
\item In our formalism edges are not given explicitly, rather they are induced from the information contained in the names of nodes. However, Appendix \ref{sec:renaminginvariance} suggests a precise alternative formalism where edges are given explicitly. Although likely heavier, the formalism ought to be evaluated: if successful, renaming-invariance would then imply full name-preservation, as used throughout the paper. 
\end{itemize}

Other mathematical challenges where left aside simply because they seemed difficult. For instance we have shown that if $U$ is $\forall
m \exists n\,\zeta^m \zeta^n$-causal, so is $U^2$, see Prop. \ref{prop:causalitycomposability}. But what if $U$ is just $\forall m\, \zeta^m \zeta$-causal? Does the final theorem presented help to provide an answer to this question? In the realm of static networks, so is $U^2$. This works because knowing $UA_vU^\dagger$ with $A_v$ is local upon some node $v$ in some region $R$, induces knowing $UA_R U^\dagger=\sum_k \prod_{v\in R} UA^{(k)}_vU^\dagger$ since $A_R$ has to be of the form $\sum_k \bigotimes_{v\in R} A^{(k)}_v$. But how does this generalise? Are splits and merges the only new generators of the quasi-local algebras ${\cal A}_v$? We leave such mathematical challenges as open problems.

{\em Perspectives.}

In the introduction we mentioned our original motivations for providing a theory of quantum networks:
\begin{itemize}
\item to provide rigorous kinematics and fully quantum dynamics for networks, equipped with rigorous notions of locality and causality;
\item for the sake of taking networks models of complex systems, into the quantum realm.  
\end{itemize}
For instance, in the field of Quantum Gravity we are now in a position to provide discrete-time versions of quantum graphity \cite{QuantumGraphity1,QuantumGraphity2}, thereby placing space and time on a equal footing, and demanding strict causality, instead of approximations à la Lieb-Robinson bound \cite{EisertSupersonic}. We are also in a better position to study the statuses of causality and unitarity in LQG \cite{RovelliLQG} and CDT \cite{LollCDT}: are these jeopardized by the Feynman path-based dynamics used in these theories?\\
Similarly, in the field of Quantum Computing, we now have a framework in which to model fully-quantum distributed computing devices, including dynamics over indefinite causal orders \cite{BruknerDynamics}, e.g. by means of causal unitary operators over superpositions of directed networks.

In order to reach this theory we have had to generalize the tensor product and partial trace in a rather modular and robust way, and this per se suggests a whole range of unexpected perspective applications:
\begin{description}
\item {\em Towards base-independence}. To make the notion of subsystem base-independent, algebraic quantum field theory tends to think of them as Von Neumann algebras instead \cite{BratteliRobinson}\cite{gogosiofunc}. This approach has been formalised in \cite{GogiosoChurch} in the the setting of categorical quantum mechanics, where the tensor product is a seen as a binary operator between states associated to commuting algebras. Such tensor products are therefore partially defined and remain abstract mathematical objects. The generalised tensor products of the present paper are everywhere defined and constructive, but they are base-dependent. It would be interesting to bridge the gap between these two notions to get the best of both worlds. One way to go about this is to use the Wedderburn-Artin theorem, which states that up to a unitary, Von Neumann algebras are direct sums of full matrix algebras tensored with the identity, i.e. of the form $A\tensorchi I$ for a well-chosen $\chi$. Another route to follow would be to take, as base-independent restrictions, any oblique projector, regardless of the graph structure.
\item {\em Decomposition techniques, causal-to-local}. These generalized operators were essential to the theorem representing causal unitary operators by means of local unitary gates. Many variants of this question are still open however, even for static networks over a handful of systems \cite{Beckman,SchumacherWestmoreland, Lorenz_2021}, as soon as we demand that the representation be exact. Recent approaches \cite{VanrietveldeRouted} to phrasing the answers to these questions make the case for annotating wires with a type system specifying which subspace will flow into them; this in turn has the flavour of a generalized tensor product. This suggests that the generalized operators, by means of their increased expressiveness, may be key to reexpress and prove a number of standing conjectures.
\item {\em Construction techniques, local-to-causal.} In \cite{ArrighiAQG} the authors provide a hands-on, concrete way of expressing a family of unitary evolutions over network configurations allowing for quantum superpositions of connectivities. One may wonder whether these addressable quantum gates, if extended to become able to split and merge, could be proven universal in the class of causal operators over quantum networks. 
\item {\em Fusion products between parts of constrained configuration spaces}. In most physical theories, the set of allowed configurations is constrained. For instance, charges and fields are constrained by the Gauss law. It follows that the tensor product between two regions of space may be ill-defined, for instance because both of them follow the Gauss law at the individual level, but not when they are placed next to one another. The hereby devised generalised tensor product is robust enough to handle these situations whilst preserving most desired algebraic properties, simply by sending them to the null vector. The gauge-invariant `fusion product' of \cite{FreidelFusion} focusses upon this issue in the continuum by following a different, gauge-invariant and partially-defined approach.
\item {\em Flexible notions of entanglement, between logical spaces}. A natural generalisation of product states emerges from this paper. Namely, $\ket{\phi}$ is a $\chi$-product state iff
$$\ket{\phi}=\ket{\psi}\tensorchi\ket{\psi'}\textrm{ and }||\ket{\phi}||=||\ket{\psi}||.||\ket{\psi'}||.$$
A contrario, non-$\chi$-product states are $\chi$-entangled states, thereby allowing us to define entanglement between left and right factors according to as specified by an almost arbitrary logical criterion $\chi$. 
For a bipartite pure state we may quantify this entanglement as the Von Neumann entropy of $\rho_{|\chi}$. It is our feeling that such generalised notions of entanglement will allow us to tackle scenarios where the standard notion fails to be defined.
\item {\em Modelling delocalized observers, quantum reference frames}. Decoherence theory \cite{PazZurek} models the observer as a quantum system interacting with others; and the post-measurement state as that obtained by tracing out the observer. But since the observer is quantum, it could be delocalized, raising the question of what it means to take the trace out then. Here we can model a delocalized observer by delocalized black particles, and trace them out. This ability to ``take the vantage point of delocalized quantum system'' is in fact a feature in common with quantum reference frames. 
\item {\em Ad hoc notions of causality, emergence of space}. The notion of $\chi\zeta$-causality allows us to define causality constraints according to almost arbitrary families of logical criteria $(\chi,\zeta)$. This includes scenarios where all black particles communicate whatever their network distance, say. In fact the very notion of network connectivity is arbitrary in the theory, i.e. $\zeta^r$ can in principle be redefined in order to better fit ad hoc causality constraints, possibly emerging in a similar way to pointer states in decoherence theory \cite{PazZurek}. 
\end{description}

\section*{Acknowledgements}

This publication was made possible through the support of the ID\# 61466 grant from the John Templeton Foundation, as part of the ``The Quantum Information Structure of Spacetime (QISS)'' Project (\href{qiss.fr}{qiss.fr}). The opinions expressed in this publication are those of the author(s) and do not necessarily reflect the views of the John Templeton Foundation. MW was supported by the EPSRC [grant number EP/L015242/1] and [grant number EP/W524335/1]. We wish to thank Kathleen Barsse and Augustin Vanrietvelde for helpful discussions helping us realise the importance of the restrictions that are not extensible.

\bibliographystyle{quantum} 
\bibliography{biblio}

\appendix

\section{Lemmas}\label{sec:lemmas}

\begin{lema}[Complement names]\label{lem:complementnames}
Let $G,H,\chi$ be such that $V(G) \hat{=} V(H)$ and $V(G_{\chi}) \hat{=} V(H_{\chi})$, then it follows that $V(G)\setminus V(G_{\chi}) \hat{=} V(H)\setminus V(H_{\chi})$.
\end{lema}
\begin{proof}
In this proof we write $V\wedge W$ if and only if ${\cal N}[V]\cap {\cal N}[W]\neq \varnothing$.

First note the following law \[(V \wedge W \textrm{ and } W \hat{=} Z ) \Rightarrow V \wedge Z\] It follows that if $(V(G)\setminus V(G_{\chi})) \wedge V(H_{\chi})$ then $(V(G)\setminus V(G_{\chi})) \wedge V(G_{\chi})$ which is in turn a contradiction with Eq. \eqref{eq:wellnamedness}.

Consider $u \in \mathcal{N}[V(G)\setminus V(G_{\chi})]$. Because $\mathcal{N}[V(G)] = \mathcal{N}[V(H)]$, this $u$ can be expressed by means of the operator $\vee$ applied on elements of $\mathcal{N}[V(H)]$. Suppose that $u$ lies beyond $\mathcal{N}[V(H)\setminus V(H_{\chi})]$, in $\mathcal{N}[V(H)] \setminus  \mathcal{N}[V(H)\setminus V(H_{\chi})]$. Then the expression for $u$ must contain at least one element of $v$ in $\mathcal{N}[V(H_{\chi})]$. As a consequence there exists $t \in \{l,r\}^{*}$ such that $v=u.t$. But because the operator $.t \in \{l,r\}^{*}$ preserves inclusion within $\mathcal{N}[V(G)\setminus V(G_{\chi})]$,  $v$ also lies in $\mathcal{N}[V(G)\setminus V(G_{\chi})]$. It follows that $(V(G)\setminus V(G_{\chi})) \wedge V(H_{\chi})$---leading to the contradiction of the previous paragraph. 

We conclude that the expression for $u$ in terms of elements of $\mathcal{N}[V(H)]$ must only include elements of $\mathcal{N}[V(H)\setminus V(H_{\chi})]$ and so $u$ itself lies inside $\mathcal{N}[V(H)\setminus V(H_{\chi})]$. By reversing the above it must be that every $u \in \mathcal{N}[V(H)\setminus V(H_{\chi})]$ also exists in $\mathcal{N}[V(G)\setminus V(G_{\chi})]$ and so $\mathcal{N}[V(G)\setminus V(G_{\chi})] = \mathcal{N}[V(H)\setminus V(H_{\chi})]$ in other words it must be the case that $V(G)\setminus V(G_{\chi}) \hat{=} V(H)\setminus V(H_{\chi})$.
\end{proof}

\begin{lema}[Tensor-bracket]\label{lem:tensorbracket}

For every restriction 
$ \chi : G\mapsto G_{\chi } \subseteq G $, 
inner products factorise with respect to $\chi$, i.e. 
$ \braket{H|G} =\braket{H_{\chi }|G_{\chi }}\braket{H_{\overline{\chi }}|G_{\overline{\chi }}}$. 
\end{lema}
\begin{proof}
Either of the RHS and LHS are either zero or one. The RHS is $1$ if and only if both $G_{\chi } =H_{\chi }$ and $G_{\overline{\chi }} =H_{\overline{\chi }}$. Since $G:=G_{\chi } \cup G_{\overline{\chi }}$ and $H:=H_{\chi } \cup H_{\overline{\chi }}$ this is equivalent to simply requiring that $G = H$. The LHS being $\braket{G|H}$ is also $1$ if and only if $G = H$, so the RHS is always equal to the LHS.
\end{proof}

\begin{lema}[Properties of restrictions]\label{lem:restrictions}
If $ \chi $ be a restriction, then $ \chi \overline{\chi } =\emptyset $.
If $\chi$ is an extensible restriction, then it is a restriction, i.e. $ \chi \chi =\chi $.
\end{lema}
\begin{proof}
The first point is because $G_{\chi \overline{\chi }} = G_{\chi } \setminus G_{\chi \chi } =G_{\chi } \setminus G_{\chi } =\emptyset $. \\
Next, for any extensible restriction $\chi $ then by definition $G_{\chi } \subseteq \ H\ \subseteq \ G\ \Rightarrow \ H_{\chi } =G_{\chi }$, noting that the assignment $H=G_{\chi }$ always satisfies the LHS of this statement the right hand side must also hold for $H = G_{\chi}$, meaning that $G_{\chi } =G_{\chi \chi }$.
\end{proof}

\begin{lema}[Special restrictions]\label{lem:specialrestrictions}
Every order-preserving idempotent function $ \nu: G\mapsto G_{\nu } \subseteq G$ is an extensible restriction and as a corollary every pointwise function $ \mu $ (and its complement function $ \overline{\mu }$) is an extensible restriction.
\end{lema}
\begin{proof}
Let $\nu $ be an order-preserving idempotent function. It follows that for any $G_{\nu } \subseteq H\subseteq G$ then since $\nu $ is order-preserving $G_{\nu \nu } \subseteq H_{\nu } \subseteq G_{\nu }$ and so since $\nu $ is idempotent $G_{\nu } \subseteq H_{\nu } \subseteq G_{\nu }$, in other words $G_{\nu } =H_{\nu }$.

Every pointwise function preserves the order: $G\subseteq H\ \Rightarrow \bigcup _{\sigma .v\ \in \ G}\{\sigma .v\}_{\mu } \ \subseteq \bigcup _{\sigma .v\ \in \ H}\{\sigma .v\}_{\mu } \ \Rightarrow G_{\mu } \subseteq H_{\mu }$. Furthermore every pointwise function is idempotent. Since every order-preserving idempotent is an extensible restriction it then follows that every pointwise function is an extensible restriction. 

Finally for every point-wise restriction $\mu $ its complement function $\overline{\mu }$ for single-node graphs satisfies \ $\{\sigma .v\}_{\overline{\mu }} \ :=\{\sigma .v\} \ \setminus \ \{\sigma .v\}_{\mu } \ $ and for generic graphs satisfies \[G_{\overline{\mu }} :=G\ \setminus \bigcup _{\sigma .v\ \in \ G}\{\sigma .v\}_{\mu } \ =\bigcup _{\sigma .v\ \in \ G}\{\sigma .v\} \ \setminus \ \{\sigma .v\}_{\mu } \ =\bigcup _{\sigma .v\ \in \ G}\{\sigma .v\}_{\overline{\mu }}\]. It follows that $\overline{\mu }$ is also pointwise function. By the previous section of the lemma $\overline{\mu }$ must be a (pointwise) restriction.
\end{proof}

Caution: the following works over ${\cal H}$, but it may not hold for an arbitrarily constrained configuration space ${\cal H}^{\cal C}$ and extensible restrictions over it.
\begin{lema}[Properties of extensible restrictions]\label{lem:combiningrestrictions}
Let $ \chi , \zeta $ be extensible restrictions, then
\begin{itemize}
\item $\zeta^{r}$ is an extensible restriction.
\item $\chi \cup \zeta $ is an extensible restriction.
\end{itemize}
Furthermore let $ \mu $ be a pointwise function, then $ \mu \chi $ and $\xi :=\mu \chi \cup \overline{\mu }\zeta $ are extensible restrictions. The extensible restrictions $\mu,\xi$ and their complements furthermore commute, satisfying $ [ \mu ,\xi ] =[\overline{\mu } ,\xi ] =[ \mu ,\overline{\xi }] =[\overline{\mu } ,\overline{\xi }] =0$.
\end{lema}
\begin{proof}
Let $G_{\zeta^{r}} \subseteq H\subseteq G$, then \ $G_{\zeta } \subseteq G_{\zeta^{r}} \subseteq H\subseteq G$. Because $\zeta $ is an extensible restriction, $H_{\zeta } =G_{\zeta } =K$. Since $H\subseteq G$, we have that $H_{\zeta^{r}} \subseteq G_{\zeta^{r}}$ because the neighbours of $K$ in $H$ are also in $G$. Since $G_{\zeta ^{r}} \subseteq H$, we have that $G_{\zeta^{r}} \subseteq H_{\zeta^{r}}$ because the neighbours of $K$ in $G$ are also in $H$. Since $H_{\zeta^{r}} \subseteq G_{\zeta^{r}}$ and $G_{\zeta^{r}} \subseteq H_{\zeta^{r}}$ then $H_{\zeta^{r}} =G_{\zeta^{r}}$.

Let $G_{\chi }{}_{\cup }{}_{\zeta } \subseteq H\subseteq G$, then since $G_{\chi } \ \subseteq G_{\chi }{}_{\cup }{}_{\zeta }$ we have $G_{\chi } \ \subseteq G_{\chi }{}_{\cup }{}_{\zeta } \subseteq H\subseteq G\ $ which since $\chi$ is an extensible restriction implies that $ H_{\chi } \ =\ G_{\chi }$. Similarly $G_{\zeta } \ \subseteq G_{\chi }{}_{\cup }{}_{\zeta } \subseteq H\subseteq G\ $ implies $ H_{\zeta } \ =\ G_{\zeta }$. The above equalities imply equalities for the union, $H_{\chi }{}_{\cup }{}_{\zeta } \ =\ H_{\chi } \ \cup \ H_{\zeta } \ =G_{\chi } \ \cup \ G_{\zeta } \ =G_{\chi }{}_{\cup }{}_{\zeta }$.

For any pointwise restriction $\mu $ and extensible restriction $\chi $ we have $\mu \chi \mu \ =\ \mu \chi $. As a result:
\begin{align*}
    & G_{\mu \chi } \ \subseteq   H\ \subseteq G \\
    & \Rightarrow \  G_{\mu \chi \mu } \ \subseteq H_{\mu } \ \subseteq G_{\mu } \\
    & \Rightarrow \  G_{\mu \chi } \ \subseteq H_{\mu } \ \subseteq G_{\mu } \\
    & \Rightarrow \  H_{\mu \chi } \ =\ G_{\mu \chi }
\end{align*}
 
Let $\xi :=\mu \chi \cup \overline{\mu }\zeta $ and notice that $\overline{\xi } =\ \mu \overline{\chi} \cup \overline{\mu }\overline{\zeta }$. The function $\xi $ is an extensible restriction since $\mu \chi $ and $\overline{\mu } \zeta $ are extensible restrictions by the previous part, and their union is an extensible restriction by the second part.

We now show that $[ \mu ,\xi ] =[\overline{\mu } ,\xi ] =[ \mu ,\overline{\xi }] =[\overline{\mu } ,\overline{\xi }] =0$. First since $\mu $ is pointwise, then for any $\nu$ we have $\mu \nu \mu =\mu \nu $. Similarly, $\overline{\mu } \nu \overline{\mu } =\overline{\mu } \nu $ and $\mu \nu \overline{\mu } =\overline{\mu } \nu \mu =\emptyset $. We also have $( \nu \cup \nu ') \mu =\nu \mu \cup \nu '\mu $ and as always $\mu ( \nu \cup \nu ') =\mu \nu \cup \mu \nu '$.  This is enough to derive the commutation rules by the following steps:

$\xi \mu =\mu \chi \mu \cup \overline{\mu } \zeta \mu =\mu \chi $ \quad and \quad $\mu \xi =\mu \mu \chi \cup \mu \overline{\mu } \zeta =\mu \chi $ 

$\xi \overline{\mu } =\mu \chi \overline{\mu } \cup \overline{\mu } \zeta \overline{\mu } =\overline{\mu } \zeta $ \quad and \quad $\overline{\mu } \xi =\overline{\mu } \mu \chi \cup \overline{\mu }\overline{\mu } \zeta =\overline{\mu } \zeta $

$\overline{\xi} \mu =\mu \overline{\chi} \mu \cup \overline{\mu }\overline{\zeta } \mu =\mu \overline{\chi} $ \ and \ $\mu \overline{\xi } =\mu \mu \overline{\chi} \cup \mu \overline{\mu }\overline{\zeta } =\mu \overline{\chi} $

$\overline{\xi }\overline{\mu } =\mu \overline{\chi} \overline{\mu } \cup \overline{\mu }\overline{\zeta }\overline{\mu } =\overline{\mu }\overline{\zeta }$ \quad and \quad $\overline{\mu }\overline{\xi } =\overline{\mu } \mu \overline{\chi} \cup \overline{\mu }\overline{\mu }\overline{\zeta } =\overline{\mu }\overline{\zeta }$
\end{proof}

\begin{lema}[Tensor]\label{lem:tensor}

For all $ A:\mathcal{G}^{4}\rightarrow \mathbb{C}$,
\begin{align*}
\sum _{G,H\in \mathcal{G}} A_{G_{\chi } H_{\chi }}{}_{G_{\overline{\chi }} H_{\overline{\chi }}} \ \ket{G}\bra{H} & =\sum _{ \begin{array}{l}
G ,H \in \mathcal{G}_{\chi } \ \\
G' ,H' \in \mathcal{G}_{\overline{\chi }}
\end{array}} A_{GHG'H'} \ \ket{G}\bra{H} \tensorchi \ket{G'}\bra{H'}\\
 & =\sum _{ \begin{array}{l}
G,H \in \mathcal{G}_{\chi } \ \\
G' ,H' \in \mathcal{G}
\end{array}} A_{GHG'H'} \ \ket{G}\bra{H} \tensorchi \ket{G'}\bra{H'} \ \\
 & =\sum _{ \begin{array}{l}
G,H \in \mathcal{G} \ \\
G' ,H' \in \mathcal{G}
\end{array}} A_{GHG'H'} \ \ket{G}\bra{H} \tensorchi \ket{G'}\bra{H'}
\end{align*}
In particular,

$ \sum _{G,H\in \mathcal{G}} A_{G_{\chi } H_{\chi }} B_{G_{\overline{\chi }} H_{\overline{\chi }}} \ \ket{G}\bra{H} \ =A\tensorchi B\ $\qquad $ I=I_{\chi } \tensorchi I_{\overline{\chi }}$ \qquad $ A\tensorchi I_{\overline{\chi }} =A\tensorchi I$ 
\end{lema}

\begin{proof}

[First part]
\begin{align*}
\sum _{G,H\in \mathcal{G}} A_{G_{\chi } H_{\chi }}{}_{G_{\overline{\chi }} H_{\overline{\chi }}} \ \ket{G}\bra{H} & =\sum _{G ,H \in \mathcal{G}} A_{G_{\chi } H_{\chi }}{}_{G_{\overline{\chi }} H_{\overline{\chi }}} \ \ket{G_{\chi }}\bra{H_{\chi }} \tensorchi \ket{G_{\overline{\chi }}}\bra{H_{\overline{\chi }}}\\
 & =\sum _{ \begin{array}{l}
G ,H \in \mathcal{G}_{\chi } \ G' ,H' \in \mathcal{G}_{\overline{\chi }}\\
\ket{G} \tensorchi \ket{G'} \neq 0\\
\ket{H} \tensorchi \ket{H'} \neq 0
\end{array}} A_{GHG'H'} \ \ket{G}\bra{H} \tensorchi \ket{G'}\bra{H'}\\
 & =\sum _{ \begin{array}{l}
G ,H \in \mathcal{G}_{\chi } \ \\
G' ,H' \in \mathcal{G}_{\overline{\chi }}
\end{array}} A_{GHG'H'} \ \ket{G}\bra{H} \tensorchi \ket{G'}\bra{H'}
\end{align*}
where the last line is obtained by summing also over those $G,G',H,H'$ that are such that $\ket{G} \tensorchi \ket{G'} =0$ or $\ket{H} \tensorchi \ket{H'} =0$, since for them $\ket{G}\bra{H} \tensorchi \ket{G'}\bra{H'} =0$. \ For the same reason we can sum also over those $G',H'\in \mathcal{G} \setminus \mathcal{G}_{\chi }$ and obtain the second stated equality, and over those $G,H\in \mathcal{G} \setminus \mathcal{G}_{\chi }$ to obtain the third.

[Second part]
\begin{align*}
 & \sum _{G,H\in \mathcal{G}} A_{G_{\chi } H_{\chi }} B_{G_{\overline{\chi }} H_{\overline{\chi }}} \ \ket{G}\bra{H} \ =\ \sum _{ \begin{array}{l}
G ,H \in \mathcal{G} \ \\
G' ,H' \in \mathcal{G}
\end{array}} A_{GH} B_{G'H'} \ \ket{G}\bra{H} \tensorchi \ket{G'}\bra{H'} =A\tensorchi B\\
I & =\ \sum _{G,H\in \mathcal{G}} \delta _{G_{\chi } H_{\chi }} \delta _{G_{\overline{\chi }} H_{\overline{\chi }}} \ \ket{G}\bra{H} =I_{\chi } \tensorchi I_{\overline{\chi }}\\
A\tensorchi I_{\overline{\chi }} & =\sum _{ \begin{array}{l}
G ,H \in \mathcal{G}_{\chi } \ \\
G' ,H' \in \mathcal{G}_{\overline{\chi }}
\end{array}} A_{GH} \delta _{G'H'} \ \ket{G}\bra{H} \tensorchi \ket{G'}\bra{H'} =\sum _{ \begin{array}{l}
G ,H \in \mathcal{G}_{\chi } \ \\
G' ,H' \in \mathcal{G}
\end{array}} A_{GH} \delta _{G'H'} \ \ket{G}\bra{H} \tensorchi \ket{G'}\bra{H'}\\
 & =A\tensorchi I_{\overline{\chi }} =A\tensorchi I
\end{align*}
\end{proof}

\begin{lema}[Tensor-tensor]\label{lem:tensortensor}

Write $ E_{( k)}^{GH} :=\ket{G^{( k)}}\bra{H^{( k)}}$. For all $ A:\mathcal{G}^{8}\rightarrow \mathbb{C}$, if $ [ \chi ,\zeta ] =[\overline{\chi } ,\zeta ] =[ \chi ,\overline{\zeta }] =[\overline{\chi } ,\overline{\zeta }] =0$,
\begin{gather*}
\sum _{ \begin{array}{l}
G^{( 0)} ,G^{( 1)} \dotsc \ \in \mathcal{G}\\
H^{( 0)} ,H^{( 1)} \dotsc \ \in \mathcal{G}
\end{array}} A_{G^{( 0)} H^{( 0)} G^{( 1)} H^{( 1)} \dotsc } \ \left( E_{( 0)}^{GH}  \tensorzeta   E_{( 1)}^{GH}\right) \tensorchi \left( E_{( 2)}^{GH}  \tensorzeta   E_{( 3)}^{GH}\right)\\
=\ \sum _{ \begin{array}{l}
G^{( 0)} ,G^{( 1)} \dotsc \ \in \mathcal{G}\\
H^{( 0)} ,H^{( 1)} \dotsc \ \in \mathcal{G}
\end{array}} A_{G^{( 0)} H^{( 0)} G^{( 1)} H^{( 1)} \dotsc } \ \left( E_{( 0)}^{GH} \tensorchi E_{( 2)}^{GH}\right)  \tensorzeta   \left( E_{( 1)}^{GH} \tensorchi E_{( 3)}^{GH}\right)
\end{gather*}
\end{lema}
\begin{proof}

\begin{align*}
 & \left( E_{( 0)}^{GH}  \tensorzeta   E_{( 1)}^{GH}\right) \tensorchi \left( E_{( 2)}^{GH}  \tensorzeta   E_{( 3)}^{GH}\right) \ \neq 0\\
\Leftrightarrow  & \exists G^{( 01)} ,H^{( 01)} /\ G^{( 0)} =G_{\zeta }^{( 01)} ,G^{( 1)} =G_{\overline{\zeta }}^{( 01)} ,H^{( 0)} =\dotsc ,\ E_{( 01)}^{GH} \tensorchi \left( E_{( 2)}^{GH}  \tensorzeta   E_{( 3)}^{GH}\right) \ \neq 0\\
\Leftrightarrow  & \exists G^{( 01)} ,H^{( 01)} ,G^{( 23)} ,H^{( 23)} /\ G^{( 0)} =G_{\zeta }^{( 01)} ,\dotsc ,G^{( 2)} =G_{\zeta }^{( 23)} ,\dotsc \ E_{( 01)}^{GH} \tensorchi E_{( 23)}^{GH} \neq 0\\
\Leftrightarrow  & \exists G ,H/\ G^{( 0)} =G_{\zeta }^{( 01)} ,G^{( 01)} =G_{\chi } ,\dotsc ,G^{( 2)} =G_{\zeta }^{( 23)} ,G^{( 23)} =G_{\overline{\chi }} ,\dotsc \\
\Leftrightarrow  & \exists G ,H/\ G^{( 0)} =G_{\chi \zeta } ,G^{( 1)} =G_{\chi \overline{\zeta }} ,G^{( 2)} =G_{\overline{\chi } \zeta } ,G^{( 3)} =G_{\overline{\chi }\overline{\zeta }} ,H^{( 0)} =\dotsc \ \\
\Leftrightarrow  & \exists G ,H/\ G^{( 0)} =G_{\zeta \chi } ,G^{( 1)} =G_{\overline{\zeta } \chi } ,G^{( 2)} =G_{\zeta \overline{\chi }} ,G^{( 3)} =G_{\overline{\zeta }\overline{\chi }} ,H^{( 0)} =\dotsc \\
\Leftrightarrow  & \left( E_{( 0)}^{GH} \tensorchi E_{( 2)}^{GH}\right)  \tensorzeta   \left( E_{( 1)}^{GH} \tensorchi E_{( 3)}^{GH}\right) \ \neq 0
\end{align*}
Moreover, when they are non-zero, 
\begin{align*}
 & \left( E_{( 0)}^{GH}  \tensorzeta   E_{( 1)}^{GH}\right) \tensorchi \left( E_{( 2)}^{GH}  \tensorzeta   E_{( 3)}^{GH}\right) \ \\
= & \left(\ket{G_{\chi \zeta }}\bra{H_{\chi \zeta }}  \tensorzeta   \ket{G_{\chi \overline{\zeta }}}\bra{H_{\chi \overline{\zeta }}}\right) \tensorchi \left(\ket{G_{\overline{\chi } \zeta }}\bra{H_{\overline{\chi } \zeta }}  \tensorzeta   \ket{G_{\overline{\chi }\overline{\zeta }}}\bra{H_{\overline{\chi }\overline{\zeta }}}\right) \ \\
= & \ket{G}\bra{H}\\
= & \left(\ket{G_{\zeta \chi }}\bra{H_{\zeta \chi }} \tensorchi \ket{G_{\zeta \overline{\chi }}}\bra{H_{\zeta \overline{\chi }}}\right)  \tensorzeta   \left(\ket{G_{\overline{\zeta } \chi }}\bra{H_{\overline{\zeta } \chi }} \tensorchi \ket{G_{\overline{\zeta }\overline{\chi }}}\bra{H_{\overline{\zeta }\overline{\chi }}}\right)\\
= & \left(\ket{G_{\chi \zeta }}\bra{H_{\chi \zeta }} \tensorchi \ket{G_{\overline{\chi } \zeta }}\bra{H_{\overline{\chi } \zeta }}\right)  \tensorzeta   \left(\ket{G_{\chi \overline{\zeta }}}\bra{H_{\chi \overline{\zeta }}} \tensorchi \ket{G_{\overline{\chi }\overline{\zeta }}}\bra{H_{\overline{\chi }\overline{\zeta }}}\right)\\
= & \left( E_{( 0)}^{GH} \tensorchi E_{( 2)}^{GH}\right)  \tensorzeta   \left( E_{( 1)}^{GH} \tensorchi E_{( 3)}^{GH}\right)
\end{align*}
So,
\begin{align*}
 & A_{G^{( 0)} H^{( 0)} \ G^{( 1)} H^{( 1)} \dotsc } \ \left( E_{( 0)}^{GH}  \tensorzeta   E_{( 1)}^{GH}\right) \tensorchi \left( E_{( 2)}^{GH}  \tensorzeta   E_{( 3)}^{GH}\right)\\
= & A_{G^{( 0)} H^{( 0)} \ G^{( 1)} H^{( 1)} \dotsc } \ \left( E_{( 0)}^{GH} \tensorchi E_{( 2)}^{GH}\right)  \tensorzeta   \left( E_{( 1)}^{GH} \tensorchi E_{( 3)}^{GH}\right)\\
= & A_{G^{( 0)} H^{( 0)} \ G^{( 1)} H^{( 1)} \dotsc }\ket{G}\bra{H} \quad \text{or}\quad 0.
\end{align*}
\end{proof}

In particular,
\begin{equation*}
( A \tensorzeta   B) \tensorchi ( C \tensorzeta   D) =\ ( A\tensorchi C)  \tensorzeta   ( B\tensorchi D) \ 
\end{equation*}
Indeed write $A=\sum _{G^{( 0)} \in \mathcal{G}} A_{G^{( 0)} H^{( 0)}}\ket{G^{( 0)}}\bra{G^{( 0)}}$, $B=\dotsc $ 
and let $A'_{G^{( 0)} H^{( 0)} G^{( 1)} H^{( 1)} \dotsc } :=A_{G^{( 0)} H^{( 0)}} B_{G^{( 1)} H^{( 1)}} \dotsc $ in the Lemma.

\begin{lema}[Trace-trace]\label{lem:tracetrace}
If $\zeta \sqsubseteq \chi $, then $(\rho _{|\chi })_{|\zeta } =\rho _{|\zeta }$ and any $\zeta $-local $A$ is also $\chi$-local.
\end{lema}
\begin{proof}
[Partial trace]
\begin{align*}
\left(\left(\ket{G}\bra{H}\right)_{|\chi }\right)_{|\zeta } & =\left(\ket{G_{\chi }}\bra{H_{\chi }}\right)_{|\zeta }\braket{H_{\overline{\chi }}|G_{\overline{\chi }}}\\
 & =\ket{G_{\chi \zeta }}\bra{H_{\chi \zeta }}\braket{H_{\chi \overline{\zeta }}|G_{\chi \overline{\zeta }}}\braket{H_{\overline{\chi }}|G_{\overline{\chi }}}\\
\text{When } \zeta \sqsubseteq \chi . & =\ket{G_{\zeta }}\bra{H_{\zeta }}\braket{H_{\overline{\zeta }}|G_{\overline{\zeta }}} =\left(\ket{G}\bra{H}\right)_{|\zeta }
\end{align*}

[Locality] 
\begin{align*}
\bra{H} A\ket{G} & =\bra{H_{\zeta }} A\ket{G_{\zeta }}\braket{H_{\overline{\zeta }}|G_{\overline{\zeta }}}\\
\text{When } \zeta \sqsubseteq \chi . & =\bra{H_{\chi \zeta }} A\ket{G_{\chi \zeta }}\braket{H_{\chi \overline{\zeta }}|G_{\chi \overline{\zeta }}}\braket{H_{\overline{\chi }}|G_{\overline{\chi }}}\\
\text{By} \ \zeta \text{-loc.} & =\bra{H_{\chi }} A\ket{G_{\chi }}\braket{H_{\overline{\chi }}|G_{\overline{\chi }}}
\end{align*}
\end{proof}

\begin{lema}[Tensor-trace 1]\label{lem:tensortrace1}
If $ \rho ,\sigma $ $ \chi $-consistent, $ ( \rho \tensorchi \sigma )_{|\chi } =\ \rho \ \sigma _{|\emptyset }$.\\
If $ \zeta \sqsubseteq \chi $, $ \rho ,\sigma $ $ \chi $-consistent, $ ( \rho \tensorchi \sigma )_{|\zeta } =\rho _{|\zeta } \ \sigma _{|\emptyset }$.
\end{lema}
\begin{proof}
[First part]

Notice that $\ket{G} \tensorchi \ket{G'} \neq 0$ implies $\left(\ket{G} \tensorchi \ket{G'}\right)_{\chi } =\ket{G}$ and $\left(\ket{G} \tensorchi \ket{G'}\right)_{\overline{\chi }} =\ket{G'}$.

Similarly if $\ket{G} \tensorchi \ket{G'} \neq 0\neq \ket{H} \tensorchi \ket{H'}$,
\begin{align*}
\left(\ket{G}\bra{H} \tensorchi \ket{G'}\bra{H'}\right)_{\chi } & =\left(\ket{G} \tensorchi \ket{G'}\right)_{\chi }\left(\bra{H} \tensorchi \bra{H'}\right)_{\chi }\\
 & =\ket{G}\bra{H}\\
\left(\ket{G}\bra{H} \tensorchi \ket{G'}\bra{H'}\right)_{\overline{\chi }} & =\left(\ket{G} \tensorchi \ket{G'}\right)_{\overline{\chi }}\left(\bra{H} \tensorchi \bra{H'}\right)_{\overline{\chi }}\\
 & =\ket{G'}\bra{H'}\\
\left(\ket{G}\bra{H} \tensorchi \ket{G'}\bra{H'}\right)_{|\chi } & =\left(\ket{G} \tensorchi \ket{G'}\right)_{\chi }\left(\bra{H} \tensorchi \bra{H'}\right)_{\chi }\left(\bra{H} \tensorchi \bra{H'}\right)_{\overline{\chi }}\left(\ket{G} \tensorchi \ket{G'}\right)_{\overline{\chi }}\\
 & =\ket{G}\bra{H}\braket{H'|G'}
\end{align*}
Next assume $\rho ,\ \sigma $ $\chi $-consistent, i.e. $\rho _{GH} \sigma _{G'H'} \neq 0$ implies $\ket{G} \tensorchi \ket{G'} \neq 0\neq \ket{H} \tensorchi \ket{H'}$.
\begin{align*}
\rho \tensorchi \sigma  & \ =\ \sum _{G,H,G',H'\ \in \ \mathcal{G}} \rho _{GH} \sigma _{G'H'} \ \ket{G}\bra{H} \tensorchi \ket{G'}\bra{H'} \ \\
( \rho \tensorchi \sigma )_{|\chi } & =\ \sum _{ \begin{array}{l}
G,H,G',H'\ \in \ \mathcal{G}\\
\ket{G} \tensorchi \ket{G'} \ \neq \ 0\\
\ket{H} \tensorchi \ket{H'} \ \neq \ 0
\end{array}} \rho _{GH} \sigma _{G'H'} \ \left(\ket{G}\bra{H} \tensorchi \ket{G'}\bra{H'}\right)_{|\chi }\\
 & =\sum _{ \begin{array}{l}
G,H,G',H'\ \in \ \mathcal{G}\\
\ket{G} \tensorchi \ket{G'} \ \neq \ 0\\
\ket{H} \tensorchi \ket{H'} \ \neq \ 0
\end{array}} \rho _{GH} \sigma _{G'H'} \ \ket{G}\bra{H} \ \braket{H'|G'}\\
\text{By consistency } & =\sum _{G,H,G',H'\ \in \ \mathcal{G}} \rho _{GH} \sigma _{G'H'} \ \ket{G}\bra{H} \ \braket{H'|G'}\\
 & =\sum _{G,H\ \in \ \mathcal{G}} \rho _{GH}\ket{G}\bra{H} \ \sum _{G',H'\ \in \ \mathcal{G}} \sigma _{G'H'} \ \braket{H'|G'}\\
 & =\ \rho \ \sigma _{|\emptyset }
\end{align*}
[Second part]

Next assume $\zeta \sqsubseteq \chi $ \ and $\rho ,\sigma $ \ $\chi $-consistent. By Lem. \ref{lem:tracetrace},
\begin{align*}
( \rho \tensorchi \sigma )_{|\zeta } & =(( \rho \tensorchi \sigma )_{|\chi })_{|\zeta }\\
\text{By first part.} & =( \rho \ \sigma _{|\emptyset })_{|\zeta }\\
 & =\rho _{|\zeta } \ \sigma _{|\emptyset }
\end{align*}
\end{proof}

\begin{lema}[Tensor-trace 2]\label{lem:tensortrace2}
If $ [ \chi ,\zeta ] =[\overline{\chi } ,\zeta ] =[ \chi ,\overline{\zeta }] =[\overline{\chi } ,\overline{\zeta }] =0$, and $ \rho ,\sigma $ $ \chi $-consistent,
\begin{equation*}
( \rho \tensorchi \sigma )_{|\zeta } =\ \rho _{|\zeta } \tensorchi \sigma _{|\zeta } \ 
\end{equation*}
\end{lema}
\begin{proof}
Assume $\rho ,\ \sigma $ $\chi $-consistent, i.e. $\rho _{GH} \sigma _{G'H'} \neq 0$ implies $\ket{G} \tensorchi \ket{G'} \neq 0\neq \ket{H} \tensorchi \ket{H'}$.
\begin{align*}
\rho \tensorchi \sigma  & =\ \sum _{G ,H, G' ,H' \in \mathcal{G}} \rho _{GH} \sigma _{G'H'} \ \ket{G}\bra{H} \tensorchi \ket{G'}\bra{H'}\\
\text{By Lem. \ref{lem:tensor}} & =\sum _{G ,H \in \mathcal{G}} \rho _{G_{\chi } H_{\chi }} \sigma _{G_{\overline{\chi }} H_{\overline{\chi }}} \ \ket{G}\bra{H}\\
( \rho \tensorchi \sigma )_{|\zeta } & =\sum _{G ,H \in \mathcal{G}} \rho _{G_{\chi } H_{\chi }} \sigma _{G_{\overline{\chi }} H_{\overline{\chi }}} \ \ket{G_{\zeta }}\bra{H_{\zeta }} \ \braket{H_{\overline{\zeta }}|G_{\overline{\zeta }}}\\
\text{By Lem. \ref{lem:tensorbracket}} & =\sum _{G ,H \in \mathcal{G}} \rho _{G_{\chi } H_{\chi }} \sigma _{G_{\overline{\chi }} H_{\overline{\chi }}} \ \ket{G_{\zeta \chi }}\bra{H_{\zeta \chi }} \tensorchi \ket{G_{\zeta \overline{\chi }}}\bra{H_{\zeta \overline{\chi }}} \ \braket{H_{\overline{\zeta } \chi }|G_{\overline{\zeta } \chi }}\braket{H_{\overline{\zeta }\overline{\chi }}|G_{\overline{\zeta }\overline{\chi }}}\\
     & =\sum _{G ,H \in \mathcal{G}} \rho _{G_{\chi } H_{\chi }}\ket{G_{\zeta \chi }}\bra{H_{\zeta \chi }} \ \braket{H_{\overline{\zeta } \chi }|G_{\overline{\zeta } \chi }} \tensorchi \sigma _{G_{\overline{\chi }} H_{\overline{\chi }}}\ket{G_{\zeta \overline{\chi }}}\bra{H_{\zeta \overline{\chi }}} \ \braket{H_{\overline{\zeta }\overline{\chi }}|G_{\overline{\zeta }\overline{\chi }}}\\
\text{By commut. } & =\sum _{G ,H \in \mathcal{G}} \rho _{G_{\chi } H_{\chi }}\ket{G_{\chi \zeta }}\bra{H_{\chi \zeta }} \ \braket{H_{\chi \overline{\zeta }}|G_{\chi \overline{\zeta }}} \tensorchi \sigma _{G_{\overline{\chi }} H_{\overline{\chi }}}\ket{G_{\overline{\chi } \zeta }}\bra{H_{\overline{\chi } \zeta }} \ \braket{H_{\overline{\chi }\overline{\zeta }}|G_{\overline{\chi }\overline{\zeta }}}\\
 & =\sum _{G ,H \in \mathcal{G}} \ \left( \rho _{G_{\chi } H_{\chi }}\ket{G_{\chi }}\bra{H_{\chi }}\right)_{|\zeta } \tensorchi \left( \sigma _{G_{\overline{\chi }} H_{\overline{\chi }}}\ket{G_{\overline{\chi }}}\bra{H_{\overline{\chi }}}\right)_{|\zeta } \ \\
 & =\sum _{ \begin{array}{l}
G ,H,G',H' \in \mathcal{G}\\
\ket{G} \tensorchi \ket{G'} \ \neq \ 0\\
\ket{H} \tensorchi \ket{H'} \ \neq \ 0
\end{array}} \ \left( \rho _{G H}\ket{G}\bra{H}\right)_{|\zeta } \tensorchi \left( \sigma _{G' H'}\ket{G'}\bra{H'}\right)_{|\zeta } \ \\
\text{By consistency} & =\sum _{G ,H,G',H' \in \mathcal{G}} \ \left( \rho _{G H}\ket{G}\bra{H}\right)_{|\zeta } \tensorchi \left( \sigma _{G' H'}\ket{G'}\bra{H'}\right)_{|\zeta } \ \\
 & =\ \rho _{|\zeta } \tensorchi \sigma _{|\zeta } \ 
\end{align*}
\end{proof}

\begin{lema}[Interchange laws]\label{lem:interchangelaws}
$( A\tensorchi I)\ket{G} =A\ket{G_{\chi }} \tensorchi \ket{G_{\overline{\chi }}}$.\\
If $ A$ $ \chi $-consistent-preserving or $ B$ $ \overline{\chi }$-consistent-preserving, $ ( A\tensorchi I)( I\tensorchi B) =( A\tensorchi B)$.\\
If $ A$, $ A'$ $ \chi $-consistent-preserving, $ ( A'\tensorchi I)( A\tensorchi I) =( A'A\tensorchi I)$.\\
If $ A$, $ A'$ $ \chi $-consistent-preserving, $ A'A$ is $ \chi $-consistent-preserving.\\
If moreover $ B$, $ B'$ are $ \chi $-consistent-preserving, $( A'\tensorchi B')( A\tensorchi B) =\ A'A\tensorchi B'B$.
\end{lema}
\begin{proof}

[First part]
\begin{align*}
A\tensorchi I & =\sum _{G' ,H ,K\ \in \ \mathcal{G}} A_{G' H} \ \ket{G'}\bra{H} \tensorchi \ket{K}\bra{K}\\
 & =\sum _{G' ,H ,K\ \in \ \mathcal{G}} A_{G' H} \ \left(\ket{G'} \tensorchi \ket{K}\right)\left(\bra{H} \tensorchi \bra{K}\right)\\
( A\tensorchi I)\ket{G} & =\sum _{G' ,H ,K\ \in \ \mathcal{G}} A_{G' H} \ \left(\ket{G'} \tensorchi \ket{K}\right)\left(\bra{H} \tensorchi \bra{K}\right)\left(\ket{G_{\chi }} \tensorchi \ket{G_{\overline{\chi }}}\right)\\
 & =\sum _{G'\ \in \ \mathcal{G}} A_{G'G_{\chi }} \ \left(\ket{G'} \tensorchi \ket{G_{\overline{\chi }}}\right)\\
 & =\left( A\ket{G_{\chi }}\right) \tensorchi \ket{G_{\overline{\chi }}}
\end{align*}

[Second part]
\begin{align*}
I\tensorchi B & =\sum _{ \begin{array}{l}
G^{( 1)} ,H^{( 1)} ,L\ \in \ \mathcal{G}\\
\ket{L} \tensorchi \ket{G^{( 1)}} \ \neq \ 0\\
\ket{L} \tensorchi \ket{H^{( 1)}} \ \neq \ 0
\end{array}} B_{G^{( 1)} H^{( 1)}} \ \left(\ket{L} \tensorchi \ket{G^{( 1)}}\right)\left(\bra{L} \tensorchi \bra{H^{( 1)}}\right)
\end{align*}
\begin{align*}
( A\tensorchi I)( I\tensorchi B) & =\sum _{ \begin{array}{l}
G^{( 0)},\ L=H^{( 0)},\ K=G^{( 1)},\ H^{( 1)}\ \in\mathcal{G}\\
\ket{G^{( 0)}} \tensorchi \ket{K} \ \neq \ 0\\
\ket{H^{( 0)}} \tensorchi \ket{K} \ \neq \ 0\\
\ket{L} \tensorchi \ket{G^{( 1)}} \ \neq \ 0\\
\ket{L} \tensorchi \ket{H^{( 1)}} \ \neq \ 0
\end{array}} \!\!\!\!\!\!\!\!\!\!\!\!\!\!\!\!\!\!\!\!\! A_{G^{( 0)} H^{( 0)}} B_{G^{( 1)} H^{( 1)}}\ket{G^{( 0)}}\bra{H^{( 0)}} \tensorchi \ket{G^{( 1)}}\bra{H^{( 1)}}\\
 & =\sum _{ \begin{array}{l}
G^{( 0)},\ L=H^{( 0)},\ K=G^{( 1)},\ H^{( 1)} \ \in \mathcal{G}\\
\ket{G^{( 0)}} \tensorchi \ket{G^{( 1)}} \ \neq \ 0\\
\ket{H^{( 0)}} \tensorchi \ket{G^{( 1)}} \ \neq \ 0\\
\ket{H^{( 0)}} \tensorchi \ket{G^{( 1)}} \ \neq \ 0\\
\ket{H^{( 0)}} \tensorchi \ket{H^{( 1)}} \ \neq \ 0
\end{array}} \!\!\!\!\!\!\!\!\!\!\!\!\!\!\!\!\!\!\!\!\! A_{G^{( 0)} H^{( 0)}} B_{G^{( 1)} H^{( 1)}}\ket{G^{( 0)}}\bra{H^{( 0)}} \tensorchi \ket{G^{( 1)}}\bra{H^{( 1)}}\\
 & =\sum _{ \begin{array}{l}
G^{( 0)},\ H^{( 0)},\ G^{( 1)},\ H^{( 1)} \ \in\mathcal{G}\\
\ket{H^{( 0)}} \tensorchi \ket{G^{( 1)}} \ \neq \ 0
\end{array}} A_{G^{( 0)} H^{( 0)}} B_{G^{( 1)} H^{( 1)}}\ket{G^{( 0)}}\bra{H^{( 0)}} \tensorchi \ket{G^{( 1)}}\bra{H^{( 1)}}\\
\text{By consist. preserv.} \  & =\sum _{G^{( 0)} ,H^{( 0)} ,G^{( 1)} ,H^{( 1)} \ \in \ \mathcal{G}} A_{G^{( 0)} H^{( 0)}} B_{G^{( 1)} H^{( 1)}} \ \ket{G^{( 0)}}\bra{H^{( 0)}} \tensorchi \ket{G^{( 1)}}\bra{H^{( 1)}}\\
 & =A\tensorchi B
\end{align*}

[Third part]
\begin{align*}
A'\tensorchi I & =\sum _{ \begin{array}{l}
G^{\prime ( 0)} ,H^{\prime ( 0)} ,K\ \in \mathcal{G}\\
\ket{G^{\prime ( 0)}} \tensorchi \ket{K} \ \neq \ 0\\
\ket{H^{\prime ( 0)}} \tensorchi \ket{K} \ \neq \ 0
\end{array}} A'_{G^{\prime ( 0)} H^{\prime ( 0)}} \ \left(\ket{G^{\prime ( 0)}} \tensorchi \ket{K}\right)\left(\bra{H^{\prime ( 0)}} \tensorchi \bra{K}\right)\\
A\tensorchi I & =\sum _{ \begin{array}{l}
G^{( 0)} ,H^{( 0)} ,L\ \in \mathcal{G}\\
\ket{G^{( 0)}} \tensorchi \ket{L} \ \neq \ 0\\
\ket{H^{( 0)}} \tensorchi \ket{L} \ \neq \ 0
\end{array}} A_{G^{( 0)} H^{( 0)}} \ \left(\ket{G^{( 0)}} \tensorchi \ket{L}\right)\left(\bra{H^{( 0)}} \tensorchi \bra{L}\right)
\end{align*}
\begin{align*}
( A'\tensorchi I)( A\tensorchi I) & =\sum _{ \begin{array}{l}
G^{\prime ( 0)} ,H^{\prime ( 0)} =G^{( 0)} ,H^{( 0)} ,K=L\ \in \ \mathcal{G}\\
\ket{G^{'( 0)}} \tensorchi \ket{K} \ \neq \ 0\\
\ket{H^{'( 0)}} \tensorchi \ket{K} \ \neq \ 0\\
\ket{G^{( 0)}} \tensorchi \ket{L} \ \neq \ 0\\
\ket{H^{( 0)}} \tensorchi \ket{L} \ \neq \ 0
\end{array}} \!\!\!\!\!\!\!\!\!\!\!\!\!\!\!\!\!\!\!\!\!\!\!\!\!\!\!\!\!\!\!\!\! A'_{G^{\prime ( 0)} H^{\prime ( 0)}} A_{H^{( 0)} H^{'( 0)}}\left(\ket{G^{\prime ( 0)}} \tensorchi \ket{K}\right)\left(\bra{H^{( 0)}} \tensorchi \bra{K}\right)\\
 & =\sum _{ \begin{array}{l}
G^{( 0)} ,H^{( 0)} ,H^{\prime ( 0)} ,\ K\ \in \mathcal{G}\\
\ket{G^{\prime ( 0)}} \tensorchi \ket{K} \ \neq \ 0\\
\ket{H^{\prime ( 0)}}\tensorchi \ket{K}\ \neq \ 0\\
\ket{H^{( 0)}} \tensorchi \ket{K} \ \neq \ 0
\end{array}} A'_{G^{\prime ( 0)} H^{\prime ( 0)}} A_{H^{\prime ( 0)} H^{( 0)}}\left(\ket{G^{\prime ( 0)}}\bra{H^{( 0)}} \tensorchi \ket{K}\bra{K}\right)\\
 & =\sum _{ \begin{array}{l}
G^{\prime ( 0)} ,H^{\prime ( 0)} ,H^{( 0)} ,\ K\ \in \ \mathcal{G}\\
\ket{H^{\prime ( 0)}}\tensorchi \ket{K}\ \neq \ 0
\end{array}} A'_{G^{\prime ( 0)} H^{\prime ( 0)}} A_{H^{\prime ( 0)} H^{( 0)}}\left(\ket{G^{\prime ( 0)}}\bra{H^{( 0)}} \tensorchi \ket{K}\bra{K}\right)\\
\text{By consist. preserv.} \  & =\sum _{G^{\prime ( 0)} ,H^{\prime ( 0)} ,H^{( 0)} ,\ K\ \in \ \mathcal{G}} A'_{G^{\prime ( 0)} H^{\prime ( 0)}} A_{H^{\prime ( 0)} H^{( 0)}}\left(\ket{G^{\prime ( 0)}}\bra{H^{( 0)}} \tensorchi \ket{K}\bra{K}\right)\\
 & =\ ( A'A\tensorchi I)
\end{align*}
When $A'_{G^{\prime ( 0)} H^{\prime ( 0)}} \neq 0$ \ and $A_{H^{\prime ( 0)} H^{( 0)}} \neq 0$, $\ket{H^{\prime ( 0)}}\tensorchi \ket{K}\ \neq \ 0$ is entailed by $A'$ $\chi $-consistent-preserving and $\ket{G^{\prime ( 0)}} \tensorchi \ket{K} \ \neq \ 0$, or by $A$ $\chi $-consistent-preserving and $\ket{H^{( 0)}} \tensorchi \ket{K} \ \neq \ 0$.

[Fourth part]
\begin{align*}
\text{Say \ } & \bra{H} A'A\ket{G_{\chi }} \neq 0\\
\Leftrightarrow  & \sum _{K}\bra{H} A'\ket{K}\bra{K} A\ket{G_{\chi }} \neq 0\\
\Rightarrow  & \exists K,\bra{H} A'\ket{K}\bra{K} A\ket{G_{\chi }} \neq 0\\
\Rightarrow  & \exists K,\ \bra{H} A'\ket{K} \neq 0\ \text{and} \ \bra{K} A\ket{G_{\chi }} \neq 0\\
\text{By consist. pres.} \quad \Rightarrow  & \exists K,\ \bra{H} A'\ket{K} \neq 0\ \text{and} \ \ket{K} \tensorchi \ket{G_{\overline{\chi }}} \neq 0\\
\text{By consist. pres.} \ quad \Rightarrow  & \ket{H} \tensorchi \ket{G_{\overline{\chi }}} \neq 0
\end{align*}
Similarly for $\bra{H}( A'A)^{\dagger }\ket{G_{\chi }} \neq 0$.

[Fifth part] \ 
\begin{align*}
( A'\tensorchi B')( A\tensorchi B) & =( I\tensorchi B')( A'\tensorchi I)( A\tensorchi I)( I\tensorchi B)\\
 & =( I\tensorchi B')( A'A\tensorchi I)( I\tensorchi B)\\
 & =( A'A\tensorchi B')( I\tensorchi B)\\
 & =( A'A\tensorchi I)( I\tensorchi B')( I\tensorchi B)\\
 & =( A'A\tensorchi I)( I\tensorchi B'B)\\
 & =( A'A\tensorchi B'B)
\end{align*}

\end{proof}

\section{Renaming-invariance}\label{sec:renaminginvariance}

As always, in order to implement a symmetry $R:\rho \mapsto R\rho R^{\dagger }$ in a quantum theory, one can symmetrize states, i.e. demanding $[ R,\rho ] =0$, so that 
\begin{equation*}
\left( A R\rho R^{\dagger }\right)_{|\emptyset } =\ \left( A \rho RR^{\dagger }\right)_{|\emptyset } =( A \rho )_{|\emptyset }.
\end{equation*}
But one can also symmetrize observables, i.e. demanding $RA =A R$, so that 
\begin{equation*}
\left( A R\rho R^{\dagger }\right)_{|\emptyset } =\ \left( RA \rho R^{\dagger }\right)_{|\emptyset } =\ \left( R^{\dagger } RA \rho \right)_{|\emptyset } =( A_{v} \rho )_{|\emptyset }
\end{equation*}
The first option was that taken for name-preservation, in order to obtain Prop. \ref{prop:npcomprehension}.

The second option is that taken for renaming-invariance in Def. \ref{def:renaminginvariance}, because it is more expressive. For instance, think of $A$ as an observable asking the question whether vertex $v$ is connected or isolated. The question would make no sense on a rename-invariant state $\rho $, because the question itself is not rename-invariant. Still we can make the question rename-invariant by parametrizing it by $v$ in a way that $RA_{v} =A_{R( v)} R$, and letting it transform according to $R:A_{v} \mapsto A_{R( v)}$. Then the question make sense on a generic $\rho $, whilst maintaining name-invariance:
\begin{equation*}
\left( A_{Rv} R\rho R^{\dagger }\right)_{|\emptyset } =\ \left( RA_{v} \rho R^{\dagger }\right)_{|\emptyset } =\ \left( R^{\dagger } RA_{v} \rho \right)_{|\emptyset } =( A_{v} \rho )_{|\emptyset }
\end{equation*}
Here are a few helpful facts to help us tame the notion of renaming.

\begin{lema}[Inverse renamings]\label{lem:inverserenamings}
Let $ R:\mathcal{V}\rightarrow \mathcal{V}$ be a homomorphism of the name algebra. \\
The test condition that $ R( x) .t=R( y) .t'$ implies $ x=y$ and $ t=t'$, is equivalent to injectivity.\\
If $ R$ is a renaming, so is $ R^{-1}$.\\
If $ A_{v}$ is renaming-invariant, so is $ A_{v}^{\dagger }$.
\end{lema}

\begin{proof}
~[Injectivity condition]

Notice that $R( x) .t=R( y) .t'$ \ is equivalent to $R( x.t) =R( y.t')$ and that $x.t=y.t'$ is equivalent to $x=y$ and $t=t'$.\\ 
Thus the injectivity of $R$ implies the test condition.

Conversely assume the test condition is satisfied. \\
$u\neq v$ implies there exists $s$ such that $u.s=x.t\neq y.t'=v.s$. \\
Then, $R( u) .s=R( u.s) =R( x.t) =R( x) .t\neq R( y) .t'=R( y.t') =R( v.s) =R( v) .s$.\\
Thus $R( u) \neq R( v)$.\\
Thus $R$ is injective.

[Inverse renaming]

Let $R$ be a renaming. We need to check that $R^{-1}$ is a homomorphism of the name algebra.\\
For any $u',v'$ take $u,v$ such that $u'=R( u)$ and $v'=R( v)$.\\
$R^{-1}( u'.t) =R{^{\ }}^{-1}( R( u) .t) =R{^{\ }}^{-1}( R( u.t)) =u.t=R^{-1}( u') .t$\\
$R^{-1}( u'\lor v') =R{^{\ }}^{-1}( R( u) \lor R( v)) =R{^{\ }}^{-1}( R( u\lor v)) =u\lor v=R^{-1}( u') \lor R^{-1}( v')$.

[Adjoint renaming-invariance]

$RA_{v}^{\dagger } =\left( A_{v} R^{\dagger }\right)^{\dagger } =\left( R^{\ \dagger } A_{R( v)}\right)^{\dagger } =A_{R( v)}^{\dagger } R$.
\end{proof}

The renaming-invariant operator $A_{v} :=\ket{\varnothing }\bra{\{0.v\}}$ may destroy name $v$. The renaming-invariant operator $A_{v}^{\dagger }$ may create it. But this is only because they are parameterized by $v$. Other than that, renaming-invariant operators preserve support up to $\pm $.

\begin{proposition}[Renaming-invariance implies $\pm$-name-preservation]\label{prop:renaminginvarianceimpliesnp}
We define the $ \pm $-vertices of $ G$ to be $ V^{\pm }( G) :=V( G) \cup \{\hyphenbullet v\ |\ \sigma .v\in G\}$.

Let $ A_{v}$ be a renaming-invariant operator over $ \mathcal{H}$, parameterized by $ v\subseteq \mathcal{V}$. Then,
\begin{equation*}
\bra{H} A_{v}\ket{G} \neq 0\ \Rightarrow \ V^{\pm }( G) \cup \{v,\hyphenbullet v\} \corresponds V^{\pm }( H) \cup \{v,\hyphenbullet v\}
\end{equation*}
\end{proposition}
\begin{proof}
$[\mathcal{N}\left[V^{\pm}(H)\right] \subseteq \mathcal{N}\left[V^{\pm}(G)\cup\{v,\hyphenbullet v\}\right]]$ 

By contradiction. Say there exists $\bra{H} A_{v}\ket{G} =\alpha \neq 0$ such that $u\in \mathcal{N}\left[ V^{\pm }( H)\right]\text{ \ and \ } u\notin \mathcal{N}\left[ V^{\pm }( G) \cup \{v,\hyphenbullet v\}\right]$.

Pick $R$ such that $RG=G$, $R( v) =v$, $R( u) \notin \mathcal{N}\left[ V^{\pm }( H)\right]$, i.e. map $u$ into a fresh name $u'$ whilst preserving $v$ and $G$. We have:
\begin{align*}
\alpha = & \bra{H} A_{v}\ket{G}\\
= & \bra{RH} RA_{v}\ket{G}\\
\text{By renaming-inv.} \ = & \bra{RH} A_{R( v)} R\ket{G}\\
\text{By choice of } R\ = & \bra{RH} A_{v}\ket{G}
\end{align*}
There are infinitely many such $R$, and since $u\in \mathcal{N}\left[ V^{\pm }( H)\right]$, there are infinitely many such $RH$. It follows that $A_{v}\ket{G}$ is unbounded, hence the contradiction. The result follows, from which we also have that $\mathcal{N}\left[ V^{\pm }( H) \cup \{v,\hyphenbullet v\}\right] \subseteq \mathcal{N}\left[ V^{\pm }( G) \cup \{v,\hyphenbullet v\}\right]$.

[$\mathcal{N}\left[ V^\pm (G)\right] \subseteq \mathcal{N}\left[ V^\pm (H) \cup \{v,\hyphenbullet v\}\right]$]

$\bra{H} A_{v}\ket{G} \neq 0\ \Leftrightarrow \ \bra{G} A_{v}^{\dagger }\ket{H}^{*} \neq 0\ \Leftrightarrow \ \bra{G} A_{v}^{\dagger }\ket{H} \neq 0$.\\
Moreover, by Lem. \ref{lem:inverserenamings}, $A_{v}^{\dagger }$ is also renaming-invariant.\\
So, the same reasoning applies. 

We therefore have that $\mathcal{N}\left[ V^{\pm }( G) \cup \{v,\hyphenbullet v\}\right] \subseteq \mathcal{N}\left[ V^{\pm }( H) \cup \{v,\hyphenbullet v\}\right]$

[$V^{\pm }( G) \cup \{v,\hyphenbullet v\} \corresponds V^{\pm }( H) \cup \{v,\hyphenbullet v\}$] is by definition $\mathcal{N}\left[ V^{\pm }( G) \cup \{v,\hyphenbullet v\}\right] =\mathcal{N}\left[ V^{\pm }( H) \cup \{v,\hyphenbullet v\}\right]$.
\end{proof}

In order to obtain full name-preservation as used the core of the paper, as a consequence of renaming-invariance, we could have restricted our attention to graphs that have no half-edges, i.e. such that if $c.t\in V( G)$, then $\hyphenbullet c.t\in V( G)$. Indeed for such closed graphs, full name-preservation amounts to $\pm -$name-preservation, and is therefore entailed by renaming-invariance. However the operations that we study in the paper do not preserve closed graphs. For instance, $G$ may be closed, but not $G_{\chi }$. This is why we have treated name-preservation as a independent assumption.

In the pursuit of obtaining full name-preservation, as a consequence of renaming-invariance, we could demand full renaming-invariance, i.e. letting renamings $R$ be arbitrary isomorphisms on $\mathcal{N}[\mathbb{K} \cup \hyphenbullet \mathbb{K}]$ rather than extensions of isomorphisms on $\mathcal{N}[\mathbb{K}]$. However, such an $R$ may map $x$ to $y$ and $\hyphenbullet x$ to $z$, thereby destroying the geometrical information held by names that $x$ and $\hyphenbullet x$ are connected by an edge. We could then compensate for that loss by providing each graph $G$ with a adjacency function $\alpha _{G}$, transforming according to $\alpha _{RG} :=R\circ \alpha _{G} \circ R^{-1}$. Which $\alpha _{G}$ should be allowed? We posit the following conditions:
\begin{itemize}
\item $\alpha _{G}$ is a partial renaming over $\mathcal{V}$.
\item $\forall u,\ \alpha _{G}( u) =v\ \Rightarrow \ \alpha _{G}( v) \ \text{undefined}$.
\item $\forall u\in \text{dom}( \alpha _{G}), \exists t, \ u.t\in \mathcal{N}[ V( G)] \ \text{or} \ \alpha _{G}( u.t) \in \mathcal{N}[ V( G)]$.
\end{itemize}
The first two conditions imply $\forall u,\ \neg \left(\mathcal{N}[ \alpha ( u)] \subset \mathcal{N}[ u] \ \text{or} \ \mathcal{N}[ \alpha ( u)] \supset \mathcal{N}[ u]\right)$, which is reinsuring as "edges from a part to its subpart" seem undesirable. Going that this route, care must be taken when defining $\alpha _{G_{\chi }}$ to also encompass edges that are incoming from $G_{\overline{\chi }}$, or else Prop. \ref{prop:unitaryextension} will fail. Indeed, say that $\alpha _{G}( u) =v$ with $u\in V( G_{\overline{\chi }})$ and $v\in V( G_{\chi })$. An operator $U$ acting over $G_{\chi }$ may otherwise fail to see that $v$ is occupied, thereby producing $G'_{\chi }$ such that $\ket{G'_{\chi }} \tensorchi \ket{G_{\overline{\chi }}}$. Overall, this is a legitimate route to take, but we chose not to clutter this paper.

Finally, notice that a number of results in the core of the paper held without name-preservation. At the cost of name-preservation, we can even reach an interesting version of Th. \ref{th:blockdecomposition}, which does not require an extra bit of information per system. 

\begin{proposition}[Unitary restriction]\label{prop:namewiseunitaryextension}
A restriction is namewise if and only if there exists $ S$ such that $ G_{\chi } \ =\left\{\sigma .v\ \in G\ |\ v\notin \mathcal{N}\left[V^{\pm }(S)\right]\right\}$.

If $ \chi $ is namewise, and $ U$ is renaming-invariant, then $ U$ preserves $ \mathcal{H}_{\chi }$.\\
If moreover $ U$ is unitary, then it is unitary over $ \mathcal{H}_{\chi }$.
\end{proposition}
\begin{proof}
Since $\chi $ is namewise, there exists $S$ such that $G=G_{\chi }$ if and only if $\neg \left(\mathcal{N}[ V( G)] \cap \mathcal{N}\left[ V^{\pm }( S)\right]\right)$.

Since $U$ is renaming-invariant, then by Prop. \ref{prop:renaminginvarianceimpliesnp}, we have $\mathcal{N}\left[ V^{\pm }\left( U\ket{G}\right)\right] \subseteq \mathcal{N}\left[ V^{\pm }\left(\ket{G}\right)\right]$, i.e. for all $u,$ $u\in \mathcal{N}\left[ V^{\pm }\left( U\ket{G}\right)\right]$ implies $u\in \mathcal{N}\left[ V^{\pm }\left(\ket{G}\right)\right]$, and so $\mathcal{N}\left[ V^{\pm }\left( U\ket{G}\right)\right] \cap \mathcal{N}\left[ V^{\pm }( S)\right]$ implies $\mathcal{N}\left[ V^{\pm }\left(\ket{G}\right)\right] \cap \mathcal{N}\left[ V^{\pm }( S)\right]$.\\ As a consequence $\neg (\mathcal{N}\left[ V^{\pm }\left(\ket{G}\right) \cap \mathcal{N}\left[ V^{\pm }( S)\right]\right)$ implies $\neg \left(\mathcal{N}\left[ V^{\pm }\left( U\ket{G}\right)\right] \cap \mathcal{N}\left[ V^{\pm }( S)\right]\right)$.

Say $G=G_{\chi }$. We therefore have $\neg (\mathcal{N}\left[ V^{\pm }\left( U\ket{G}\right] \cap \mathcal{N}\left[ V^{\pm }( S)\right]\right)$. As a consequence for any \ $H$ such that $\bra{H} U\ket{G} \neq 0$, we have $\neg \left(\mathcal{N}\left[ V^{\pm }( H)\right] \ \cap \mathcal{N}\left[ V^{\pm }( S)\right]\right)$. Thus $H=H_{\chi }$.

[Unitary case]

If $U$ is renaming-invariant and unitary, then by Lemma Inverse renaming and renaming-invariance, so is $U^{\dagger }$. It follows that $U^{\dagger }$ preserves $\mathcal{H}_{\chi }$. Therefore, $U$ is unitary when restricted to $\mathcal{H}_{\chi }$.
\end{proof}

\begin{thm}[Block decomposition without ancilla]\label{def:blockcecompositionnoancilla}

Let $ \zeta _{v}$ be the pointwise restriction such that $ \zeta _{v}(\{\sigma .u\}) :=\begin{cases}
\{\sigma .u\} & \text{if} \ \mathcal{N}[ u] \cap \mathcal{N}[ v]\\
\varnothing  & \text{otherwise}
\end{cases}$.\\
Consider $ U$ a renaming-invariant unitary operator over $ \mathcal{H}$, which for all $ v\in \mathcal{V}$ is $ \chi _{v} \zeta '_{v}$-causal with $ \zeta _{v} \sqsubseteq \zeta '_{v}$ and $\chi_v$ an extensible restriction.

Let $ \mu $ be the namewise restriction such that $ \mu (\{\sigma .u\}) :=\begin{cases}
\{\sigma .u\} & \text{if} \ u\notin \mathcal{N}\left[\mathbb{Z}^{*} .1\right]\\
\varnothing  & \text{otherwise}
\end{cases}$
where $\mathbb{Z}.1$ denotes odd numbers in their binary notation.\\
Similarly let $ \mathcal{V}.0$ denote those names built out of even numbers.

There exists $ \tau _{x}$ a non-name-preserving $ \zeta _{x}$-local unitary and $ K_{x}$ a non-name-preserving $ \xi _{x}$-local unitary such that
\begin{equation*}
\begin{aligned}
\forall \ket{\psi } \in \mathcal{H}_{\mu } \cong \mathcal{H} ,\ \left(\prod _{x\ \in \ \mathbb{N}} \tau _{x}\right)\left(\prod _{x\ \in \ \mathbb{N}} K_{x}\right)\ket{\psi } & =U\ket{\psi }
\end{aligned}
\end{equation*}
where $ \xi _{x} :=\mu \chi _{x} \cup \overline{\mu } \zeta _{x}$ is an extensible restriction. In addition, $ \ [ K_{x} ,K_{y}] =[ \tau _{x} ,\tau _{y}] =0$. 
\end{thm}
\begin{proof}
Clearly $[ \mu ,\zeta ] =[\overline{\mu } ,\zeta ] =[ \mu ,\overline{\zeta }] =[\overline{\mu } ,\overline{\zeta }] =0$ as both are pointwise.

By Prop. \ref{prop:unitaryextension}, since $\mu $ is namewise, and $U$ is a renaming-invariant unitary, it is unitary over $\mathcal{H}_{\mu }$, and $U':=U \tensormu  I$ is unitary over $\mathcal{H}$ with $U^{\prime \dagger } =U^{\dagger }  \tensormu  I$.\\
By Prop. \ref{prop:causalextension} and since $U$ is $\chi _{v} \zeta '_{v}$-causal, it is $\chi _{v} \zeta _{v}$-causal, and $U'$ is $\xi _{v} \zeta _{v}$-causal, with $\xi _{v} :=\mu \chi _{v} \cup \overline{\mu } \zeta _{v}$ and extensible restriction.

Let the toggle $\tau _{x}$ be the renaming such that $\tau _{x}( y.b) =\begin{cases}
\tau _{x}( x.\neg b) & \text{if} \ x=y\\
y & \text{otherwise}
\end{cases}$.\\
I.e. $\tau _{x}$ toggles the last by of $x.0$ and $x.1$.\\
Notice that it is unitary and $\zeta _{x}$-local, and that $[ \tau _{x} ,\tau _{y}] =0$.\\
Moreover, 
\begin{equation*}
\left(\prod _{x\ \in \ \mathbb{N}} \tau _{x}\right) =\tau 
\end{equation*}
where $\tau $ is the renaming such that $\tau ( y.b) =\tau ( y.\neg b)$.

Let $K_{x} :=U^{\prime \dagger } \tau _{x} U'$. \\
Since adjunction by a unitary is a morphism, $[ K_{x} ,K_{y}] =0$.\\
By Prop. \ref{prop:dualcausality}, is $\xi _{x}$-local.

Finally,
\begin{align*}
\left(\prod _{x\ \in \ \mathbb{N}} \tau _{x}\right)\left(\prod _{x\ \in \ \mathbb{N}} K_{x}\right)\ket{G_{\mu }} & =\tau \dotsc \left( U^{\prime \dagger } \tau _{1} U'\right)\left( U^{\prime \dagger } \tau _{0} U'\right)\ket{G_{\mu }}\\
\text{By unitarity of } U'. & =\tau \ U^{\prime \dagger } \ \left(\prod _{x\ \in \ \mathbb{N}} \tau _{x}\right) \ U'\ \ket{G_{\mu }}\\
 & =\tau \ U^{\prime \dagger } \ \tau \ U'\ \ket{G}\\
\text{By Prop. \ref{prop:unitaryextension}} & =\tau \ \left( U^{\dagger }  \tensormu  I\right) \ \tau \ ( U \tensormu  I) \ \left(\ket{G_{\mu }}  \tensormu  \ket{\varnothing }\right) \ \\
 & =\tau \ \left( U^{\dagger }  \tensormu  I\right) \ \tau \ \left( U\ket{G_{\mu }}  \tensormu  \ket{\varnothing }\right)\\
\text{Since} \ U\ \text{preserves the range of } \mu . & =\tau \ \left( U^{\dagger }  \tensormu  I\right)\left(\ket{\varnothing }  \tensormu  \tau U\ket{G_{\mu }}\right)\\
\text{By Prop. \ref{prop:renaminginvarianceimpliesnp}} & =\tau \ \left(\ket{\varnothing } \tensormu  \tau U\ket{G_{\mu }}\right)\\
 & =\tau ^{2} \ U\ket{G_{\mu }}\\
\text{Since} \ \tau \ \text{involutive.} & =U\ket{G_{\mu }}
\end{align*}
\end{proof}

\section{Standard tensor product and direct sum as special cases}\label{sec:reconstruction}
In this section we are going to make precise the connection with more standard notions of decompositions of Hilbert spaces. Before we begin let us outline a general way to construct sub-spaces on which our tensor products act. 

Given any restriction $\chi$ we can define the Hilbert space $\mathcal{H}_{\chi}$ as the Hilbert space whose orthonormal basis is $\mathcal{G}_{\chi}=\{\ket{G_\chi}\,|\,G\in \mathcal{G}\}$. Note that living inside this space is the empty graph since $ \varnothing_\chi=\varnothing$;  this is a remnant of the `fock-space of qubits' point of view taken in this paper. The quantum circuit formalism, on the other hand, works with a fixed number of qubit and hence has no vacuum state. Thus, to make contact with the standard formalism, we will need to remove de vacuum state and consider
$ \pi_{\overline{\varnothing}}{\mathcal{H}_{\chi}} $, i.e. the Hilbert space given by the range of $\pi_{\overline{\varnothing}}$, the projector over the orthogonal subspace of $\ket{\varnothing}$. 

Consider $\chi_A$ such that $\chi\chi_A=\chi_A$. Just like a linear map $A$ over $\mathcal{H}_A$ can be extended as $A \otimes I$ to act over $\mathcal{H}_A \otimes \mathcal{H}_B$, we need to be able to take an operator $A$ over $\pi_{\overline{\varnothing}}{\mathcal{H}_{\chi_A}}$ and extend it as an operator $\underline{A \tensorchi_{\!\!A} I}$ over $ \pi_{\overline{\varnothing}}{\mathcal{H}_\chi}{}$. In this appendix this will be done by taking
\[  \underline{A \tensorchi_{\!\!A} I} :=  \pi_{\overline{\varnothing}}{((A \oplus I) \tensorchi_{\!\!A} I)}{}.\]

Similarly, the generalised trace $(\bullet)_{|\chi_A}$ over $\mathcal{B}_1(\mathcal{H}_\chi)$ may yield vacuum states, which we do not want in the standard formalism. For this reason will need a `$\pi_{\overline{\varnothing}}$-projected generalised trace' $\operatorname{Tr}_{\overline{\chi}_A}(\bullet)$ 
as defined in the following way:
\[   \operatorname{Tr}_{\overline{\chi}_A}(\rho) := \pi_{\overline{\varnothing}}{\rho_{|\chi_A}} \pi_{\overline{\varnothing}}.   \]

\subsection{Recovering the standard tensor product}

The general results below will be formulated in terms of a three restrictions $\chi, \chi_A, \chi_B$. In order to gain concrete intuitions about them, it will be helpful to keep in mind the example where
\begin{description}
\item[$\mathcal{H}_\chi$] is the Hilbert space spanned by the empty graph and the graphs consisting exactly two nodes with names $u$ and $v$. In other words, every graph in $\mathcal{H}$ is either empty or has the form $\{a.u\} \cup \{b.v\}$ for $a,b \in \Sigma$ and $u,v \in V$. 
\item[$\mathcal{H}_{\chi_A}$] is the Hilbert space spanned by the empty graph and the graphs consisting of one node named $u$.
\item[$\mathcal{H}_{\chi_B}$] is the Hilbert space spanned by the empty graph and the graphs consisting of one node named $v$.
\end{description}
So, consider three restrictions $\chi, \chi_A, \chi_B$ such that
\begin{itemize}
\item $\chi\chi_A=\chi_A$, $\chi\chi_B=\chi_B$.
\item $\chi_A=\chi\overline{\chi}_B$ i.e. $\chi_A$ and $\chi_B$ are the complement of the other in $\chi$.
\item for all $\ket{\psi}\in{\cal H}_{\chi_A}$, $\ket{\psi'}\in{\cal H}_{\chi_B}$, $||\ket{\psi}\tensorchi_{\!\!A}\ket{\psi'}||^2=||\ket{\psi}||^2.||\ket{\psi'}||^2$ 
\item for all $G,G'  \in \cal{G}_{\chi_A} \times \cal{G}_{\chi_B}$ with $G,G' \neq \varnothing$ then $G \cup G' \in \mathcal{G}_{\chi}$.
\end{itemize}
and consider ${\cal H}_A:= \pi_{\overline{\varnothing}}{\mathcal{H}_{\chi_A}}$, ${\cal H}_B:= \pi_{\overline{\varnothing}}{\mathcal{H}_{\chi_B}}$.

We will now formalise the sense in which the standard tensor product 
$$\otimes: {\cal H}_A\times{\cal H}_B\to{\cal H}_{A}\otimes{\cal H}_{B}$$
coincides with $\tensorchi_{A}$, by means of the unitary isomorphism
\begin{align*}
E: \mathcal{H}_A \otimes \mathcal{H}_B  &\stackrel{\cong}{\rightarrow} \pi_{\overline{\varnothing}}{\mathcal{H}_\chi}.\\
\ket{G}\otimes\ket{G'}&\mapsto \ket{G}\tensorchi_{\!\!A}\ket{G'}=\ket{G\cup G'}
\end{align*}
The existence of this isomorphism holds thanks to the norm and union conditions, and furthermore witnesses an equivalence $ \underline{A \tensorchi_{\!\!A} I} \cong A \otimes I$:
\begin{align*}
\bra{G} \bra{G'} E^{\dagger} (\underline{A \tensorchi_{\!\!A} I}) E \ket{H} \ket{H'} & =   \bra{G \cup G'} \underline{A \tensorchi_{\!\!A} I} \ket{H \cup H'}    \\
 & =  \bra{G \cup G'} (A \oplus I) \tensorchi_{\!\!A} I \ket{H \cup H'}   \\
 & =  \bra{G} A \oplus I \ket{H }  \braket{G'| H'} \\
 & = \bra{G} A \ket{H }  \braket{G' | H'}  \\
 & = \bra{G} A \ket{H }  \bra{G'} I \ket{H'} \\
 & = \bra{G} \bra{G'} A \otimes I \ket{H} \ket{H'} \\
\textrm{so that }  (A \otimes I) &= E^{\dagger} (\underline{A \tensorchi_{\!\!A} I})  E.
\end{align*}
In this precise sense then, the generalised tensors recover the standard tensors. 
Similarly, the $\pi_{\overline{\varnothing}}$-projected generalised trace recovers the standard trace out by means of
 \[  \operatorname{Tr}_{B}(\bullet) = \operatorname{Tr}_{\chi_B}(E^{\dagger} (\bullet)E),  \] where $\operatorname{Tr}_B(\rho) = \sum_{R} (I \otimes \bra{R}) \rho (I \otimes \ket{R})$ is the standard partial trace. Indeed,
\begin{align*}
    \operatorname{Tr}_{\chi_B}(E (\ket{G} \otimes \ket{G'}) (\bra{H} \otimes \bra{H'}) E^{\dagger}) & = \operatorname{Tr}_{\chi_B}( (\ket{G \cup G'}) (\bra{H \cup H'}) )  \\
    & = \ket{G} \bra{H} \braket{G'|H'} \\
    & = \sum_{R\in\mathcal{G}_{\chi_B}} (I \otimes \bra{R}) (\ket{G} \otimes \ket{G'}) (\bra{H} \otimes \bra{H'}) (I \otimes \ket{R}) \\
    & = \operatorname{Tr}_{B}((\ket{G} \otimes \ket{G'}) (\bra{H} \otimes \bra{H'})).
\end{align*}
By linearity the result holds for any trace class operator.

\subsection{Recovering the direct sum}

The general results below will be formulated in terms of a three restrictions $\chi, \chi_A, \chi_B$. In order to gain concrete intuitions about them, it will be helpful to keep in mind the example where
\begin{description}
\item[$\mathcal{H}_\chi$] is the Hilbert space spanned by the empty graph and the graphs consisting of one node with name $u$. In other words, every graph in $\mathcal{H}$ is either empty or has the form $\{a.u\}$ for $a \in \Sigma$ and $u \in V$. 
\item[$\mathcal{H}_{\chi_A}$] is the Hilbert space spanned by the empty graph and the graphs consisting of one node with name $u$ with internal state in $S$. In other words, every graph in $\mathcal{H}_{\chi_A}$ is either empty or has the form $\{a.u\}$ for $a \in S\subset\Sigma$ and $u \in V$.
\item[$\mathcal{H}_{\chi_B}$] is the Hilbert space spanned by the empty graph and the graphs consisting of one node with name $u$ with internal state not in $S$. In other words, every graph in $\mathcal{H}_{\chi_B}$ is either empty or has the form $\{a.u\}$ for $a \in \Sigma\setminus S$ and $u \in V$.
\end{description}
So, consider three restrictions $\chi, \chi_A, \chi_B$ such that
\begin{itemize}
\item $\chi\chi_A=\chi_A$, $\chi\chi_B=\chi_B$. 
\item $\chi_A=\chi\overline{\chi}_B$ i.e. $\chi_A$ and $\chi_B$ are the complement of the other in $\chi$.
\item for all $\ket{\psi}\in{\cal H}_{\chi_A}$, $\ket{\psi'}\in{\cal H}_{\chi_B}$, 
$$||\ket{\psi}\tensorchi_{\!\!A}\ket{\psi'}||^2=
|\braket{\varnothing|\psi'}|^2.||\pi_{\overline{\varnothing}}\ket{\psi}||^2
+|\braket{\varnothing|\psi}|^2.||\pi_{\overline{\varnothing}}\ket{\psi'}||^2
+|\braket{\varnothing|\psi}|^2|\braket{\varnothing|\psi'}|^2$$
\item for all $G \in \cal{G}_{\chi_A} \cup \cal{G}_{\chi_B}$ then $G \in \cal{G}_{\chi}$
\end{itemize}
and consider ${\cal H}_A:= \pi_{\overline{\varnothing}}{\mathcal{H}_{\chi_A}}$, ${\cal H}_B:= \pi_{\overline{\varnothing}}{\mathcal{H}_{\chi_B}}$.

We will now formalise the sense in which the standard direct sum 
$$\otimes: {\cal H}_A\times{\cal H}_B\to{\cal H}_{A}\oplus{\cal H}_{B}$$
coincides with $\tensorchi_{A}$, by means of the unitary isomorphism
\begin{align*}
E:  \mathcal{H}_A \oplus \mathcal{H}_B   &\stackrel{\cong}{\rightarrow}  \pi_{\overline{\varnothing}}{\mathcal{H}_\chi}\\
\ket{G}\oplus\ket{G'}&\mapsto \ket{G}+\ket{G'}
\end{align*}  
Indeed it witnesses the equivalence $ \underline{A \tensorchi_{\!\!A} I} \cong A \oplus I$:

\begin{align*}
(\bra{G} \oplus  \bra{G'} )E^{\dagger} (\underline{A \tensorchi_{\!\!A} I}) E (\ket{H}  \oplus \ket{H'}) & =  ( \bra{G} + \bra{G'} ) \underline{A \tensorchi_{\!\!A} I}( \ket{H} + \ket{H'} )   \\
 & =  (\bra{G} + \bra{ G'} )((A \oplus I) \tensorchi_{\!\!A} I) (\ket{H} + \ket{H'} )  \\
  & =  \bra{G} ((A \oplus I) \tensorchi_{\!\!A} I )\ket{H}  +  \bra{G} ((A \oplus I) \tensorchi_{\!\!A} I )\ket{H'} \\
    & +  \bra{G'} ((A \oplus I) \tensorchi_{\!\!A} I) \ket{H}  +  \bra{G'} ((A \oplus I) \tensorchi_{\!\!A} I) \ket{H'} \\
\textrm{Using the norm requirement}\quad      & =  \bra{G} (A \oplus I) \ket{H}  \braket{\varnothing | \varnothing} +  \bra{G} (A \oplus I)\ket{\varnothing} \braket{\varnothing | H'} \\
    & +  \bra{\varnothing} (A \oplus I)  \ket{H} \braket{G'| \varnothing}  +  \bra{\varnothing} (A \oplus I) \ket{\varnothing} \braket{G' | H'} \\
 & =  \bra{G} A \oplus I \ket{H } +  \braket{G' | H'} \\
 & = \bra{G} A \ket{H } +   \braket{G' | H'} \\
 & = \bra{G} A \ket{H }  +  \bra{G'} I \ket{H'} \\
 & = (\bra{G}  \oplus\bra{G'} ) A \oplus I ( \ket{H}  \oplus \ket{H'} )\\
\textrm{and so this time}\quad
(A \oplus I) &= E^{\dagger}(\underline{A \tensorchi_{\!\!A} I})E .
\end{align*}
In this precise sense then, the generalised tensors recovers the standard direct sum. 

Interestingly, the $\pi_{\overline{\varnothing}}$-projected generalised trace turns out to be just a projection:  \[  \operatorname{Tr}_{\chi_B}(E^{\dagger} (\bullet)E)=\pi_A(\bullet) \pi_A,  \] where $\pi_A$ is the projector upon  $\mathcal{H}_A$. Indeed, 
\begin{align*}
    \operatorname{Tr}_{\chi_B}(E^{\dagger} \rho E) & = \sum_{G,H} \rho_{GH} \operatorname{Tr}_{\chi_B}(E^{\dagger} \ket{G}\bra{H} E)   \\
    & =  \sum_{G,H} \rho_{GH} \operatorname{Tr}_{\chi_B}( \ket{G}\bra{H}) \\
    & = \sum_{G,H} \rho_{GH} \pi_{\overline{\varnothing}}(\ket{G}\bra{H})_{|\chi_A}\pi_{\overline{\varnothing}}\\  
    & = \sum_{G,H \in {\cal G}^2_{\chi_A}} \rho_{GH} \ket{G} \bra{H} \\
    & = \pi_{A} \rho \pi_{A}
\end{align*}
Note that whilst $(\bullet )|_{\chi_A}$ is always trace preserving, $\operatorname{Tr}_{\chi_B}(\bullet)$ is not. This is because $(\bullet )|_{\chi_A}$ preserves trace by those sending states that would have been projected out by $\pi_A$, to the empty graph. 

In the standard, non-Fock-space setting then, we therefore are faced with a choice between charybdis and scylla. 1/ We can either leave things as such, and sacrifice trace-preservation; or 2/ we can renormalise $\operatorname{Tr}_{\chi_B}$ to get the post-measurement state-collapse rule $\frac{\pi_A \rho \pi_A}{Tr(\pi_A \rho \pi_A)}$ at the cost of sacrificing linearity and being only partially defined (since no amount of renormalisation will help if the result of $\operatorname{Tr}_{\chi_B}(\rho)$ is $0$).

The Fock-space of qubit setting is both mathematically more convenient, and physically well-motivated, we argue. Mathematically, it allows use to have our cake and eat it, having both linearity and trace-preservation for a trace based on a direct-sum structure. Physically, it has the natural interpretation of yielding representation of the viewpoint of a limited observer, who sees nothing beyond the state set ${\cal H}_A$. Consider for instance a situation in which a particle may occupy one of many energy levels distributed across space (say for instance the energy levels of the hydrogen atom), and imagine that an observer, Alice, only has access to a limited range of spatial positions and thus a limited range of all possible eigenstates. According to her viewpoint, she sees $\pi_A \rho \pi_A$ with some probability, and nothing $\ket{\varnothing} \bra{\varnothing}$ with some other, i.e. precisely 
$$\rho_{|\chi_A}=\pi_A \rho \pi_A + \operatorname{Tr}(\pi_{\overline{A}} \rho \pi_{\overline{A}}) \ket{\varnothing} \bra{\varnothing}.$$
Thus, the generalised trace is also a linear, trace-preserving way to take the viewpoint of a limited observer.

\subsection{Between the direct sum and the tensor product}

Let us find an example, simpler in spirit than for instance disks of fixed radius around nodes, that goes beyond both the standard tensor product and direct sum. Specifically, the example mixes them up.\\ 
Consider $\mathcal{H}_\chi$ the Hilbert space spanned by the empty graph and the graphs consisting of one or two nodes with names $u$ and $v$.
Consider the restriction $\zeta_{u:S}$ which picks out the node $u$ but only when $u$ has a state in $S \subset \Sigma$. This time we can define $\mathcal{H}_{u:S} := \pi_{\overline{\varnothing}}{\mathcal{H}_{\zeta_{u:S}}}$ and $\mathcal{H}_{u:\Sigma - S} := \pi_{\overline{\varnothing}}{\mathcal{H}_{\zeta_{u:\Sigma - S}}}$ along with $\mathcal{H}_{v} = \pi_{\overline{\varnothing}}{\mathcal{H}_{\zeta_v}}$. We can now see that there is a unitary isomorphism:
\begin{align*}
E : {\rightarrow} (\mathcal{H}_{u:S} \oplus  \mathcal{H}_{u:\Sigma -S}  ) \otimes \mathcal{H}_v \stackrel{\cong} \pi_{\overline{\varnothing}}{\mathcal{H}} \\
\alpha \bra{a.u} \oplus  \beta \bra{b.u} ) \bra{c.v} &\mapsto \alpha \bra{a.u \cup c.v} +  \beta \bra{b.u \cup c.v}
\end{align*}
which witnesses the equivalence $ \underline{A \tensorzeta_{\!u:S} I} \cong (A \oplus I_{\mathcal{H}_{u:\Sigma-S}}) \otimes I_{\mathcal{H}_v}$ in the usual sense, indeed (where we include scalars this time for clarity):
\begin{align*}
& (\alpha \bra{a.u} \oplus  \beta \bra{b.u} ) \bra{c.v}(E^{\dagger} \underline{A \tensorzeta_{\!u:S} I} E) (\alpha' \ket{a^{'}.u}  \oplus \beta' \ket{b^{'}.u}) \ket{c^{'}.v} \\
& =   (\alpha \bra{a.u \cup c.v} +  \beta \bra{b.u \cup c.v} )\underline{A \tensorzeta_{\!u:S} I} (\alpha' \ket{a^{'}.u \cup c^{'}.v}  + \beta' \ket{b^{'}.u \cup c^{'}.v})  \\
 & = (\alpha \bra{a.u \cup c.v} +  \beta \bra{b.u \cup c.v} ) (A \oplus I) \tensorzeta_{\!u:S} I (\alpha' \ket{a^{'}.u \cup c^{'}.v}  + \beta' \ket{b^{'}.u \cup c^{'}.v}) \\
  & =  \alpha \alpha' \bra{a.u \cup c.v} (A \oplus I) \tensorzeta_{\!u:S} I \ket{a^{'}.u \cup c^{'}.v}  +  \alpha \beta' \bra{a.u \cup c.v} (A \oplus I) \tensorzeta_{\!u:S} I \ket{b^{'}.u \cup c^{'}.v} \\
    & +  \beta \alpha'  \bra{b.u \cup c.v} (A \oplus I) \tensorzeta_{\!u:S} I \ket{a^{'}.u \cup c^{'}.v}  + \beta \beta'   \bra{b.u \cup c.v} (A \oplus I) \tensorzeta_{\!u:S} I \ket{b^{'}.u \cup c^{'}.v} \\
      & = \alpha \alpha'  \bra{a.u} (A \oplus I) \ket{a^{'}.u}  \braket{c.v | c^{'}.v} + \alpha \beta'  \bra{a.u } (A \oplus I)\ket{\varnothing} \braket{c.v | b^{'}.u \cup c^{'}.v}  \\
    & +  \beta \alpha'  \bra{\varnothing} (A \oplus I)  \ket{a^{'}.u} \braket{b.u \cup c.v | c^{'}.v}  + \beta \beta'   \bra{\varnothing} (A \oplus I) \ket{\varnothing} \braket{b.u \cup c.v | b^{'}.u \cup c^{'}.v} \\
 & =  \alpha \alpha'  \bra{a.u} A \oplus I \ket{a^{'}.u } \delta_{cc'} + \beta \beta' \delta_{bb'} \delta_{cc'} \\
 & = (\alpha \alpha'  \bra{a.u} A \ket{a^{'}.u } +  \beta \beta' \delta_{bb'}) \delta_{cc'} \\
 & =  (\alpha \bra{a.u}  \oplus \beta \bra{b.u} ) \bra{c.v} (A \oplus I) \otimes I(\alpha' \ket{a^{'}.u}  \oplus \beta' \ket{b^{'}.v} ) \ket{c^{'}.v} 
 \end{align*}
and so in this case we can say that:
\begin{align*}
 ( (A \oplus I_{\mathcal{H}_{u:\Sigma-S}}) \otimes I_{\mathcal{H}_v} ) =  E^{\dagger}(\underline{A \tensorzeta_{\!u:S} I})E   .
\end{align*}
Thus, generalised tensors recover mixtures of tensors and direct sums.
Similarly $\operatorname{Tr}_{\overline{\zeta_u}}$ is now a combination of a standard partial trace and a projection:  \[  \operatorname{Tr}_{\overline{\zeta_u}}(E^{\dagger} (\bullet)E)=\pi_S \operatorname{Tr}_{\overline{u}}(\bullet) \pi_S  \] where $\operatorname{Tr}_{\overline{u}}(\rho) = \sum_{b} (I_u \otimes \bra{b.v}) \rho (I_u \otimes \ket{b.v})$ is the standard partial trace. Indeed, 
\begin{align*}
\operatorname{Tr}_{\overline{\zeta_{u:S}}} (E^{\dagger} (\ket{a.u} \otimes \ket{b.v})(\bra{a^{'}.u} \otimes \bra{b^{'}.v}) E) 
&= \operatorname{Tr}_{\zeta_{u:S}} (\ket{a.u \cup b.v})(\bra{a^{'}.u \cup b^{'}.v})) \\
&= 
\begin{cases} 
\ket{a_{u}}\bra{a^{'}.v}, & \text{if } a,a' \in S \text{ and } b = b' \\
0, & \text{otherwise}
\end{cases} \\
&= \pi_S \operatorname{Tr}_{\overline{u}}(\ket{a.u} \otimes \ket{b.v})(\bra{a^{'}.u} \otimes \bra{b^{'}.v})) \pi_S
\end{align*}
Note that $\operatorname{Tr}_{\overline{\zeta_{u:S}}} (\ket{a.u \cup b.v})(\bra{a^{'}.u \cup b^{'}.v})) = 0$ when $a,a' \notin S$ since then \[(\ket{a.u \cup b.v})(\bra{a^{'}.u \cup b^{'}.v})_{| \zeta_{u:S}} = \ket{\varnothing} \bra{\varnothing} \delta_{aa'} \delta_{bb'}\] so that $\operatorname{Tr}_{\overline{\zeta_{u:S}}}$ which is $(\bullet)_{| \zeta_{u:S}}$ up to quotient by $\ket{\varnothing}$ returns $0$. 

\color{black}

\end{document}

%% file: tikzstyles.tex
\tikzstyle{doubled}=[line width=1.5pt] %

\tikzstyle{dot}=[inner sep=0mm,minimum width=2mm,minimum height=2mm,draw,shape=circle]  
\tikzstyle{ddot}=[inner sep=0mm, doubled, minimum width=2.5mm,minimum height=2.5mm,draw,shape=circle]

\tikzstyle{pdot}=[inner sep=0mm, doubled, minimum width=2.5mm,minimum height=2.5mm,shape=circle]
\tikzstyle{phase dimensions}=[minimum size=6mm,font=\footnotesize,inner sep=0.2mm,outer sep=-2mm]

\tikzstyle{phase dot}=[pdot,phase dimensions]
\tikzstyle{wphase dot}=[dot, phase dimensions]

\tikzstyle{hadamard}=[fill=white,draw,inner sep=0.6mm,font=\footnotesize,minimum height=6mm,minimum width=8mm]

\tikzstyle{anti} = [fill=white,draw,inner sep=0.6mm,font=\footnotesize,minimum height=3mm,minimum width=3mm]

\tikzstyle{triang}=[regular polygon,regular polygon sides=3,draw,scale=0.75,inner sep=-0.75pt,minimum width=9mm,fill=white,regular polygon rotate=180]
\tikzstyle{triang_lesssep}=[regular polygon,regular polygon sides=3,draw,scale=0.75,inner sep=-4pt,minimum width=9mm,fill=white,regular polygon rotate=180, text depth=4mm]
\tikzstyle{triangdag}=[regular polygon,regular polygon sides=3,draw,scale=0.75,inner sep=-0.5pt,minimum width=9mm,fill=white]

\makeatletter
\newcommand{\boxshape}[3]{%
\pgfdeclareshape{#1}{
\inheritsavedanchors[from=rectangle] %
\inheritanchorborder[from=rectangle]
\inheritanchor[from=rectangle]{center}
\inheritanchor[from=rectangle]{north}
\inheritanchor[from=rectangle]{south}
\inheritanchor[from=rectangle]{west}
\inheritanchor[from=rectangle]{east}
\backgroundpath{%
\southwest \pgf@xa=\pgf@x \pgf@ya=\pgf@y
\northeast \pgf@xb=\pgf@x \pgf@yb=\pgf@y

\@tempdima=#2
\@tempdimb=#3

\pgfpathmoveto{\pgfpoint{\pgf@xa - 5pt + \@tempdima}{\pgf@ya}}
\pgfpathlineto{\pgfpoint{\pgf@xa - 5pt - \@tempdima}{\pgf@yb}}
\pgfpathlineto{\pgfpoint{\pgf@xb + 5pt + \@tempdimb}{\pgf@yb}}
\pgfpathlineto{\pgfpoint{\pgf@xb + 5pt - \@tempdimb}{\pgf@ya}}
\pgfpathlineto{\pgfpoint{\pgf@xa - 5pt + \@tempdima}{\pgf@ya}}
\pgfpathclose
}
}}

\boxshape{NEbox}{0pt}{5pt}
\boxshape{SEbox}{0pt}{-5pt}
\boxshape{NWbox}{5pt}{0pt}
\boxshape{SWbox}{-5pt}{0pt}
\boxshape{EBox}{-3pt}{3pt}
\boxshape{WBox}{3pt}{-3pt}
\makeatother

\tikzstyle{map}=[draw,shape=NEbox,inner sep=2pt,minimum height=6mm,fill=white]
\tikzstyle{mapdag}=[draw,shape=SEbox,inner sep=2pt,minimum height=6mm,fill=white]
\tikzstyle{maptrans}=[draw,shape=SWbox,inner sep=2pt,minimum height=6mm,fill=white]
\tikzstyle{mapconj}=[draw,shape=NWbox,inner sep=2pt,minimum height=6mm,fill=white]

\tikzstyle{dmap}=[draw,doubled,shape=NEbox,inner sep=2pt,minimum height=6mm,fill=white]
\tikzstyle{dmapdag}=[draw,doubled,shape=SEbox,inner sep=2pt,minimum height=6mm,fill=white]
\tikzstyle{dmaptrans}=[draw,doubled,shape=SWbox,inner sep=2pt,minimum height=6mm,fill=white]
\tikzstyle{dmapconj}=[draw,doubled,shape=NWbox,inner sep=2pt,minimum height=6mm,fill=white]

\makeatletter
\pgfdeclareshape{cornerpoint}{
\inheritsavedanchors[from=rectangle] %
\inheritanchorborder[from=rectangle]
\inheritanchor[from=rectangle]{center}
\inheritanchor[from=rectangle]{north}
\inheritanchor[from=rectangle]{south}
\inheritanchor[from=rectangle]{west}
\inheritanchor[from=rectangle]{east}
\backgroundpath{%
\southwest \pgf@xa=\pgf@x \pgf@ya=\pgf@y
\northeast \pgf@xb=\pgf@x \pgf@yb=\pgf@y

\pgfmathsetmacro{\pgf@shorten@left}{\pgfkeysvalueof{/tikz/shorten left}}
\pgfmathsetmacro{\pgf@shorten@right}{\pgfkeysvalueof{/tikz/shorten right}}

\pgfpathmoveto{\pgfpoint{0.5 * (\pgf@xa + \pgf@xb)}{\pgf@ya - 5pt}}
\pgfpathlineto{\pgfpoint{\pgf@xa - 8pt + \pgf@shorten@left}{\pgf@yb - 1.5 * \pgf@shorten@left}}
\pgfpathlineto{\pgfpoint{\pgf@xa - 8pt + \pgf@shorten@left}{\pgf@yb}}
\pgfpathlineto{\pgfpoint{\pgf@xb + 8pt - \pgf@shorten@right}{\pgf@yb}}
\pgfpathlineto{\pgfpoint{\pgf@xb + 8pt - \pgf@shorten@right}{\pgf@yb - 1.5 * \pgf@shorten@right}}
\pgfpathclose
}
}

\pgfdeclareshape{cornercopoint}{
\inheritsavedanchors[from=rectangle] %
\inheritanchorborder[from=rectangle]
\inheritanchor[from=rectangle]{center}
\inheritanchor[from=rectangle]{north}
\inheritanchor[from=rectangle]{south}
\inheritanchor[from=rectangle]{west}
\inheritanchor[from=rectangle]{east}
\backgroundpath{%
\southwest \pgf@xa=\pgf@x \pgf@ya=\pgf@y
\northeast \pgf@xb=\pgf@x \pgf@yb=\pgf@y

\pgfmathsetmacro{\pgf@shorten@left}{\pgfkeysvalueof{/tikz/shorten left}}
\pgfmathsetmacro{\pgf@shorten@right}{\pgfkeysvalueof{/tikz/shorten right}}

\pgfpathmoveto{\pgfpoint{0.5 * (\pgf@xa + \pgf@xb)}{\pgf@yb + 5pt}}
\pgfpathlineto{\pgfpoint{\pgf@xa - 8pt + \pgf@shorten@left}{\pgf@ya + 1.5 * \pgf@shorten@left}}
\pgfpathlineto{\pgfpoint{\pgf@xa - 8pt + \pgf@shorten@left}{\pgf@ya}}
\pgfpathlineto{\pgfpoint{\pgf@xb + 8pt - \pgf@shorten@right}{\pgf@ya}}
\pgfpathlineto{\pgfpoint{\pgf@xb + 8pt - \pgf@shorten@right}{\pgf@ya + 1.5 * \pgf@shorten@right}}
\pgfpathclose
}
}

\makeatother

\pgfkeyssetvalue{/tikz/shorten left}{0pt}
\pgfkeyssetvalue{/tikz/shorten right}{0pt}

\tikzstyle{kpoint common}=[draw,fill=white,inner sep=1pt,minimum height=4mm]
\tikzstyle{kpoint}=[shape=cornerpoint,shorten left=5pt,kpoint common]
\tikzstyle{kpoint adjoint}=[shape=cornercopoint,shorten left=5pt,kpoint common]
\tikzstyle{kpoint conjugate}=[shape=cornerpoint,shorten right=5pt,kpoint common]
\tikzstyle{kpoint transpose}=[shape=cornercopoint,shorten right=5pt,kpoint common]

\tikzstyle{kpointdag}=[kpoint adjoint]
\tikzstyle{kpointadj}=[kpoint adjoint]
\tikzstyle{kpointconj}=[kpoint conjugate]
\tikzstyle{kpointtrans}=[kpoint transpose]

\tikzstyle{big kpoint}=[kpoint, minimum width=1.0 cm, minimum height=2mm, inner sep=4pt, text depth=1.5mm]

 \tikzstyle{upground}=[circuit ee IEC,thick,ground,rotate=90,scale=1.5]
 \tikzstyle{downground}=[circuit ee IEC,thick,ground,rotate=-90,scale=1.5]

%% file: styles.tikzstyles
\tikzstyle{discarding}=[fill=white, draw=black, shape=circle, style=upground]
\tikzstyle{smalldiscarding}=[fill=white, draw=black, style=upground, scale=0.5]
\tikzstyle{backdiscard}=[fill=white, draw=black, shape=circle, style=downground, scale=0.5]
\tikzstyle{smallbackdiscard}=[fill=white, draw=black, shape=circle, style=downground, scale=0.5]
\tikzstyle{state}=[fill=white, draw=black, style=triang, tikzit shape=rectangle]
\tikzstyle{kstate}=[fill=white, draw=black, style=kpoint, tikzit shape=rectangle]
\tikzstyle{kstateconj}=[fill=white, draw=black, style=kpoint conjugate, tikzit shape=rectangle]
\tikzstyle{kstateBIG}=[fill=white, draw=black, style=big kpoint, tikzit shape=rectangle]
\tikzstyle{effect}=[fill=white, draw=black, style=triangdag]
\tikzstyle{keffect}=[fill=white, draw=black, style=kpoint adjoint]
\tikzstyle{keffectconj}=[fill=white, draw=black, style=kpoint transpose]
\tikzstyle{morphdag}=[style=mapdag]
\tikzstyle{morph}=[style=hadamard]
\tikzstyle{WIDEmorph}=[style=hadamard, minimum width=14mm]
\tikzstyle{morphtrans}=[style=maptrans]
\tikzstyle{morphconj}=[style=mapconj]
\tikzstyle{CPMmorph}=[style=dmap]
\tikzstyle{CPMmorphconj}=[style=dmapconj]
\tikzstyle{CPMmorphdag}=[style=dmapdag]
\tikzstyle{CPMmorphtrans}=[style=dmaptrans]
\tikzstyle{CPMstate}=[fill=white, draw=black, style=triang, doubled]
\tikzstyle{CPMstateBIG}=[fill=white, draw=black, style={triang_lesssep}, doubled]
\tikzstyle{CPMkstate}=[fill=white, draw=black, style=kpoint, tikzit shape=rectangle, doubled]
\tikzstyle{CPMkstateconj}=[fill=white, draw=black, style=kpoint conjugate, tikzit shape=rectangle, doubled]
\tikzstyle{CPMkstateBIG}=[fill=white, draw=black, style=big kpoint, tikzit shape=rectangle, doubled]
\tikzstyle{CPMkeffect}=[fill=white, draw=black, style=kpoint adjoint, doubled]
\tikzstyle{CPMkeffectconj}=[fill=white, draw=black, style=kpoint transpose, doubled]
\tikzstyle{UHfB}=[fill=white, draw=black, style=triangdag, doubled, inner sep=-2pt]
\tikzstyle{leak}=[style=tinypoint, regular polygon rotate=-90]
\tikzstyle{leakfill}=[style=tinypoint, regular polygon rotate=-90, fill=black]
\tikzstyle{Z}=[style=dot, fill=green]
\tikzstyle{X}=[style=dot, fill=red]
\tikzstyle{black_dot}=[style=dot, fill=black]
\tikzstyle{white_dot}=[style=dot, fill=white]
\tikzstyle{qblack_dot}=[style=ddot, fill=black]
\tikzstyle{qwhite_dot}=[style=ddot, fill=white]
\tikzstyle{whitephase}=[style=wphase dot, fill=white]
\tikzstyle{qredphase}=[style=phase dot, fill=red]
\tikzstyle{qgreenphase}=[style=phase dot, fill=green]
\tikzstyle{had}=[style=hadamard, doubled]
\tikzstyle{box}=[style=hadamard]
\tikzstyle{classhad}=[style=hadamard]
\tikzstyle{antipode}=[style=anti]

\tikzstyle{dottededge}=[-, dotted]
\tikzstyle{double edge}=[-, style=doubled, draw=black, tikzit draw={rgb,255: red,191; green,0; blue,64}]
\tikzstyle{new edge style 0}=[<-]